\documentclass[11pt]{preprint}

\usepackage{difftrees} 
\usepackage[full]{textcomp}
\usepackage[osf]{newtxtext} 
\usepackage[cal=boondoxo]{mathalfa}
\usepackage{colortbl}

\usepackage[all,cmtip]{xy}
\usepackage{comment}

\usepackage{amssymb}
\usepackage{mathtools}
\usepackage{hyperref}
\usepackage{breakurl}
\usepackage{mhenvs}
\usepackage{mhequ} 
\usepackage{mhsymb}
\usepackage{booktabs}

\usepackage{tcolorbox}
\usepackage{mathrsfs}
\usepackage[utf8]{inputenc}
\usepackage{longtable}
\usepackage{wrapfig}
\usepackage{rotating} 
\usepackage{subcaption}
\usepackage{mathrsfs}
\usepackage{epsfig}
\usepackage{microtype}
\usepackage{comment}
\usepackage{wasysym}
\usepackage{centernot}
\usepackage{enumitem}
\usepackage{bm}
\usepackage{stackrel}

\usepackage{graphicx}
\usepackage{axodraw}
\usepackage{xspace}
\usepackage{subcaption} 
\usepackage{epsfig} 
\usepackage{axodraw} 
\usepackage{xspace} 
\usepackage[toc,page]{appendix} 
\usepackage[all,cmtip]{xy} 
\usepackage{relsize} 
\usepackage{shuffle} 
\makeatletter
\newcommand{\globalcolor}[1]{%
	\color{#1}\global\let\default@color\current@color
}
\makeatother

\usepackage{tikz}

\usepackage{pgf}
\usepgflibrary{fpu}

\usepackage{tikz-cd} 

\usetikzlibrary{cd} 
\usetikzlibrary{matrix,automata,arrows,positioning,calc,patterns,fit,calc,
	decorations.pathreplacing,decorations.pathmorphing,decorations.markings,matrix,
	tikzmark,shadows.blur}
\usetikzlibrary{calc}
\usetikzlibrary{decorations}
\usetikzlibrary{positioning}
\usetikzlibrary{shapes}
\usetikzlibrary{external}	
	
\input{diagrams_tikz.sty}

\newcommand{\TreeVertex}[2]{\draw[fill=black] (#1,#2) circle (0.04)}

\newcommand{\FGI}{\,\raisebox{-3pt}{\tikz{%
			\path[use as bounding box] (-0.04, -0.2) rectangle (0.69,0.2);
			\draw[semithick] (0,0) --(0.65,0);
			\TreeVertex{0}{0};
			\TreeVertex{0.65}{0};
			\drawbox;
		}\,}}

\newcommand{\FGIII}{\,\raisebox{-3pt}{\tikz{%
			\path[use as bounding box] (-0.04, -0.2) rectangle (0.69,0.2);
			\draw[semithick] (0,0)--(0.65,0);
			\draw[semithick] (0,0) edge [out=60,in=120] (0.65,0);
			\draw[semithick] (0,0) edge [out=-60,in=-120] (0.65,0);
			\TreeVertex{0}{0};
			\TreeVertex{0.65}{0};
			\drawbox;
		}\,}}

\newcommand{\FGVI}{\,\raisebox{-5.5pt}{\tikz{%
			\path[use as bounding box] (-0.04, -0.1) rectangle (0.64,0.56);
			\draw[semithick] (0,0) edge [out=20,in=160] (0.6,0);
			\draw[semithick] (0,0) edge [out=-40,in=-140] (0.6,0);
			\draw[semithick] (0,0) edge [out=100,in=-160] (0.3,0.52);
			\draw[semithick] (0,0) edge [out=40,in=-100] (0.3,0.52);
			\draw[semithick] (0.3,0.52) edge [out=-20,in=80] (0.6,0);
			\draw[semithick] (0.3,0.52) edge [out=-80,in=140] (0.6,0);
			\TreeVertex{0}{0};
			\TreeVertex{0.6}{0};
			\TreeVertex{0.3}{0.52};
			\drawbox;
		}\,}}

\newcommand{\FGLoop}{\,\raisebox{-4pt}{\tikz{%
			\path[use as bounding box] (-0.2, -0.04) rectangle (0.2,0.5);
			\draw[semithick] (0,0) edge [loop,distance=25pt,out=135,in=45] (0,0);
			\TreeVertex{0}{0};
			\drawbox;
		}\,}}

\newcommand{\FGIIIplus}{\,\raisebox{-5.5pt}{\tikz{%
			\path[use as bounding box] (-0.04, -0.2) rectangle (0.69,0.56);
			\draw[semithick] (0,0)--(0.65,0);
			\draw[semithick] (0,0) edge [out=60,in=120] (0.65,0);
			\draw[semithick] (0,0) edge [out=-60,in=-120] (0.65,0);
			\draw[semithick] (0,0) edge [out=80,in=-160] (0.325,0.5);
			\draw[semithick] (0.65,0) edge [out=100,in=-20] (0.325,0.5);
			\TreeVertex{0}{0};
			\TreeVertex{0.65}{0};
			\TreeVertex{0.325}{0.5};
			\drawbox;
		}\,}}

\newcommand{\FGIItwice}{\,\raisebox{-3pt}{\tikz{%
			\path[use as bounding box] (-0.04, -0.2) rectangle (1.34,0.2);
			\draw[semithick] (0,0) edge [out=60,in=120] (0.65,0);
			\draw[semithick] (0,0) edge [out=-60,in=-120] (0.65,0);
			\draw[semithick] (0.65,0) edge [out=60,in=120] (1.3,0);
			\draw[semithick] (0.65,0) edge [out=-60,in=-120] (1.3,0);
			\TreeVertex{0}{0};
			\TreeVertex{0.65}{0};
			\TreeVertex{1.3}{0};
			\drawbox;
		}\,}}

\newcommand{\Xpairing}{\,\raisebox{-3pt}{\tikz{%
			\path[use as bounding box] (-0.4, -0.2) rectangle (1.05,0.2);
			\draw[semithick] (0,0) edge [out=60,in=120] (0.65,0);
			\draw[semithick] (0,0) edge [out=-60,in=-120] (0.65,0);
			\draw[semithick] (0,0)--(0.65,0);
			\draw[semithick] (0,0)--(-0.35,0);
			\draw[semithick] (0,0)--(1,0);
			\TreeVertex{0}{0};
			\TreeVertex{0.65}{0};
			\drawbox;
		}\,}}

\newcommand{\YII}{\,\raisebox{-5.5pt}{\tikz{%
			\path[use as bounding box] (-0.01, -0.1) rectangle (0.56,0.56);
			\draw[semithick] (0.2,0) edge [out=160,in=-160] (0.2,0.5);
			\draw[semithick] (0.2,0) edge [out=10,in=-15] (0.2,0.5);
			\TreeVertex{0.2}{0};
			\TreeVertex{0.2}{0.5};
			\drawbox;
		}\,}}

\newcommand{\YIIi}{\,\raisebox{-5.5pt}{\tikz{%
			\path[use as bounding box] (-0.01, -0.06) rectangle (0.7,1);
			\draw[semithick] (0.3,0) edge [out=160,in=-160] (0.3,0.5) node[left, yshift=2mm] {\tiny{$e_1$}};
			\draw[semithick] (0.3,0) edge [out=0,in=-20] (0.3,0.5) node[right, xshift =0.5mm, yshift=2mm] {\tiny{$e_2$}};
			\draw [fill=black]  (0.3, 0.5)  circle(0.04) node[above] {\tiny{$v_1$}};
			\draw [fill=black]  (0.3, 0)  circle(0.04) node[below] {\tiny{$v_2$}};
			\drawbox;
		}\,}}

\newcommand{\FGIIIi}{\,\raisebox{-3pt}{\tikz{%
			\path[use as bounding box] (-0.3, -0.2) rectangle (1,0.2);
			\draw[semithick] (0,0)--(0.65,0);
			\draw[semithick] (0,0) edge [out=60,in=120] (0.65,0);
			\draw[semithick] (0,0) edge [out=-60,in=-120] (0.65,0);
			\TreeVertex{0}{0} node[left] {\tiny{$u_1$}};
			\TreeVertex{0.65}{0} node[right] {\tiny{$u_2$}};
			\drawbox;
		}\,}}

\newcommand{\YIIt}{\,\raisebox{-5.5pt}{\tikz{%
			\path[use as bounding box] (-0.2, 0) rectangle (1,0.7);
			\draw[semithick] (0.15,0.1) edge [out=140,in=-170] (0.3,0.5) node[left, yshift=2mm] {\tiny{$\ell_1$}};
			\draw[semithick] (0.45,0.1) edge [out= 30,in=0] (0.3,0.5) 
			node[right, xshift =0.5mm, yshift=2mm] {\tiny{$\ell_2$}};
			\draw [fill=black]  (0.3, 0.5)  circle(0.04) node[above] {\tiny{$v_1$}};
			
			\drawbox;
		}\,}}

\newcommand{\FGIIIplusL}{\,\raisebox{-5.5pt}{\tikz{%
			\path[use as bounding box] (-0.4, -0.2) rectangle (1,0.56);
			\draw[semithick] (0,0)--(0.65,0);
			\draw[semithick] (0,0) edge [out=60,in=120] (0.65,0);
			\draw[semithick] (0,0) edge [out=-60,in=-120] (0.65,0);
			\draw[semithick] (0,0) edge [out=80,in=-160] (0.325,0.5) node[left, xshift = 2mm, yshift=4mm] {\tiny{$e'_1$}};
			\draw[semithick] (0.65,0) edge [out=100,in=-20] (0.325,0.5) node[right, xshift = -2mm, yshift=4mm] {\tiny{$e'_2$}};
			\TreeVertex{0}{0} node[left] {\tiny{$u_1$}};
			\TreeVertex{0.65}{0} node[right] {\tiny{$u_2$}};
			\TreeVertex{0.325}{0.5} node[above] {\tiny{$v_1$}};
			\drawbox;
		}\,}}

\newcommand{\FGIIIplusR}{\,\raisebox{-5.5pt}{\tikz{%
			\path[use as bounding box] (-0.4, -0.2) rectangle (1,0.56);
			\draw[semithick] (0,0)--(0.65,0);
			\draw[semithick] (0,0) edge [out=60,in=120] (0.65,0);
			\draw[semithick] (0,0) edge [out=-60,in=-120] (0.65,0);
			\draw[semithick] (0,0) edge [out=80,in=-160] (0.325,0.5) node[left, xshift = 2mm, yshift=4mm] {\tiny{$e''_1$}};
			\draw[semithick] (0.65,0) edge [out=100,in=-20] (0.325,0.5) node[right, xshift = -2mm, yshift=4mm] {\tiny{$e''_2$}};
			\TreeVertex{0}{0} node[left] {\tiny{$u_2$}};
			\TreeVertex{0.65}{0} node[right] {\tiny{$u_1$}};
			\TreeVertex{0.325}{0.5} node[above] {\tiny{$v_1$}};
			\drawbox;
		}\,}}

\newcommand{\ISOL}{\,\raisebox{-5.5pt}{\tikz{%
			\path[use as bounding box] (-0.4, -0.2) rectangle (1,0.56);
			\draw[semithick] (0,0)--(0.65,0);
			\draw[semithick] (0,0) edge [out=40,in=-100] (0.325,0.5);
			\draw[semithick] (0,0) edge [out=-60,in=-120] (0.65,0);
			\draw[semithick] (0,0) edge [out=80,in=-160] (0.325,0.5) ;
			\draw[semithick] (0.65,0) edge [out=100,in=-20] (0.325,0.5) ;
			\TreeVertex{0}{0} node[left] {\tiny{$v$}};
			\TreeVertex{0.65}{0} node[right] {\tiny{$u_2$}};
			\TreeVertex{0.325}{0.5} node[above] {\tiny{$u_1$}};
			\drawbox;
		}\,}}

\newcommand{\ISOR}{\,\raisebox{-5.5pt}{\tikz{%
			\path[use as bounding box] (-0.4, -0.2) rectangle (1,0.56);
			\draw[semithick] (0,0)--(0.65,0);
			\draw[semithick] (0.65,0) edge [out=-80,in=140] (0.325,0.5);
			\draw[semithick] (0,0) edge [out=-60,in=-120] (0.65,0);
			\draw[semithick] (0,0) edge [out=80,in=-160] (0.325,0.5) ;
			\draw[semithick] (0.65,0) edge [out=100,in=-20] (0.325,0.5) ;
			\TreeVertex{0}{0} node[left] {\tiny{$u_1$}};
			\TreeVertex{0.65}{0} node[right] {\tiny{$v$}};
			\TreeVertex{0.325}{0.5} node[above] {\tiny{$u_2$}};
			\drawbox;
		}\,}}

\newcommand{\FGIIIa}{\,\raisebox{-3pt}{\tikz[scale=2.5]{%
			\path[use as bounding box] (-0.15, -0.1) rectangle (0.69,0.2);
			\draw[semithick, blue] (0,0)--(0.325,0) 
			node[left, xshift = -1mm, yshift=1mm] {\tiny{$\ell_{b_1}$}};
			\draw[semithick, blue] (0.325,0)--(0.65,0) 
			node[left, xshift = -1mm, yshift=1mm] {\tiny{$\ell_{b_2}$}};
			\draw[semithick,red] (0,0) edge [out=60,in=180] (0.325,0.16) 
			node[above, xshift = 4mm, yshift=2.5mm] {\tiny{$\ell_{r_1}$}};
			\draw[semithick,red] (0.325,0.16) edge [out=0,in=120] (0.65,0) 
			node[above, xshift = 4mm, yshift=-1.5mm] {\tiny{$\ell_{r_2}$}};
			\draw[semithick,darkgreen] (0,0) edge [out=-60,in=180] (0.325,-0.16) 
			node[below, xshift = 4mm, yshift=0.5mm] {\tiny{$\ell_{g_1}$}};
			\draw[semithick,darkgreen] (0.325,-0.16) edge [out=0,in=-120] (0.65,0) 
			node[above, xshift = 4mm, yshift=-0.7mm] {\tiny{$\ell_{g_2}$}};
			\TreeVertex{0}{0}  node[left] {\tiny{$v_1$}};
			\TreeVertex{0.65}{0}  node[right] {\tiny{$v_2$}};
			\drawbox;
		}\,}}

\newcommand{\FGIIIb}{\,\raisebox{-3pt}{\tikz[scale=2.5]{%
			\path[use as bounding box] (-0.15, -0.1) rectangle (0.69,0.2);
			\draw[semithick,darkgreen] (0,0)--(0.325,0) 
			node[left, xshift = -1mm, yshift=1mm] {\tiny{$\ell_{g_2}$}};
			\draw[semithick, darkgreen] (0.325,0)--(0.65,0) 
			node[left, xshift = -1mm, yshift=1mm] {\tiny{$\ell_{g_1}$}};
			\draw[semithick,blue] (0,0) edge [out=60,in=180] (0.325,0.16) 
			node[above, xshift = 4mm, yshift=2.5mm] {\tiny{$\ell_{b_2}$}};
			\draw[semithick,blue] (0.325,0.16) edge [out=0,in=120] (0.65,0) 
			node[above, xshift = 4mm, yshift=-1.5mm] {\tiny{$\ell_{b_1}$}};
			\draw[semithick,red] (0,0) edge [out=-60,in=180] (0.325,-0.16) 
			node[below, xshift = 4mm, yshift=0.5mm] {\tiny{$\ell_{r_2}$}};
			\draw[semithick,red] (0.325,-0.16) edge [out=0,in=-120] (0.65,0) 
			node[above, xshift = 4mm, yshift=-0.7mm] {\tiny{$\ell_{r_1}$}};
			\TreeVertex{0}{0}  node[left] {\tiny{$v_2$}};
			\TreeVertex{0.65}{0}  node[right] {\tiny{$v_1$}};
			\drawbox;
		}\,}}

\newcommand{\FGIIIc}{\,\raisebox{-3pt}{\tikz[scale=2.5]{%
			\path[use as bounding box] (-0.15, -0.1) rectangle (0.69,0.2);
			\draw[semithick, blue] (0,0)--(0.325,0) 
			node[left, xshift = -1mm, yshift=1mm] {\tiny{$\ell_{b_2}$}};
			\draw[semithick, blue] (0.325,0)--(0.65,0) 
			node[left, xshift = -1mm, yshift=1mm] {\tiny{$\ell_{b_1}$}};
			\draw[semithick,red] (0,0) edge [out=60,in=180] (0.325,0.16) 
			node[above, xshift = 4mm, yshift=2.5mm] {\tiny{$\ell_{r_2}$}};
			\draw[semithick,red] (0.325,0.16) edge [out=0,in=120] (0.65,0) 
			node[above, xshift = 4mm, yshift=-1.5mm] {\tiny{$\ell_{r_1}$}};
			\draw[semithick,darkgreen] (0,0) edge [out=-60,in=180] (0.325,-0.16) 
			node[below, xshift = 4mm, yshift=0.5mm] {\tiny{$\ell_{g_2}$}};
			\draw[semithick,darkgreen] (0.325,-0.16) edge [out=0,in=-120] (0.65,0) 
			node[above, xshift = 4mm, yshift=-0.7mm] {\tiny{$\ell_{g_1}$}};
			\TreeVertex{0}{0}  node[left] {\tiny{$v_2$}};
			\TreeVertex{0.65}{0}  node[right] {\tiny{$v_1$}};
			\drawbox;
		}\,}}

\newcommand{\FGIIId}{\,\raisebox{-3pt}{\tikz[scale=2.5]{%
			\path[use as bounding box] (-0.15, -0.1) rectangle (0.69,0.2);
			\draw[semithick, blue] (0,0)--(0.325,0) 
			node[left, xshift = -1mm, yshift=1mm] {\tiny{$\ell_{b_1}$}};
			\draw[semithick, blue] (0.325,0)--(0.65,0) 
			node[left, xshift = -1mm, yshift=1mm] {\tiny{$\ell_{b_2}$}};
			\draw[semithick,darkgreen] (0,0) edge [out=60,in=180] (0.325,0.16) 
			node[above, xshift = 4mm, yshift=2.5mm] {\tiny{$\ell_{g_1}$}};
			\draw[semithick,red] (0.325,0.16) edge [out=0,in=120] (0.65,0) 
			node[above, xshift = 4mm, yshift=-1.5mm] {\tiny{$\ell_{r_2}$}};
			\draw[semithick,red] (0,0) edge [out=-60,in=180] (0.325,-0.16) 
			node[below, xshift = 4mm, yshift=0.5mm] {\tiny{$\ell_{r_1}$}};
			\draw[semithick,darkgreen] (0.325,-0.16) edge [out=0,in=-120] (0.65,0) 
			node[above, xshift = 4mm, yshift=-0.7mm] {\tiny{$\ell_{g_2}$}};
			\TreeVertex{0}{0}  node[left] {\tiny{$v_1$}};
			\TreeVertex{0.65}{0}  node[right] {\tiny{$v_2$}};
			\drawbox;
		}\,}}

\newcommand{\ZIV}{\,\raisebox{-3pt}{\tikz[scale=0.015]{
		\path[use as bounding box] (270, 265) rectangle (190,330);
		\draw[semithick]    (203.17,273.41) -- (247.82,273.14) ;
		\draw[semithick]    (203.17,273.41) .. controls (222.02,252.45) and (235.43,258.98) .. (247.82,273.14) ;
		\draw[semithick]    (203.17,273.41) .. controls (224.55,297.09) and (243.02,279.13) .. (247.82,273.14) ;
		\draw[semithick]    (203.17,307.98) -- (247.82,307.71) ;
		\draw[semithick]    (203.17,307.98) .. controls (222.02,287.02) and (235.43,293.56) .. (247.82,307.71) ;
		\draw[semithick]    (203.17,307.98) .. controls (224.55,331.67) and (243.02,313.7) .. (247.82,307.71) ;
		\draw [semithick]   (203.17,273.41) .. controls (204.31,273.14) and (179.26,288.38) .. (203.17,307.98) ;
		\draw [semithick]    (247.82,273.14) .. controls (248.96,272.86) and (274.64,289.2) .. (247.82,307.71) ;
		\draw[semithick]   (201.62,273.41) .. controls (201.62,272.49) and (202.31,271.74) .. (203.17,271.74) .. controls (204.03,271.74) and (204.72,272.49) .. (204.72,273.41) .. controls (204.72,274.33) and (204.03,275.08) .. (203.17,275.08) .. controls (202.31,275.08) and (201.62,274.33) .. (201.62,273.41) -- cycle ;
\draw  [fill={rgb, 255:red, 0; green, 0; blue, 0 }  ,fill opacity=1 ] (200.58,273.41) .. controls (200.58,271.87) and (201.74,270.62) .. (203.17,270.62) .. controls (204.6,270.62) and (205.76,271.87) .. (205.76,273.41) .. controls (205.76,274.95) and (204.6,276.2) .. (203.17,276.2) .. controls (201.74,276.2) and (200.58,274.95) .. (200.58,273.41) -- cycle ;
\draw  [fill={rgb, 255:red, 0; green, 0; blue, 0 }  ,fill opacity=1 ] (244.19,307.27) .. controls (244.19,305.73) and (245.35,304.48) .. (246.79,304.48) .. controls (248.22,304.48) and (249.38,305.73) .. (249.38,307.27) .. controls (249.38,308.81) and (248.22,310.06) .. (246.79,310.06) .. controls (245.35,310.06) and (244.19,308.81) .. (244.19,307.27) -- cycle ;
\draw  [fill={rgb, 255:red, 0; green, 0; blue, 0 }  ,fill opacity=1 ] (200.58,307.98) .. controls (200.58,306.44) and (201.74,305.19) .. (203.17,305.19) .. controls (204.6,305.19) and (205.76,306.44) .. (205.76,307.98) .. controls (205.76,309.53) and (204.6,310.78) .. (203.17,310.78) .. controls (201.74,310.78) and (200.58,309.53) .. (200.58,307.98) -- cycle ;
\draw  [fill={rgb, 255:red, 0; green, 0; blue, 0 }  ,fill opacity=1 ] (245.45,272.03) .. controls (245.45,270.48) and (246.61,269.24) .. (248.05,269.24) .. controls (249.48,269.24) and (250.64,270.48) .. (250.64,272.03) .. controls (250.64,273.57) and (249.48,274.82) .. (248.05,274.82) .. controls (246.61,274.82) and (245.45,273.57) .. (245.45,272.03) -- cycle ;
			\drawbox;
		}\,}}

	\newcommand{\ZIVI}{\,\raisebox{-3pt}{\tikz[scale=0.015]{%
	\path[use as bounding box] (-10, 220) rectangle (80,290);
		\draw  [fill={rgb, 255:red, 0; green, 0; blue, 0 }  ,fill opacity=1 ] (37.77,220.68) .. controls (37.77,219.14) and (38.93,217.89) .. (40.36,217.89) .. controls (41.8,217.89) and (42.96,219.14) .. (42.96,220.68) .. controls (42.96,222.22) and (41.8,223.47) .. (40.36,223.47) .. controls (38.93,223.47) and (37.77,222.22) .. (37.77,220.68) -- cycle ;
		\draw  [fill={rgb, 255:red, 0; green, 0; blue, 0 }  ,fill opacity=1 ] (37.59,280.68) .. controls (37.59,279.14) and (38.75,277.89) .. (40.18,277.89) .. controls (41.61,277.89) and (42.77,279.14) .. (42.77,280.68) .. controls (42.77,282.22) and (41.61,283.47) .. (40.18,283.47) .. controls (38.75,283.47) and (37.59,282.22) .. (37.59,280.68) -- cycle ;
		\draw  [fill={rgb, 255:red, 0; green, 0; blue, 0 }  ,fill opacity=1 ] (6.4,249.24) .. controls (6.4,247.7) and (7.57,246.45) .. (9,246.45) .. controls (10.43,246.45) and (11.59,247.7) .. (11.59,249.24) .. controls (11.59,250.78) and (10.43,252.03) .. (9,252.03) .. controls (7.57,252.03) and (6.4,250.78) .. (6.4,249.24) -- cycle ;
		\draw  [fill={rgb, 255:red, 0; green, 0; blue, 0 }  ,fill opacity=1 ] (67.39,250.4) .. controls (67.39,248.86) and (68.55,247.61) .. (69.98,247.61) .. controls (71.42,247.61) and (72.58,248.86) .. (72.58,250.4) .. controls (72.58,251.94) and (71.42,253.19) .. (69.98,253.19) .. controls (68.55,253.19) and (67.39,251.94) .. (67.39,250.4) -- cycle ;
		\draw[semithick]    (8.91,250.32) .. controls (6.34,234.03) and (22.4,220.69) .. (40.36,220.68) ;
		\draw[semithick]    (40.87,280.41) .. controls (42.91,264.04) and (56.3,252.74) .. (69.98,250.4) ;
		\draw [semithick]    (8.91,250.32) .. controls (38.37,238.68) and (36.42,228.68) .. (40.67,221.81) ;
		\draw[semithick]    (40.18,280.68) .. controls (30.25,282.38) and (6.82,270.86) .. (8.91,250.32) ;
		\draw [semithick]   (69.96,249.7) .. controls (59.7,252.27) and (43.78,234.83) .. (40.36,220.68) ;
		\draw [semithick]   (40.96,221.09) .. controls (60.45,220.82) and (72.64,235.73) .. (69.96,249.7) ;
		\draw [semithick]   (9.61,249.54) .. controls (28.21,254.12) and (35.52,261.43) .. (40.18,280.68) ;
		\draw [semithick]   (39.27,281.05) .. controls (74.45,277) and (69.73,255.36) .. (70.33,251.25) ;
		
				\drawbox;
			}\,}}

\newcommand{\ZIVII}{\,\raisebox{-3pt}{\tikz[x=0.65pt,y=0.65pt,yscale=-1,xscale=1]{%
	\path[use as bounding box] (-10, 260) rectangle (60,300);
			\draw [semithick]    (12,267) .. controls (22.2,261) and (32.6,260.6) .. (42.6,267.8) ;
			\draw [semithick]    (12,267) .. controls (17.4,272.2) and (27.4,278.6) .. (42.6,267.8) ;
			\draw[semithick]   (26.73,283.69) .. controls (18.33,285.29) and (16.2,299.02) .. (25.8,302.62) ;
			\draw [semithick]    (26.07,284.09) .. controls (38.33,291.29) and (28.07,306.22) .. (24.07,302.62) ;
			\draw [semithick]    (12.4,266.33) .. controls (14.2,275.93) and (14.61,281.72) .. (26.07,284.09) ;
			\draw [semithick]    (40.14,269.15) .. controls (44.54,264.35) and (39.13,286.09) .. (26.73,283.69) ;
			\draw [semithick]    (12,267) .. controls (3.44,275.67) and (0.87,300.09) .. (24.73,303.82) ;
			\draw [semithick]    (42.6,267.8) .. controls (49.44,271) and (55.13,303.02) .. (26.6,303.69) ;
			\draw  [fill={rgb, 255:red, 0; green, 0; blue, 0 }  ,fill opacity=1 ] (24.12,302.62) .. controls (24.12,301.63) and (24.87,300.82) .. (25.8,300.82) .. controls (26.73,300.82) and (27.48,301.63) .. (27.48,302.62) .. controls (27.48,303.62) and (26.73,304.42) .. (25.8,304.42) .. controls (24.87,304.42) and (24.12,303.62) .. (24.12,302.62) -- cycle ;
			\draw  [fill={rgb, 255:red, 0; green, 0; blue, 0 }  ,fill opacity=1 ] (24.39,284.09) .. controls (24.39,283.09) and (25.14,282.28) .. (26.07,282.28) .. controls (26.99,282.28) and (27.74,283.09) .. (27.74,284.09) .. controls (27.74,285.08) and (26.99,285.89) .. (26.07,285.89) .. controls (25.14,285.89) and (24.39,285.08) .. (24.39,284.09) -- cycle ;
			\draw  [fill={rgb, 255:red, 0; green, 0; blue, 0 }  ,fill opacity=1 ] (10.72,268.14) .. controls (10.72,267.14) and (11.47,266.33) .. (12.4,266.33) .. controls (13.33,266.33) and (14.08,267.14) .. (14.08,268.14) .. controls (14.08,269.13) and (13.33,269.94) .. (12.4,269.94) .. controls (11.47,269.94) and (10.72,269.13) .. (10.72,268.14) -- cycle ;
			\draw  [fill={rgb, 255:red, 0; green, 0; blue, 0 }  ,fill opacity=1 ] (40.12,267.8) .. controls (40.12,266.8) and (40.87,266) .. (41.8,266) .. controls (42.73,266) and (43.48,266.8) .. (43.48,267.8) .. controls (43.48,268.8) and (42.73,269.6) .. (41.8,269.6) .. controls (40.87,269.6) and (40.12,268.8) .. (40.12,267.8) -- cycle ;
			\drawbox;
		}\,}}
		
\def\M{\mathfrak{M}}

\def\emptyset{{\centernot\ocircle}}

\def\CP{\mathcal{P}}

\def\CG{\mathcal{G}}
\def\CJ{\mathcal{J}}
\def\CA{\mathcal{A}}
\def\CE{\mathcal{E}}

\def\CM{\mathcal{M}}

\def\id{\mathrm{id}}

\def\proj{\mathbf{p}}

\definecolor{blush}{rgb}{0.87, 0.36, 0.51}
	\definecolor{brightcerulean}{rgb}{0.11, 0.67, 0.84}
	\definecolor{greenryb}{rgb}{0.4, 0.69, 0.2}

\newif\ifdark
\darkfalse

\ifdark
\definecolor{darkred}{rgb}{0.9,0.2,0.2}
\definecolor{darkblue}{rgb}{0.7,0.3,1}
\definecolor{darkgreen}{rgb}{0.1,0.9,0.1}
\definecolor{franck}{rgb}{0,0.8,1}
\definecolor{pagebackground}{rgb}{.15,.21,.18}
\definecolor{pageforeground}{rgb}{.84,.84,.85}
\pagecolor{pagebackground}
\AtBeginDocument{\globalcolor{pageforeground}}
\definecolor{symbols}{rgb}{0,0.7,1}
\colorlet{connection}{red!80!black}
\colorlet{boxcolor}{blue!50}

\else

\definecolor{darkred}{rgb}{0.7,0.1,0.1}
\definecolor{darkblue}{rgb}{0.4,0.1,0.8}
\definecolor{darkgreen}{rgb}{0.1,0.7,0.1}
\definecolor{franck}{rgb}{0,0,1}
\definecolor{pagebackground}{rgb}{1,1,1}
\definecolor{pageforeground}{rgb}{0,0,0}
\colorlet{symbols}{blue!90!black}
\colorlet{connection}{red!30!black}
\colorlet{boxcolor}{blue!50!black}

\fi

\def\slash{\leavevmode\unskip\kern0.18em/\penalty\exhyphenpenalty\kern0.18em}
\def\dash{\leavevmode\unskip\kern0.18em--\penalty\exhyphenpenalty\kern0.18em}

\DeclareMathAlphabet{\mathbbm}{U}{bbm}{m}{n}

\DeclareFontFamily{U}{BOONDOX-calo}{\skewchar\font=45 }
\DeclareFontShape{U}{BOONDOX-calo}{m}{n}{
  <-> s*[1.05] BOONDOX-r-calo}{}
\DeclareFontShape{U}{BOONDOX-calo}{b}{n}{
  <-> s*[1.05] BOONDOX-b-calo}{}
\DeclareMathAlphabet{\mcb}{U}{BOONDOX-calo}{m}{n}
\SetMathAlphabet{\mcb}{bold}{U}{BOONDOX-calo}{b}{n}

\setlist{noitemsep,topsep=4pt,leftmargin=1.5em}

\DeclareMathAlphabet{\mathbbm}{U}{bbm}{m}{n}

\DeclareMathAlphabet{\mcb}{U}{BOONDOX-calo}{m}{n}
\SetMathAlphabet{\mcb}{bold}{U}{BOONDOX-calo}{b}{n}
\DeclareFontFamily{U}{mathx}{\hyphenchar\font45}
\DeclareFontShape{U}{mathx}{m}{n}{
      <5> <6> <7> <8> <9> <10>
      <10.95> <12> <14.4> <17.28> <20.74> <24.88>
      mathx10
      }{}
\DeclareSymbolFont{mathx}{U}{mathx}{m}{n}
\DeclareMathSymbol{\bigtimes}{1}{mathx}{"91}

\providecommand{\figures}{false}
{ \ifthenelse{\equal{\figures}{false}} {#1}{\[ {\rm Figure \ missing !} \]} }{}

\tikzstyle{tinydots}=[dash pattern=on \pgflinewidth off \pgflinewidth]
\tikzstyle{superdense}=[dash pattern=on 4pt off 1pt]

\newcommand{\beq}{\begin{equation}}
\newcommand{\eeq}{\end{equation}}

\usepackage{empheq}

\def\${|\!|\!|}

\def\proj{\mathfrak{p}}

\newcounter{theorem}

\newtheorem{ex}[theorem]{Example}

\newenvironment{DIFnomarkup}{}{} 

\theorembodyfont{\rmfamily}

\newcommand{\rrightarrow}{{\to\hskip -4.9mm\raise 1pt\hbox{$\to$}}}

\newfont{\indic}{bbmss12}

\def\Nabla_#1{\nabla_{\!#1}}

\makeatletter
\pgfdeclareshape{crosscircle}
{
  \inheritsavedanchors[from=circle] 
  \inheritanchorborder[from=circle]
  \inheritanchor[from=circle]{north}
  \inheritanchor[from=circle]{north west}
  \inheritanchor[from=circle]{north east}
  \inheritanchor[from=circle]{center}
  \inheritanchor[from=circle]{west}
  \inheritanchor[from=circle]{east}
  \inheritanchor[from=circle]{mid}
  \inheritanchor[from=circle]{mid west}
  \inheritanchor[from=circle]{mid east}
  \inheritanchor[from=circle]{base}
  \inheritanchor[from=circle]{base west}
  \inheritanchor[from=circle]{base east}
  \inheritanchor[from=circle]{south}
  \inheritanchor[from=circle]{south west}
  \inheritanchor[from=circle]{south east}
  \inheritbackgroundpath[from=circle]
  \foregroundpath{
    \centerpoint%
    \pgf@xc=\pgf@x%
    \pgf@yc=\pgf@y%
    \pgfutil@tempdima=\radius%
    \pgfmathsetlength{\pgf@xb}{\pgfkeysvalueof{/pgf/outer xsep}}%
    \pgfmathsetlength{\pgf@yb}{\pgfkeysvalueof{/pgf/outer ysep}}%
    \ifdim\pgf@xb<\pgf@yb%
      \advance\pgfutil@tempdima by-\pgf@yb%
    \else%
      \advance\pgfutil@tempdima by-\pgf@xb%
    \fi%
    \pgfpathmoveto{\pgfpointadd{\pgfqpoint{\pgf@xc}{\pgf@yc}}{\pgfqpoint{-0.707107\pgfutil@tempdima}{-0.707107\pgfutil@tempdima}}}
    \pgfpathlineto{\pgfpointadd{\pgfqpoint{\pgf@xc}{\pgf@yc}}{\pgfqpoint{0.707107\pgfutil@tempdima}{0.707107\pgfutil@tempdima}}}
    \pgfpathmoveto{\pgfpointadd{\pgfqpoint{\pgf@xc}{\pgf@yc}}{\pgfqpoint{-0.707107\pgfutil@tempdima}{0.707107\pgfutil@tempdima}}}
    \pgfpathlineto{\pgfpointadd{\pgfqpoint{\pgf@xc}{\pgf@yc}}{\pgfqpoint{0.707107\pgfutil@tempdima}{-0.707107\pgfutil@tempdima}}}
  }
}
\makeatother

\def\symbol#1{\textcolor{symbols}{#1}}

\def\decorate#1#2{
        \ifnum#2>0
    		\foreach \count in {1,...,#2}{
	       	let
				\p1 = (sourcenode.center),
                \p2 = (sourcenode.east),
				\n1 = {\x2-\x1},
				\n2 = {1mm},
				\n3 = {(1.3+0.6*(\count-1))*\n1},
				\n4 = {0.7*\n1}
			in 
        		node[rectangle,fill=symbols,rotate=30,inner sep=0pt,minimum width=0.2*\n2,minimum height=\n2] at ($(sourcenode.center) + (\n3,\n4)$) {}
				}
		\fi
        \ifnum#1>0
    		\foreach \count in {1,...,#1}{
	       	let
				\p1 = (sourcenode.center),
                \p2 = (sourcenode.east),
				\n1 = {\x2-\x1},
				\n2 = {1mm},
				\n3 = {(1.3+0.6*(\count-1))*\n1},
				\n4 = {0.7*\n1}
			in 
        		node[rectangle,fill=symbols,rotate=-30,inner sep=0pt,minimum width=0.2*\n2,minimum height=\n2] at ($(sourcenode.center) + (-\n3,\n4)$) {}
				}
		\fi
}

\tikzset{
    dectriangle/.style 2 args={
        triangle,
        alias=sourcenode,
        append after command={\decorate{#1}{#2}}
    },
    dectriangle/.default={0}{0},
}

\tikzset{
	cross/.style={path picture={ 
  		\draw[symbols]
			(path picture bounding box.south east) -- (path picture bounding box.north west) (path picture bounding box.south west) -- (path picture bounding box.north east);
		}},
root/.style={circle,fill=green!50!black,inner sep=0pt, minimum size=1.2mm},
        dot/.style={circle,fill=pageforeground,inner sep=0pt, minimum size=1mm},
        dotred/.style={circle,fill=pageforeground!50!pagebackground,inner sep=0pt, minimum size=2mm},
        var/.style={circle,fill=pageforeground!10!pagebackground,draw=pageforeground,inner sep=0pt, minimum size=3mm},
        kernel/.style={semithick,shorten >=2pt,shorten <=2pt},
        kernels/.style={snake=zigzag,shorten >=2pt,shorten <=2pt,segment amplitude=1pt,segment length=4pt,line before snake=2pt,line after snake=5pt,},
        rho/.style={densely dashed,semithick,shorten >=2pt,shorten <=2pt},
           testfcn/.style={dotted,semithick,shorten >=2pt,shorten <=2pt},
        renorm/.style={shape=circle,fill=pagebackground,inner sep=1pt},
        labl/.style={shape=rectangle,fill=pagebackground,inner sep=1pt},
        xic/.style={very thin,circle,draw=symbols,fill=symbols,inner sep=0pt,minimum size=1.2mm},
        g/.style={very thin,rectangle,draw=symbols,fill=symbols!10!pagebackground,inner sep=0pt,minimum width=2.5mm,minimum height=1.2mm},
        xi/.style={very thin,circle,draw=symbols,fill=symbols!10!pagebackground,inner sep=0pt,minimum size=1.2mm},
	xies/.style={very thin,rectangle,fill=green!50!black!25,draw=symbols,inner sep=0pt,minimum size=1.1mm},
	xiesf/.style={very thin,rectangle,fill=green!50!black,draw=symbols,inner sep=0pt,minimum size=1.1mm},
        xix/.style={very thin,crosscircle,fill=symbols!10!pagebackground,draw=symbols,inner sep=0pt,minimum size=1.2mm},
        X/.style={very thin,cross,rectangle,fill=pagebackground,draw=symbols,inner sep=0pt,minimum size=1.2mm},
	xib/.style={thin,circle,fill=symbols!10!pagebackground,draw=symbols,inner sep=0pt,minimum size=1.6mm},
	xie/.style={thin,circle,fill=green!50!black,draw=symbols,inner sep=0pt,minimum size=1.6mm},
	xid/.style={thin,circle,fill=symbols,draw=symbols,inner sep=0pt,minimum size=1.6mm},
	xibx/.style={thin,crosscircle,fill=symbols!10!pagebackground,draw=symbols,inner sep=0pt,minimum size=1.6mm},
	kernels2/.style={very thick,draw=connection,segment length=12pt},
	keps/.style={thin,draw=symbols,->},
	kepspr/.style={thick,draw=connection,->},
	krho/.style={thin,draw=symbols,superdense,->},
	krhopr/.style={thick,draw=connection,superdense},
	triangle/.style = { regular polygon, regular polygon sides=3},
	not/.style={thin,circle,draw=connection,fill=connection,inner sep=0pt,minimum size=0.5mm},
	diff/.style = {very thin,draw=symbols,triangle,fill=red!50!black,inner sep=0pt,minimum size=1.6mm},
	diff1/.style = {very thin,dectriangle={1}{0},fill=red!50!black,draw=symbols,inner sep=0pt,minimum size=1.6mm},
	diff2/.style = {very thin,dectriangle={1}{1},fill=red!50!black,draw=symbols,inner sep=0pt,minimum size=1.6mm},
		diffmini/.style = {very thin,rectangle,fill=black,draw=black,inner sep=0pt,minimum size=0.75mm},
	 kernelsmod/.style={very thick,draw=connection,segment length=12pt},
	 rec/.style = {very thin,rectangle,fill=black,draw=black,inner sep=0pt,minimum size=2mm},
	cerc/.style={very thin,circle,draw=black,fill=symbols,inner sep=0pt,minimum size=2mm},
	stars/.style={very thin,star,star points=6,star point ratio=0.5, draw=black,fill=red,inner sep=0pt,minimum size=0.7mm},
	>=stealth,
        }
        \tikzset{
root/.style={circle,fill=black!50,inner sep=0pt, minimum size=3mm},
        circ/.style={circle,fill=white,draw=black,very thin,inner sep=.5pt, minimum size=1.2mm},
        round1/.style={fill=white,outer sep = 0,inner sep=2pt,rounded corners=1mm,draw,text=black,thin,minimum size=1.2mm},
          circ1/.style={circle,fill=red!10,draw=red,very thin,inner sep=.5pt, minimum size=1.2mm},
        rect/.style={fill=white,outer sep = 0,inner sep=2pt,rectangle,draw,text=black,thin,minimum size=1.2mm},
        rect1/.style={fill=white,outer sep = 0,inner sep=2pt,rectangle,draw,text=black,thin,minimum size=1.2mm},
        round2/.style={fill=red!10,outer sep = 0,inner sep=2pt,rounded corners=1mm,draw,text=black,thin,minimum size=1.2mm},
       round3/.style={fill=blue!10,outer sep = 0,inner sep=2pt,rounded corners=1mm,draw,text=black,thin,minimum size=1.2mm}, 
        rect2/.style={fill=black!10,outer sep = 0,inner sep=2pt,rectangle,draw,text=black,thin,minimum size=1.2mm},
        dot/.style={circle,fill=black,inner sep=0pt, minimum size=1.2mm},
        dotred/.style={circle,fill=black!50,inner sep=0pt, minimum size=2mm},
        var/.style={circle,fill=black!10,draw=black,inner sep=0pt, minimum size=3mm},
        kernel/.style={semithick,shorten >=2pt,shorten <=2pt},
         diag/.style={thin,shorten >=4pt,shorten <=4pt},
        kernel1/.style={thick},
        kernels/.style={snake=zigzag,shorten >=2pt,shorten <=2pt,segment amplitude=1pt,segment length=4pt,line before snake=2pt,line after snake=5pt,},
		kernels1/.style={snake=zigzag,segment amplitude=0.5pt,segment length=2pt},
		rho1/.style={densely dotted,semithick},
        rho/.style={densely dashed,semithick,shorten >=2pt,shorten <=2pt},
           testfcn/.style={dotted,semithick,shorten >=2pt,shorten <=2pt},
           visible/.style={draw, circle, fill, inner sep=0.25ex},
        renorm/.style={shape=circle,fill=white,inner sep=1pt},
        labl/.style={shape=rectangle,fill=white,inner sep=1pt},
        xic/.style={very thin,circle,fill=symbols,draw=black,inner sep=0pt,minimum size=1.2mm},
        xi/.style={very thin,circle,fill=blue!10,draw=black,inner sep=0pt,minimum size=1.2mm},
	xib/.style={very thin,circle,fill=blue!10,draw=black,inner sep=0pt,minimum size=1.6mm},
	xie/.style={very thin,circle,fill=green!50!black,draw=black,inner sep=0pt,minimum size=1mm},
	xid/.style={very thin,circle,fill=symbols,draw=black,inner sep=0pt,minimum size=1.6mm},
	edgetype/.style={very thin,circle,draw=black,inner sep=0pt,minimum size=5mm},
	nodetype/.style={very thick,circle,draw=black,inner sep=0pt,minimum size=5mm},
	kernels2/.style={very thick,draw=connection,segment length=12pt},
clean/.style={thin,circle,fill=black,inner sep=0pt,minimum size=1mm},	not/.style={thin,circle,fill=symbols,draw=connection,fill=connection,inner sep=0pt,minimum size=0.8mm},
	>=stealth,
        }

 \def\1{\mathbf{\symbol{1}}}

\DeclareMathAlphabet{\mathpzc}{OT1}{pzc}{m}{it}

\def\eqref#1{(\ref{#1})}

\makeatletter 
\newcommand*{\bigcdot}{}
\DeclareRobustCommand*{\bigcdot}{%
  \mathbin{\mathpalette\bigcdot@{}}%
}
\newcommand*{\bigcdot@scalefactor}{.5}
\newcommand*{\bigcdot@widthfactor}{1.15}
\newcommand*{\bigcdot@}[2]{%
  \sbox0{$#1\vcenter{}$}
  \sbox2{$#1\cdot\m@th$}%
  \hbox to \bigcdot@widthfactor\wd2{%
    \hfil
    \raise\ht0\hbox{%
      \scalebox{\bigcdot@scalefactor}{%
        \lower\ht0\hbox{$#1\bullet\m@th$}%
      }%
    }%
    \hfil
  }%
}
\makeatother

\def\BPHZ{\textnormal{\tiny \textsc{bphz}}}
\def\F{\textnormal{\tiny \textsc{f}}}
\def\M{\textnormal{\tiny \textsc{m}}}

\tcbset
{colframe=boxcolor,colback=symbols!7!pagebackground,coltext=pageforeground,
fonttitle=\bfseries,nobeforeafter,center title,size=fbox,boxsep=1.5pt,
top=0mm,bottom=0mm,boxsep=0mm,tcbox raise base}

\def\two{{\<generic>\kern0.05em\<genericb>}}
\def\twoI{{\<Ito>\kern0.05em\<Itob>}}

\def\st{\mathsf{fgt}}
\def\mail#1{\burlalt{#1}{mailto:#1}}

\usepackage{thmtools} 

\declaretheorem[style=definition]{example}

\begin{document}

\def\st{\mathsf{fgt}}
\def\mail#1{\burlalt{#1}{mailto:#1}}
\title{Renormalising Feynman diagrams with multi-indices}
\author{Yvain Bruned, Yingtong Hou,}
\institute{ 
	IECL (UMR 7502), Université de Lorraine
	\\
	Email:\ \begin{minipage}[t]{\linewidth}
		\mail{yvain.bruned@univ-lorraine.fr}
		\\
		\mail{yingtong.hou@univ-lorraine.fr}
\end{minipage}}

\def\dsqcup{\sqcup\mathchoice{\mkern-7mu}{\mkern-7mu}{\mkern-3.2mu}{\mkern-3.8mu}\sqcup}

\maketitle

\begin{abstract}
\ \ \ \ In this work, we provide a method to obtain the renormalised measure in quantum field theory directly from the renormalisation of the expansion of the original measure.
Our approach is based on BPHZ renormalisation via multi-indices, a combinatorial structure extremely successful for describing scalar-valued singular SPDEs. 
We propose the multi-indices counterpart to the Hopf algebraic program initiated by Connes and Kreimer for the renormalisation of  Feynman diagrams. This new Hopf algebra also bridges the gap between the analysis of ``pre-Feynman diagrams" and traditional diagrammatic methods.
The construction relies on a well-chosen extraction-contraction coproduct of multi-indices equipped with a correct symmetry factor. We illustrate our method by the $ \Phi^4 $ measure example. 
%
\end{abstract}

\setcounter{tocdepth}{1}
\tableofcontents

\section{ Introduction }
\subsection{Some history on the algebraic side of BPHZ renormalisation}\label{sec:history}
The BPHZ renormalisation \cite{BP, Hepp, Zim}, named after Bogoliubov, Parasiuk, Hepp, and Zimmermann, provides a systematic framework for renormalising Feynman diagrams, often expressed combinatorially via forest formulae.
In the late 1990s, Connes and Kreimer introduced a novel algebraic perspective on this renormalisation procedure, based on Hopf algebras constructed from Trees and Feynman diagrams.
The Hopf algebra of Feynman diagrams is defined via an extraction–contraction coproduct, which systematically identifies and extracts divergent subgraphs \cite{CK2}. 
In parallel, a Hopf algebra structure on trees, equipped with a coproduct performing admissible cuts, encodes the hierarchy of nested subdivergences \cite{CK1}.
Notably, this latter structure also appears in the field of numerical analysis \cite{Butcher72}, where it models the composition of numerical schemes.
Another Hopf algebra on trees with extraction-contraction coproduct plays a significant role in the context of substitution for B-series see \cite{CHV05,CHV07} and see \cite{CEM} for the co-interaction between admissible cuts and extraction-contraction coproducts. 
Moreover, this Hopf algebra approach is one of the cornerstones of the theory of Regularity Structures, introduced by Martin Hairer in \cite{reg} to rigorously treat singular Stochastic Partial Differential Equations (SPDEs). 
 In \cite{BHZ}, two co-interacted Hopf algebras on decorated trees are derived for recentering and for renormalising iterated stochastic integrals, respectively, which extend those from quantum field theory (QFT) and numerical analysis. 
The convergence of these renormalised iterated integrals has been proved in \cite{CH16} and, together with \cite{BCCH}, one solves a large class of singular SPDEs via the theory of Regularity Structures. 
 
Connes and Kreimer \cite{CK2} proved that, for a measure in QFT, the renormalised value and counterterms satisfy the Birkhoff decomposition. However, their work relies on the dimensional regularisation technique.
Subsequently, Hairer translates the algebraic results of  \cite{BHZ} to Feynman diagrams in  \cite{Hairer18}. He reformulated Connes and Kreimer's theory to a general context where dimensional regularisation is not available. 
In his formulation, the BPHZ counterterms are defined through a ``BPHZ character" \eqref{eq:BPHZ} constructed from a ``twisted antipode" \eqref{eq:antipode1} of a Hopf algebra on Feynman diagrams. This Hopf algebra is  equipped with a Connes-Kreimer type extraction-contraction coproduct. Furthermore, the author shows a group structure on the BPHZ character, analogous to the characters group in \cite{CK2}.
The same type of Hopf algebra has been recently understood via a deformation in  \cite{BM22} and a post-Lie product in \cite{BYK23} in the context of singular SPDEs. Similar Hopf algebra also appears in \cite{BS}, where it is used to analyse the local error of resonance-based schemes for dispersive PDEs.

\subsection{Results and the motivation of multi-indices formulation} \label{subsec:intro}
One considers the following Lagrangian
\begin{equs}
\mathcal{L} = \CL_0-\CL_\alpha	 = -\left(\frac{1}{2}\Vert \nabla \phi(x)\Vert^2 +\frac{1}{2}\phi(x)^2 + \CL_{\alpha}\right)
\end{equs}
where $\phi$ is a classic field and $x$ take values in the space $\Lambda$ (for example $d$-dimensional torus). Usually the perturbation $\CL_\alpha$ is a polynomial of $\phi$ and its derivatives (the same setting as in \cite{CK2}). 
Here, for the simplicity in the notation, we choose the perturbation as the polynomial of $\phi$
\begin{equs}
	\CL_\alpha = \sum_{k \in \mathbb{N}} \alpha_k \phi(x)^k.
\end{equs}
Nevertheless, results in this paper can be applied to the general case involving derivatives of $\phi$ by adding some ``decorations" to the combinatorial object ``multi-indices".
Denote $Z(0)$ as the partition function 
 of the measure with Lagrangian $\CL_0$. In other words,
 \begin{equs}
 	\mu_0 = \frac{1}{Z(0)}e^{\int_{\Lambda}\CL_0(x) dx }d\phi = \frac{1}{Z(0)} e^ {-\int_{\Lambda}\frac{1}{2}\Vert \nabla \phi(x)\Vert^2 +\frac{1}{2}\phi(x)^2 dx }d\phi.
 \end{equs}
Then, the partition function of the measure involving perturbation is
 \begin{equs}
 	Z(\alpha) = Z(0) \mathbb{E}\left[e^{- \int_\Lambda \sum_{k \in \mathbb{N}} \alpha_k \phi(x)^k dx }\right].
 \end{equs}
Our study is based on the renormalisation of
\begin{equs}\label{eq:unrenormalised _cumulant}
	\mathbb{E}\left[e^{- \int_\Lambda \sum_{k \in \mathbb{N}} \alpha_k \phi(x)^k dx }\right]
\end{equs}
 and it is implemented as below. 
\begin{itemize}
	\item Replace the ill-posed products of $\phi$ by its Wick product \cite{WB08}. The $k$-th Wick power can be represented as Hermite polynomials $H_k$.
	\item Cumulant expand the  moment-generating function $\mathbb{E}\left[e^{- \int_\Lambda \sum_{k \in \mathbb{N}} \alpha_k H_k(X(x)) dx }\right]$.
	 By the Linked Cluster Theorem \cite{connect1,connect2,connect3} (see also \cite[Section 3.3]{connect4} and \cite[Propisition 3.1]{BK23} for the explanation in the Hopf algebra setting), the cumulant expansion are projected to the valuation \eqref{eq:valuation} of connected Feynman diagrams.
	\item Apply the BPHZ method to the Feynman diagrams and get the renormalised cumulant series.
\end{itemize} 

However, the challenge is to determine the renormalised measure whose cumulant expansion of the partition function matches the renormalised cumulant series of the original measure. In other words, one seeks values of $\gamma_k$ such that
\begin{equs}
\mathbb{E}\left[\exp\left(-\int_\Lambda\sum_{k \in \mathbb{N}} (\alpha_k+\gamma_k) H_{k}(X(x)) dx\right) \right]
\end{equs}
has the cumulant equals the BPHZ renormalised cumulant of \eqref{eq:unrenormalised _cumulant}.
The main result of this paper is establishing a self-contained and systematic approach that can help find the renormalised measure through the renormalised cumulant expansion. This approach is equivalent to the classical BPHZ renormalisation procedure via Feynman diagrams.

To explain why multi-indices are the suitable combinatorial objects in solving the problem discussed above, we first examine the nature of the combinatorial factors appearing in the cumulant expansion. 
 The Linked Cluster Theorem mainly states that the moment-generating function of $Z(\alpha) / Z(0)$ can be obtained by projecting moments in the power series to the valuation $\Pi_\F$ \eqref{eq:valuation} of connected Feynman diagrams, i.e.,
	\begin{equs} \label{cumulant}
		&\log	\mathbb{E}\left[\exp\left(-\int_\Lambda\sum_{k \in \mathbb{N}} \alpha_k H_{k}(X(x)) dx\right)\right] 
		\\=& 
		\sum_{n=2}^\infty \frac{1}{n!}\proj \left(\mathbb{E}\left[
		\left(\sum_{k \in \mathbb{N}} -\alpha_k \int_{\Lambda}H_{k}(X(x)) dx\right)^n
		\right]\right)
	\end{equs}
	where $\proj$ is the projection and the connected Feynman digrams are formed by regarding each $H_k$ as a vertex with $k$ half-edges and by then pairing all half-edges from different vertices in all possible ways. For example
	 \begin{equs}
		\proj \left(\mathbb{E}\left[
		\left( -\alpha_3 \int_{\Lambda}H_{3}(X(x)) dx\right)^2
		\right]\right) = (-\alpha_3)^2	3! \Pi_\F \left(\FGIII\right)
	\end{equs}
	where $3!$ is the number of diagrams isomorphic to $\FGIII$ which can be formed by pairing two vertices each having $3$ half-edges.
	The disconnected diagram $\FGIII \FGIII$ is not allowed to appear in $\proj \left(\mathbb{E}\left[
	\left( -\alpha_3 \int_{\Lambda}H_{3}(X(x)) dx\right)^4
	\right]\right) $.
	Then, by the multinomial theorem, we can rewrite \eqref{cumulant} as 
	\begin{equs}  \label{eq:cumulat_diagram}
	 \sum_{\Gamma \in \mathbf{F}}\frac{\Upsilon^\alpha_\F (\Gamma)N(\Gamma)}{\hat{S}_\F(\Gamma)} \Pi_\F(\Gamma)
	\end{equs}
	where $\mathbf{F}$ denotes the set of connected Feynman diagrams and $\Pi_\F$ is a suitable valuation map on this space (formula in \eqref{eq:valuation}). 
	The coefficients
	\begin{equs} 
		\Upsilon^\alpha_\F (\Gamma) = \prod_{k \in \mathbb{N}} (-\alpha_k)^{\beta(\Gamma, k)}, \quad \quad
		\hat{S}_\F(\Gamma) = \prod_{k \in \mathbb{N}} \beta(\Gamma, k)!
	\end{equs}
	where $\beta(\Gamma, k)$ is the number of vertices connected with $k$ edges in the Feynman diagram $\Gamma$.
	By interpreting each vertex in a Feynman diagram with $k$ edges as having $k$ half-edges, the entire diagram can be understood as arising from pairing of all half-edges. $N(\Gamma)$ is the number of possible coupling that lead to the Feynman diagrams isomorphic to $\Gamma$.
\begin{remark}
	Lemma \ref{prop:cumulant_F} will prove that \eqref{eq:cumulat_diagram} is the diagrammatic representation of the moment-generating function rigorously.
\end{remark}

Then, the problem concerning finding the renormalised measure from renormalised cumulant expansion can be illustrated by the following diagram.
{\scriptsize
	\begin{equs}
		\begin{tikzcd}
			&\mathbb{E}\left[\exp\left(-\int_\Lambda\sum_{k \in \mathbb{N}} \alpha_k H_{k}(X(x)) dx\right)\right] \arrow[r, "\text{\shortstack{cumulant\\expansion}}"] \arrow[ddd, swap, "\text{{\Large ?}}"] 
			& \sum_{\Gamma \in \mathbf{F}}\frac{\Upsilon^\alpha_\F (\Gamma)N(\Gamma)}{\hat{S}_\F(\Gamma)} \Pi_\F(\Gamma) 
			\arrow[r, "\text{BPHZ}"]  
			& \sum_{\Gamma \in \mathbf{F}}\frac{\Upsilon^\alpha_\F (\Gamma)N(\Gamma)}{\hat{S}_\F(\Gamma)} \Pi_\F(\hat{ M}_{\F}\Gamma)
			\\&&&
			\\&&&
			\\ & \mathbb{E}\left[\exp\left(-\int_\Lambda\sum_{k \in \mathbb{N}} (\alpha_k+\gamma_k) H_{k}(X(x)) dx\right) 
			\arrow[rruuu, swap,"\text{cumulant expansion}"] \right]
			&
			&
		\end{tikzcd}
	\end{equs}
}
where $\hat{M}_\F$ is the BPHZ procedure in \cite{Hairer18} and its formula is in \eqref{eq:BPHZ_martin}.

To find the value of $\gamma$, one has to match the cumulant of renormalised measure
\begin{equs}
	\mathbb{E}\left[\exp\left(-\int_\Lambda\sum_{k \in \mathbb{N}} (\alpha_k+\gamma_k) H_{k}(X(x)) dx\right) \right]
	=
	\sum_{\Gamma' \in \mathbf{F}}\frac{\Upsilon^{\alpha+\gamma}_\F (\Gamma')N(\Gamma')}{\hat{S}_\F(\Gamma')} \Pi_\F(\Gamma')
\end{equs}
with the renormalised cumulant. However, it is not a easy task as the structure of diagrams can be very involved in the renormalisation procedure $\hat{M}_\F$ and the cumulant expansion. Indeed, one has to solve complicated equations of $\gamma_k$ determined by $\beta(\Gamma', k)$ and many combinatorial structures. Since the cumulant series has infinite terms indexed by diagrams, it is essential to ask if the solution $\gamma_k$ exists and if it is unique. 

Nevertheless, one important observation gives the hint: The multi-index $\beta(\cdot,k)$ suffices to describe all the combinatorial factors $\Upsilon^\alpha_\F (\cdot)$, $N (\cdot)$, $\hat{S}_{\F}(\cdot)$ and hence the cumulant expansion \eqref{eq:cumulat_diagram} as well as the renormalised cumulant. Moreover, cumulant expanding the renormalised measure amounts to pairing the half-edges connected to vertices representing the Hermite polynomials. Therefore, the BPHZ renormalisation procedure $\hat{M}_\M$ through a Hopf algebra of those vertices offers hope to avoid the combinatorial factors associated with the geometric configuration of Feynman diagrams.
While these vertices can be paired via their half-edges to form Feynman diagrams, $\hat{M}_\M$ itself does not rely on the diagrammatic structure. 
Consequently, we can circumvent the intermediate equations of $\gamma_k$ and derive $\gamma_k$ directly from the renormalisation procedure $\hat{M}_\M$ (see Theorem \ref{prop:exponential}).

In the literature, there are works such as \cite{pre_FD} studying the aforementioned sets of vertices, referred to as pre-Feynman diagrams. However, no Hopf algebra has been formulated on these structures.
Recently, the first attempt at vertex-based Hopf algebraic renormalisation appeared in \cite{BK23}, where the authors studied the $\Phi^4$ measure in QFT. 
 They regard the vertices  with $4$ half-edges representing the $4$-th order Wick product as a variable $X$ and those with $2$ half-edges as the variable $Y$. Then they described the Hopf algebra by the monomials of $X$ and $Y$.
Nevertheless, their approach could not be extended to other models, as their derivation of the coproduct (a key component of the renormalisation procedure) relies on the simplicity of divergent subdiagrams, allowing the coproduct to be obtained by comparison with the extraction-contraction coproduct of Feynman diagrams.
From their work, one observes that the pre-Feynman diagrams can be described
combinatorially using ``multi-indices", a tool introduced in \cite{OSSW, LOT, LO23, LOTT} for scalar singular SPDEs, and later applied in rough path theory \cite{L23} and numerical analysis \cite{BEH}.
Specifically, a vertex with $k$ half-edges is identified with an abstract variable $z_k$,
and the pre-Feynman diagrams correspond to monomials
\begin{equs}
	z^{\beta} : = \prod_{k \in \mathbb{N}} z_k^{\beta(k)}
\end{equs}
defined in \eqref{def:multi_indices} where $\beta(k) \in \mathbb{N}$ is the frequency of the vertex with $k$ half-edges in a Feynman diagram. 
Although the renormalisation of SPDEs using multi-indices has been studied in \cite{OSSW, BL23},  their methods are based on the postulation of counterterms of SPDEs and thus cannot be adopted to our problem. Moreover, there is no explicit formula for the multi-indices extraction-contraction coproduct. The explicit formula was found in \cite{BH24} and later another formulation was built in the third arXiv version of \cite{GMZ24}. 
For notational simplicity, we restrict ourselves in this paper to a renormalisation without higher-order Taylor series terms. The construction in the general case is quite similar as it will be based on a coproduct introduced in \cite{BH24} on extended multi-indices for singular SPDEs.

Another result of this paper is that the Hopf algebra via multi-indices constructed in QFT renormalisation fills the algebraic gap between  pre-Feynman diagrams and  Feynman diagrams. The link between these two combinatorial objects will be illustrated by some morphisms. Moreover, this Hopf algebra helps compare renormalisation in QFT with similar procedures in related fields (introduced in Section \ref{sec:history}) such as SPDEs and numerical analysis via rooted trees and multi-indices. Additionally, since the the transformation between multi-indices and Feynman diagrams is essentially a ``pairing half-edges" problem, the algebraic results in this paper could apply to study the same type of problems provided that suitable valuation maps are available to model the problems.
For example, it seems that the map $\CP$ (Definition \ref{def:map_P}) can also be used to study symmetries in the configuration models (see \cite{config1, config2, config3}) in the field of random graphs.

\subsection{Overview of the methodology}
The main idea of the paper is to establish the multi-indices-based BPHZ renormalisation and to show that it is equivalent to the diagram-based one. The summary of the procedure is shown in the first line of the graph below.
{\scriptsize
		\begin{equs}
			\begin{tikzcd}
				&\mathbb{E}\left[\exp\left(-\int_\Lambda\sum_{k \in \mathbb{N}} \alpha_k H_{k}(X(x)) dx\right)\right] \arrow[r, "\text{\shortstack{cumulant\\expansion}}"] \arrow[ddd, swap, "\text{Theorem \ref{prop:exponential}}"] 
				& \sum_{z^\beta \in \mathbf{M}}\frac{\Upsilon^\alpha_\M (z^\beta)}{\hat{S}_\M(z^\beta)} \Pi_\M(z^\beta) \arrow[r, "\text{ BPHZ}"]  
				& \sum_{z^\beta \in \mathbf{M}}\frac{\Upsilon^\alpha_\M (z^\beta)}{\hat{S}_\M(z^\beta)} \Pi_\M(\hat{ M}_{\M}z^\beta)
				\\&&&
				\\&&&
				\\ &  \mathbb{E}\left[\exp\left(-\int_\Lambda\sum_{k \in \mathbb{N}} (\alpha_k+\gamma_k) H_{k}(X(x)) dx\right) 
				\arrow[rruuu, swap,"\text{cumulant expansion}"] \right]
				&
				&
			\end{tikzcd}
		\end{equs}}
where $\mathbf{M}$ is the set of multi-indices and other combinatorial factors are the counterparts to those of Feynman diagrams. We will show that the cumulant expansion represented by multi-indices equals the one  represented by Feynman diagrams, and that the two renormalisations are equivalent.  Finally, Theorem \ref{prop:exponential} provides the method of obtaining the renormalised measure from the renormalised cumulant.

	Let us summarise the main strategies.
	We firstly define a ``counting map" $\Phi$ \eqref{eq:Phi_map} linking the multi-indices and Feynman diagrams by counting the frequencies of vertices in Feynman diagrams
	\begin{equs}
		\Phi(\Gamma) := \prod_{ v \in \CV(\Gamma)} z_{k(v)}
	\end{equs}
	where $k(v)$ is the number of edges attached to the vertex $v$ and $\CV$ is the vertex set of the Feynman diagram $\Gamma$. As an illustration, consider the following example
		\begin{equs}
			\Phi(\FGIIIplus) = z_2z_4^2
		\end{equs}
		as $\FGIIIplus$ has one vertex with arity $2$ and two vertices with arity $4$.
	The equivalence between two representations of cumulant expansion and the equivalence between two renormalisations rely on the fact that the map $\CP$, the adjoint of the map $\Phi$ under suitable inner product,
	lifts multi-indices to a sum of all possible Feynman diagrams by pairing of half-edges represented by $z^\beta$ . Indeed, the map $\CP$ gives the coefficient $N(\cdot)$ in \ref{eq:cumulat_diagram} when one forms the Feynman diagram by pairing half edges to represent the cumulant, which
	 will be shown in Proposition \ref{prop:Sta_Orb}. 
	 This leads to the equivalence between the valuation of multi-indices and that of Feynman diagrams. Specifically, $\Pi_\M(z^\beta) =  \Pi_\F \circ \CP (z^\beta)$ (see Corollary \ref{coro:commute}), which shows that the equivalence in valuations actually comes from the ``pairing half-edges" problem.
	
	Another key element of showing the equivalence between the renormalisation via multi-indices and diagram-based BPHZ renormalisation (Theorem \ref{them:main}) is to find a simultaneous insertion product of Feynman diagrams $\star_\F$ which has the morphism property with the simultaneous insertion product of multi-indices $\star_\M$ 
\begin{equs} \label{morphism_star}
	\Phi(	\tilde{\prod}_{\F \, i=1}^n \Gamma_i \star_{\F} \bar\Gamma ) = \Phi(\tilde{\prod}_{i=1}^n \Gamma_i) \star_\M \Phi(\bar\Gamma)
\end{equs}
and satisfies the following adjoint relation
\begin{equs} \label{adjoint}
	\langle \Delta_\F \Gamma , \tilde{\prod}_{\F \, i=1}^n \Gamma_i  \otimes \bar\Gamma \rangle = 	\langle \Gamma, \tilde{\prod}_{\F \, i=1}^n \Gamma_i  \star_\F \bar\Gamma \rangle .
\end{equs}
Here $\tilde{\prod}_{\F \, i=1}^n \Gamma_i $ is a disjoint union of connected Feynman diagrams, $\Phi$ \eqref{eq:Phi_map} is the map translating Feynman diagrams to their unpaired vertex set -- multi-indices, and the inner product is defined in \eqref{inner_product_diagram}.
$\Delta_{\F}$ is the reduced extraction-contraction coproduct central to the renormalisation $\hat{ M}_\F$.
The suitable $\star_\F$ (Definition \ref{def:insertion_FD}) is obtained from the Guin-Oudom \cite{Guin1, Guin2} construction for the insertion of Feynman diagrams which is a pre-Lie product and is a modification of the one in the  literature \cite{K2006, VS2007} to satisfy both \eqref{morphism_star} and \eqref{adjoint}.
The morphism property \eqref{morphism_star} will be proved through the Proposition \ref{prop:morphism}  while the adjoint relation \eqref{adjoint} will be justified in Theorem \ref{them: adjoint_Feyman}.

Let us outline the paper by summarising the sections. In Section \ref{Sec::2}, we start by recalling the definition of Feynman diagrams and how they are involved in
the BPHZ renormalisation, where the reduced extraction-contraction coproduct $\Delta_\F$ \eqref{eq:EC_F} and its adjoint operator, the simultaneous insertion $\star_{\F}$ (Def \ref{def:insertion_FD}), play important roles.

In Section \ref{Sec::3}, we start with motivating multi-indices by introducing the``pairing half-edges" problem which leads to the correspondence between multi-indices and Feynman diagrams. This relation will be described through map $\Phi$ \eqref{eq:Phi_map}
and its adjoint $\CP$ (Definition \ref{def:map_P}). We define the symmetry factor of multi-indices \eqref{eq:SF_multiindices} as the cardinal of the isomorphism group of their corresponding Feynman diagrams, a key quantity in proving of all equivalences in this paper. Finally, we derive the explicit formulae of simultaneous insertion product of multi-indices and its adjoint reduced extraction-contraction coproduct.

The Theorem \ref{them:main} describing the equivalence between two renormalisations, $\hat{M}_\F$ and $\hat{M}_\M$, is proved in Section \ref{sec:BPHZ}.
The commutative diagram \eqref{main_diagram} illustrates the equivalence. The proof is essentially achieved by the adjoint and dual relations between products and coproducts and morphism properties of the map $\Phi$. 
As the symmetry factors are used as the underlying pairing of the inner products, it is necessary to discuss the choice of symmetry factor of Feynman diagrams such that \eqref{adjoint} holds (Theorem \ref{them: adjoint_Feyman}) and such that $\CP$ lifts multi-indices to a sum of all possible Feynman diagrams by pairing of half-edges.
We finish this section with main Theorem~\ref{prop:exponential} which provides a new Hopf algebraic understanding of the renormalisation of measures in QFT. The combinatorial description used (multi-indices) reflects well the products in the Lagrangian of the measures.

Finally, in Section \ref{sec:Example} we will revisit the example of renormalising $\Phi^4_3$ measure studied in \cite{BK23} and show that the combinatorial objects used there can be formally described as multi-indices and the map $\CP$ used there is a special case of the map $\CP$ defined in this paper. Some examples of Theorem \ref{them:main2} will also be given. $\Phi^4_3$  is used here for the illustration. However, our results can be applied to more general Lagrangians as described in Section \ref{subsec:intro}.

\subsection*{Acknowledgements}

{\small
	Y. B. gratefully acknowledges funding support from the European Research Council (ERC) through the ERC Starting Grant Low Regularity Dynamics via Decorated Trees (LoRDeT), grant agreement No.\ 101075208.
	Views and opinions expressed are however those of the author(s) only and do not necessarily reflect those of the European Union or the European Research Council. Neither the European Union nor the granting authority can be held responsible for them.
} 

%

\section{Feynman diagrams and BPHZ renormalisation} \label{Sec::2}
Feynman diagrams are commonly used combinatorial tools in renormalising SPDEs and corresponding invariant measures. The BPHZ renormalisation amounts to extracting divergent subgraphs from a Feynman graph and recursively applying a twisted antipode to the extracted parts. In this section, we will briefly review Feynman diagrams and the BPHZ renormalisation. Finally, we will conclude this section by introducing the extraction-contraction coproduct of Feynman diagrams and its adjoint simultaneous insertion product, which are the core of the Hopf algebra appearing in the BPHZ renormalisation.  
\subsection{Feynman diagrams}
%
%
In this paper, a Feynman diagram is a connected graph containing vertices connected by edges. Some vertices may have legs (external edges). 
A diagram with no legs is called a vacuum diagram. We start with the case without any decorations on vertices or edges. Feynman diagrams are used to represent the integral of a product of some kernel $K: \Lambda \rightarrow \mathbb{R}^d$ and some test functions $\psi$. Precisely, we have the following valuation map.
For a Feynman diagram $\Gamma = (\CV, \CE, \CL)$, where $\CV$ is the set of vertices in $\Gamma$, $\CE$ is its edge set, and $\CL$ is its leg set, its valuation is
\begin{equs}
	\Pi_{\F}(\Gamma) = \int_{\Lambda^{|\CV|}} \prod_{e \in \CE} K(x_{e_{+}}- x_{e_{-}}) \psi(x_{l_1}, ..., x_{l_{|\CL|}}) dx
\end{equs}
in which  $e_{+}$ and  $e_{-}$ are two vertices connected by the edge $e$. $x_{l_i}$ is the variable encoded by the vertex where the leg $l_i \in \CL$ is attached to. If the kernel $ K $ is symmetric i.e., $K(-x) = K(x)$ then the orientation on edges is irrelevant, which means that
it does not matter which vertex is chosen to be ${e_{+}}$ and which to be $e_{-}$. If one has different kernels, the decorations of edges can be used to distinguish them (details can be found in \cite{Hairer18}). 

To be concise, in this paper, we consider the simplest case that we only have one symmetric kernel $K$ and there is no leg (vacuum diagram) but our results can be generalised to more complex settings by putting decorations or orientations. In other words, we consider the integrals
\begin{equs} \label{eq:valuation}
	\Pi_{\F}(\Gamma) = \int_{\Lambda^{\CV}} \prod_{e \in \CE} K(x_{e_{+}}- x_{e_{-}}) dx.
\end{equs}
Furthermore, we do not allow half-edges (or sometimes called free legs), and self-connected vertices in a Feynman diagram are forbidden, which means each edge must be attached to two vertices like $\FGI$.
For example, $\Xpairing $ and $\FGLoop$ are not considered as Feynman diagrams in this paper. Note that here ``free legs" (half-edges) are different from ``legs". They do not represent the test functions. Instead, Graphs with free legs are intermediate elements to form a Feynman diagram.  
One has to pair two free legs from different vertices (to prevent self-connection) to form an edge.
The following are some examples of the valuation of Feynman diagrams considered in this paper.
\begin{example}
\begin{equs}
	&\Pi_{\F}(\text{$\FGIII$}) = \int_{\Lambda^{2}} K(x_1- x_{2})^3 dx_1dx_2
	\\&\Pi_{\F}(\text{$\FGIItwice$}) =  \int_{\Lambda^{3}} K(x_1- x_{2})^2 K(x_2- x_{3})^2dx_1dx_2dx_3
	\\&\Pi_{\F}(\text{$\FGVI$}) = \int_{\Lambda^{3}} K(x_1- x_{2})^2 K(x_2- x_{3})^2K(x_3- x_{1})^2dx_1dx_2dx_3
	\\&\Pi_{\F}(\text{$\FGIIIplus$})=\int_{\Lambda^{3}} K(x_1- x_{2}) K(x_2- x_{3})^3K(x_3- x_{1})dx_1dx_2dx_3
\end{equs}
\begin{equs}
\end{equs}
\end{example}

In this paper, we consider an important notion of equivalence for Feynman diagrams, namely that of isomorphic Feynman diagrams, defined below.
	\begin{definition}(isomorphic Feynman diagrams)
		\\
		Two Feynman diagrams $\Gamma_1, \Gamma_2 $ are isomorphic (equivalent) to each other if they have the same connection of edges between vertices provided ignoring the labels of edges and vertices. 
		We use the notation $\Gamma_1 \cong \Gamma_2$.
	\end{definition}
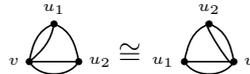
\begin{example}
	\begin{equs}
		\ISOL \cong \ISOR
	\end{equs}
\end{example}
For clarity, let us introduce some notations.
\begin{itemize}
	\item  $\mathbf{F}$ is the space of connected Feynman diagrams without decorations, orientations, legs, free legs, or self-connection. For isomorphic graphs, we treat them as the same single element in this space. 
	\item  A forest is a unconnected union of connected Feynman diagrams, corresponding to the symmetric algebra generated by connected Feynman diagrams. $\tilde{\prod}_\F$ (or alternatively $\mu_\F$ or $\tilde{\bullet}_\F$) denotes the commutative forest product which is a juxtaposition of connected Feynman diagrams.
	\item $\CF$ is the space of forests of connected Feynman diagrams in $\mathbf{F}$.
	\item $\emptyset_\F$ is the empty forest of Feynman diagrams, which means the number of Feynman diagrams in the forest is $0$. It is the identity element of the forest product. Therefore we say $\emptyset_\F \in \CF$.
	\item $\langle \mathbf{F} \rangle$ is the linear span of  connected Feynman diagrams in $\mathbf{F}$. 
	\item $\langle \CF \rangle$ is the linear span of  forests formed by connected Feynman diagrams in $\mathbf{F}$, in which $\emptyset_\F$ is the unit.
	\end{itemize}
Since a forest of connected Feynman diagrams is essentially disconnected Feynman diagrams with several connected pieces, the valuation preserves the forest product
\begin{equs}
\Pi_\F\left(\tilde{\prod}_{\F \, i=1}^n \Gamma_i \right) := \prod_{i=1}^n \Pi_\F(\Gamma_i).
\end{equs}

\subsection{Degree of a Feynman diagram and BPHZ renormalisation character}
When the kernel is smooth enough, integrals $\Pi_{\F}(\Gamma)$ are well-defined. However, suppose we have a kernel $K$ that is smooth everywhere except at the origin where its behaviour exhibits a Hölder degree $\ell < 1$ ($\ell$ can be negative). This means that for any $x \in \Lambda$ near the origin, there exists a constant $C$ such that
\begin{equs}
	|K(x)| \le C |x|^\ell.
\end{equs}
Then, the valuation of Feynman diagrams for this kernel can blow up, and the renormalisation is necessary. The first step of the renormalisation is to spot where the divergent parts are encoded in the Feynman diagram. Thanks to \cite[Proposition 2.3]{Hairer18}, this can be solved by assigning a degree to each Feynman diagram
\begin{equs}\label{eq:Fdeg}
\deg \Gamma : =  \ell |\CE| + d(|\CV| - 1)
\end{equs}
where $|\CE|$ and $|\CV|$ representing the number of edges and the number of vertices are cardinals of sets $\CE$ and $\CV$, respectively. We also have to introduce the term ``subgraph" before illustrating how to detect the divergence. A subgraph $\bar\Gamma$ of a Feynman diagram $\Gamma (\CV, \CE)$ has vertices set $\bar \CV \subset \CV$ and edges set $\bar \CE \subset \CE$ where  each of its vertices is incident to at least one edge in $\bar\CE$, and we denote $\bar\Gamma \subset \Gamma$. Notice that $\bar\Gamma \in \CF$ is not necessarily connected which means that it can be a forest consisting of multiple individual sub-Feynman diagrams.
\begin{proposition} \label{prop:divergence_degree}
	For a Feynman diagram $\Gamma$, if $\deg \bar\Gamma > 0$ for every subgraph $\bar\Gamma \subset \Gamma$ then
	the integral $\Pi_\F(\Gamma)$ defined in \eqref{eq:valuation} is bounded. 
\end{proposition}
This proposition is a special case of \cite[ Proposition 2.3]{Hairer18} as we do not have legs. Theorem 5.2.3. in \cite{B22} gives a proof of this case. We then define the space $\mathbf{F}_{-}$  of Feynman diagrams with non-positive degrees, the space $\CF_{-}$ of non-empty forests in which each individual Feynman diagram has a non-positive degree, and the linear span of $\CF_{-}$ with unit $\emptyset_\F$ as $\langle \CF_- \rangle$. 

The BPHZ renormalisation is implemented through the twisted antipode  $ \CA_\F : \langle \CF_- \rangle \rightarrow \langle \CF \rangle $  defined recursively as
\begin{equs}\label{eq:antipode1}
	& \CA_\F(\emptyset_\F) := \emptyset_\F, \quad \CA_\F( \Gamma_1 \tilde{\bullet}_\F \ \Gamma_2) := \CA_\F( \Gamma_1 ) \tilde{\bullet}_\F \CA_\F( \Gamma_2 ) \quad \text{ for any } \Gamma_1, \Gamma_2 \in \mathcal{F}_-,
	\\&
	\CA_\F(\Gamma) := -\Gamma - \sum_{\CF_{-} \ni \bar\Gamma \subsetneq \Gamma} \CA_\F(\bar \Gamma)\tilde{\bullet}_\F(\Gamma / \bar \Gamma),
	 \quad \text{ for any } \Gamma \in \mathbf{F}_-,
\end{equs}
where  $\Gamma / \bar \Gamma$ is obtained by extracting each individual Feynman diagram in the forest $\bar \Gamma$ from $\Gamma$ and contracting each to a single vertex. The map $\CA_\F$ is called a ``twisted" antipode because we only apply it recursively to the forest with non-positive degree $\bar\Gamma$ but not to remaining $\Gamma / \bar \Gamma$.
 Below, we provide one example 
\begin{equs}
	\Gamma = \FGIIIplus, \quad \bar \Gamma = \FGIII, \quad \Gamma / \bar \Gamma = \YII.
\end{equs}
  The lower vertex in $\Gamma / \bar \Gamma$ is obtained by contracting $\bar{\Gamma}$ and extracting it to a single node. Then, the BPHZ renormalisation character can be defined by 
\begin{equs} \label{eq:BPHZ}
	\Pi^{\BPHZ}_\F (\Gamma) = \Pi_\F \left(\CA_\F(\Gamma)\right).
\end{equs}
Let us stress again that the formula for $\CA_\F$ could be more involved if one has to face a quite high degree of divergence. In this work, we limit ourselves to the case without going to higher orders. The general case requires the introduction of monomials and derivatives on the kernel $K$, which means decorations on vertices and edges are needed.

\subsection{Extraction-contraction coproduct, BPHZ renormalisation, and simultaneous insertion product of Feynman diagrams}
One can observe that in the twisted antipode \eqref{eq:antipode1}, the recursive part is controlled by the reduced extraction-contraction coproduct defined below.
\begin{equs} \label{eq:EC_F}
	\Delta_{\F} \Gamma =  \sum_{\CF_{-} \ni \bar\Gamma \subsetneq \Gamma} \bar\Gamma \otimes \Gamma / \bar\Gamma 
\end{equs}
where the sum and $ \Gamma / \bar\Gamma $ are the same as in \eqref{eq:antipode1}. 
The twisted antipode can be rewritten as 
\begin{equs}
		& \CA_\F(\emptyset_\F) := \emptyset_\F,
	\\&
	\CA_\F(\Gamma) = -\Gamma - \mu_\F\circ(\CA_\F \otimes \id)\circ\Delta_{\F} \Gamma,
	\quad \text{ for any } \Gamma \in \mathbf{F}_-.
\end{equs}
The extraction-contraction coproduct is the reduced version plus the primitive terms.
\begin{equs}
\Delta^{\! -}_{\F} \Gamma =  \emptyset_\F \otimes \Gamma + \Gamma \otimes \emptyset_\F +
 \sum_{\CF_{-} \ni \bar\Gamma \subsetneq \Gamma} \bar\Gamma \otimes \Gamma / \bar\Gamma.
\end{equs}
It is straightforward to verify from this definition that $\Delta^{\! -}_{\F}$ preserves the forest product:
\begin{equs}
	\Delta^{\! -}_{\F} \left(\Gamma_1 \tilde{\bullet}_{\F} \Gamma_2 \right) = \Delta^{\! -}_{\F} (\Gamma_1) \tilde{\bullet}_{\F} \Delta^{\! -}_{\F}(\Gamma_2) 
\end{equs}
where the forest product on the right-hand side has to be understood in the following way:
\begin{equs}
	\left( \Gamma_1 \otimes \Gamma_2 \right) \tilde{\bullet}_{\F} \left( \Gamma_3 \otimes \Gamma_4 \right) = \left( \Gamma_1  \tilde{\bullet}_{\F} \Gamma_3 \right) \otimes \left( \Gamma_2 \tilde{\bullet}_{\F} \Gamma_4 \right)  
\end{equs}
for $ \Gamma_i \in \mathcal{F} $.
Moreover, equivalently we can write $\CA_\F$ in the following form 
\begin{equs}
	\mu_\F\circ(\CA_\F \otimes \id)\circ\Delta^{\! -}_{\F} \Gamma	= 0
\end{equs}
for $ \Gamma $ not empty and
where $ \mu_\F $ is defined by
\begin{equs}
	\mu_\F( \Gamma_1\otimes \Gamma_2) = \Gamma_1\tilde{\bullet}_{\F}  \Gamma_2.
\end{equs}
The BPHZ renormalisation procedure with the character \eqref{eq:BPHZ} defined via $\CA_{\F}$ is
\begin{equs}\label{eq:BPHZ_martin}
	\hat{M}_\F := \left(\Pi_{\F} (\CA_{\F}\cdot) \otimes \id \right) \Delta_{\F}^{-} .
\end{equs}
The following example shows how to compute $\Delta_{\F}$, $\Delta^{\! -}_{\F}$ and $\CA_{\F}$.
\begin{example}
	Consider the graph $\FGIIIplus$ and 
	suppose that we are dealing with the $\Phi^4_3$ model ($d = 3$), which means the kernel is the Green's function and it behaves like $-1$-Hölder ($\ell = -1$). Then 
	\begin{equs}
		\deg \Gamma : =  - |\CE| + 3(|\CV| - 1)
	\end{equs}
	Then the only subgraph of $\FGIIIplus$ with a non-positive degree  is $\FGIII$. Therefore we have
	\begin{equs}
		\Delta_{\F} \left(\FGIIIplus\right) & = \FGIII \otimes \YII
	\\
		\Delta^{\! -}_{\F} \left(\FGIIIplus\right) & = \emptyset_\F \otimes \FGIIIplus + \FGIIIplus \otimes \emptyset_\F + \FGIII \otimes \YII
	\\
		\CA_\F \left(\FGIIIplus\right) &= - \FGIIIplus + \FGIII \tilde\bullet_\F \YII
	\end{equs}
\end{example}
\vspace{10pt}

We then need to introduce some operators on Feynman diagrams for passing to the dual side of the reduced extraction-contraction coproduct.
\begin{itemize}
	\item Firstly, we define the operator $C_v(\Gamma)$ which cuts the vertex $v \in \CV(\Gamma)$ down and transforms edges connecting other vertices $\bar{v}_1,\ldots, \bar{v}_{k_v} \in \CV(\Gamma)$ with $v$ to ``free legs" (half-edges) attached to $\bar{v}_1,\ldots, \bar{v}_{k_v}$, where $k_v$ is the number of edges connected to $v$ in $\Gamma$ and $\bar{v}_1,\ldots, \bar{v}_{k_v}$ are not necessarily distinct.
	\item Secondly, 
	for $\hat \Gamma$, a graph with half-edges attached to $\Gamma \in \mathbf{F}$, we associate the map $\mathfrak{L}_{\hat \Gamma}$ which sends a half-edge $l_i$ in $\hat \Gamma$ to the vertex $v_i$ where it is attached to.
	\item Finally, for any $\hat \Gamma$ which is a graph obtained by adding $k$ half-edges to a Feynman diagram, we define
	 the operator
	\begin{equs}
		\CG \left( \hat \Gamma, \Gamma_1\right)  := \sum_{u_1,...,u_{k} \in \CV(\Gamma_1)} \hat \Gamma  \curvearrowright^{ \prod_{i=1}^{k} (\mathfrak{L}_{\hat \Gamma}(l_i),u_i)}\Gamma_1
	\end{equs}
	acting on $\Gamma_1 \in \mathbf{F}$ and $\hat \Gamma$,
	where $ \hat \Gamma \curvearrowright^{(\mathfrak{L}_{\hat \Gamma}(l_i),u_i)}\Gamma_1$ means connecting $\hat \Gamma$ and $\Gamma_1$ by adding an edge between the vertices $\mathfrak{L}_{\hat \Gamma}(l_i)$ and $u_i$ and eliminating the half-edge $l_i$. Note that $u_1, \ldots, u_k $ are not necessarily distinct.
\end{itemize}
Finally, we have the ``insertion product"  $\triangleright$ ,  a pre-Lie product on
Feynman diagrams.
\begin{definition}(insertion product) \\
	For any $\Gamma_1, \Gamma_2 \in \mathbf{F}$, the insertion product  $\triangleright: \mathbf{F} \otimes \mathbf{F} \mapsto \langle \mathbf{F} \rangle$ is
	\begin{equs} \label{eq:def_insertion}
		\Gamma_1 \triangleright \Gamma_2 := \sum_{v \in \CV(\Gamma_2)}\CG \left(C_v(\Gamma_2), \Gamma_1\right).
	\end{equs}
\end{definition}
This product can be explained as a two-step algorithm.  Firstly, we choose one vertex $v$ in $\Gamma_2$ and cut it down, converting all the edge that were previously attached to $v$ to half-edges. We call remaining part (with those half-edges)  in $\Gamma_2$ a ``trunk". Secondly, we connect each half-edge in the trunk to one vertex in $\Gamma_1$.
\begin{remark}
	Our definition of insertion is a modification of the one in the literature (see \cite{K2006} and \cite{VS2007}), as we want an insertion whose Guin-Oudom construction is the adjoint of the reduced extraction-contraction coproduct under the inner product defined via the symmetry factor of Feynman diagrams (will be later introduced in Section \ref{subsec:symmetry}).
\end{remark}

	However, if we consider one specific equation, only certain types of vertices will show up in Feynman diagrams encoding the expansions of solutions or invariant measures. For example, in the renormalisation of $\Phi_3^4$ model, all the vertices in the Feynman diagrams must have arity $2$ or $4$, which means each vertex has $2$ or $4$ edges connected to it. Therefore, when considering a specific equation, it is restricted to Feynman diagrams with certain types of vertices, and this restriction is called a rule.
\begin{definition} \label{def:rules}
(Rules of insertion)
\\
A rule $\CR$ is a set of vertex types
\begin{equs}
	\CR := \{v: k_v \in \CK(\CR) \}
\end{equs}
where $k_v$ is the arity of the vertex $v$, the number of edges attached to $v$, and $\CK(\CR)$ is the set of possible arities. Applying a rule to an insertion product means projecting the insertion defined by \eqref{eq:def_insertion} to
\begin{equs} 
	\Gamma_1 \triangleright_{\CR} \Gamma_2 := \sum_{v \in \CV(\Gamma_2)}\CG_\CR \left(C_v(\Gamma_2), \Gamma_1\right)
\end{equs}
where $\Gamma_2$ obeys the rule $\CR$ and 
	\begin{equs}
	\CG_{\CR} \left(\hat \Gamma, \Gamma_1\right)  := \sum_{u_1,...,u_{k} \in \CV_\CR(\Gamma_1)} \hat \Gamma  \curvearrowright^{ \prod_{i=1}^{k} (\mathfrak{L}_{\hat \Gamma}(l_i),u_i)}\Gamma_1.
\end{equs}
Here $\CV_\CR(\Gamma_1)$ means, after grafting free legs to $\Gamma_1$, each vertex in the Feynman diagram $\hat \Gamma  \curvearrowright^{ \prod_{i=1}^{k} (\mathfrak{L}_{\hat \Gamma}(l_i),u_i)}\Gamma_1$ satisfies the rule $\CR$.
\end{definition}
We use the example below to illustrate how the insertion product works.
\begin{example} 
We consider the $\Phi_3^4$ model, which means the rule is
\begin{equs}
	\CR = \{v: k_v \in \{2,4\} \}.
\end{equs} 
Let us consider the insertion $\Gamma_1 \triangleright_{\CR} \Gamma_2$ with rule $\CR$,
$\Gamma_1 = \FGIII$, and $\Gamma_2 = \YII$.
For clarity, we index vertices in $\Gamma_1 = \FGIIIi$ and index vertices as well as edges in $\Gamma_2 = \YIIi$. Then,
\begin{equs}
	\sum_{v \in \CV(\Gamma_2)}\CG_\CR \left(C_v(\Gamma_2), \Gamma_1\right) = \CG_\CR \left(C_{v_1}(\Gamma_2), \Gamma_1\right)  +\CG_\CR \left(C_{v_2}(\Gamma_2), \Gamma_1\right). 
\end{equs}
We now compute the term $\CG_\CR \left(C_{v_2}(\Gamma_2), \Gamma_1\right)$. Firstly, 
\begin{equs}
	C_{v_2}(\Gamma_2) = \YIIt 
\end{equs}
where $\ell_1$ is from cutting $e_1$ and $\ell_2$ is from $e_2$.
Then, according to the rule, we have only two possible grafting $\{(\ell_1, u_1), (\ell_2, u_2)\}$ and $\{(\ell_1, u_2), (\ell_2, u_1)\}$. Thus,
\begin{equs}
	\CG_\CR \left(C_{v_2}(\Gamma_2), \Gamma_1\right) = \FGIIIplusL + \FGIIIplusR
\end{equs}
where $e'_i$ and $e''_i$ are obtained from grafting the half-edge $\ell_i$ for $i = 1, 2$.
The insertion $\CG_\CR \left(C_{v_1}(\Gamma_2), \Gamma_1\right)$ is the same as that for $v= v_2$. Finally, since all the obtained terms are isomorphic to each other, we have 
\begin{equs}
\Gamma_1 \triangleright_{\CR} \Gamma_2 = 4	\FGIIIplus.
\end{equs}
\end{example}

The following-defined`` simultaneous insertion" product obtained by the Guin--Oudom  construction (see \cite{Guin1} \cite{Guin2}) on  $\triangleright$ will later be shown in Theorem~\ref{them: adjoint_Feyman} to be the adjoint dual of the reduced extraction-contraction coproduct $\Delta_\F$.
\begin{definition} \label{def:insertion_FD}
	For any Feynman diagrams $\Gamma_1,\ldots, \Gamma_n, \bar\Gamma \in \mathbf{F}$, the simultaneous insertion product $\star_\F : \CF \times \mathbf{F} \mapsto \langle \mathbf{F}\rangle$ is
\begin{equs}
	\tilde{\prod}_{\F \, i=1}^n \Gamma_i \star_{\F} \bar\Gamma 
	:=
	 \sum_{v_1 \ne v_2 \ne \ldots \ne v_n \in \CV(\bar\Gamma)}\CG\left[C_{v_n}\left(\ldots \CG\left[ C_{v_2}\left(\CG\left[C_{v_1}(\bar\Gamma), \Gamma_1\right]\right),\Gamma_2\right] \ldots \right) , \Gamma_n \right].
\end{equs}
\end{definition}
This product can be interpreted as the following algorithm: Firstly, we choose $n$ distinct vertices in $\bar\Gamma$ as spots to insert $\Gamma_1,\ldots,\Gamma_n$ respectively. ``Distinct" means that each node $v$ can be used at most once as the position of insertion. Secondly, for each $i$, graft the free legs obtained by cutting $v_i$ down to vertices of $\Gamma_i$. The order of the vertices does not matter for defining this insertion.
We can also rewrite the insertion product as 
\begin{equs}
	\tilde{\prod}_{\F \, i=1}^n \Gamma_i \star_{\F} \bar\Gamma = \sum_{\Gamma \in \mathbf{F}} I(\tilde{\prod}_{\F \,i=1}^n \Gamma_i, \bar\Gamma, \Gamma)\Gamma
\end{equs}
where $ I( \tilde{\prod}_{\F \, i=1}^n \Gamma_i, \bar\Gamma, \Gamma)  $ is the number of insertions of $\tilde{\prod}_{\F \, i=1}^n \Gamma_i$ into $\bar\Gamma$ that give $\Gamma$ up to isomorphism.

\section{Multi-indices} \label{Sec::3}

Feynman diagrams appear in the literature to represent problems that can be described as pairing half-edges of some vertices. For example,  the Wick power $:X^n:$ of a Gaussian variable $X$ can be expressed as the Hermite polynomial  
\begin{equs}
	H_n(X) = \left(-\mathrm{Var}(X)\right)^n \exp \left({\frac{X^2}{2\mathrm{Var}(X)}}\right) \frac{d^n}{dX^n}\exp \left({\frac{-X^2}{2\mathrm{Var}(X)}}\right)
\end{equs}
which can be later graphically represented by a vertex with $n$ half edges. When taking covariances of integral of Hemite polynomials, we essentially use those vertices to form Feynman diagrams by pairing their half-edges. The analytical foundation of this process is
the \cite[Lemma 4.2.9]{B22} cited below.
\begin{lemma}\label{lemma:expectation}
	(\cite[Lemma 4.2.9]{B22})
	
	Let $X$ and $Y$ be jointly Gaussian centred random variables. Then for any $n,m \ge 0$ one has
	\begin{equs}
		\mathbb{E}\left[H_n(X)H_m(Y)\right] = \delta_m^n n!\mathbb{E}\left[XY\right]^n
	\end{equs}
	where $\delta_m^n $ is the usual Dirac's delta.
\end{lemma}
Now we rewrite $H_n(X)$ as $ \psi_{n,X}(x)$ and $H_m(Y)$ as $ \psi_{m,Y}(y)$. Note that, by the definition of the Hermite polynomial, $H_1 (X)= X$. Thus, we use $X(x)$ representing $ \psi_{1,X}(x)$.
Then, by Lemma \ref{lemma:expectation}, the expectation 
\begin{equs}
	\mathbb{E}\left[\int_\Lambda \psi_{n,X}(x) dx \int_\Lambda \psi_{m,Y}(y) dy\right] &= 
	\int_{\Lambda^2} \mathbb{E}\left[\psi_{n,X}(x)\psi_{m,Y}(y)\right] dx dy
	\\&=
		\int_{\Lambda^2} \mathbb{E}\left[\psi_{n,X}(x)\psi_{m,Y}(y)\right] dx dy
	\\&=
	\delta_m^n n! \int_{\Lambda} \mathbb{E}\left[X(x)Y(y)\right]^n dxdy.
\end{equs}
Suppose that the covariance $\mathbb{E}\left[X(x)Y(y)\right]$ can be described as a Kernel $K(x-y)$ for $x,y$ in the space $\Lambda$. We can further rewrite
the expectation as
\begin{equs}\label{eq:expectation}
	\mathbb{E}\left[\int_\Lambda \psi_{n,X}(x) dx \int_\Lambda \psi_{n,Y}(y) dy\right] =
\delta_m^n n! \int_{\Lambda} K(x-y)^n dxdy.
\end{equs}
Observe that the integral on the right-hand side is a valuation of Feynman diagram with two vertices and $n$ edges in between and $n!$ is the number of pairings of the half edges to form the edges.
Specifically, the valuation $\Pi_{\F}$ translates an edge to $\mathbb{E}[H_1(X_{e_+})H_1(X_{e_-})] = \mathbb{E}[X_{e_+}(x_+)X_{e_-}(x_-)] $, the covariance of variables encoded by the edge's two end vertices, which can be further  calculated as the integral of kernels as defined in \eqref{eq:valuation}. For example,
\begin{equs}
	\Pi_\F \left(\FGIII\right) = \int_{\Lambda^2} \mathbb{E}[X(x)Y(y)]^4dxdy = \int_{\Lambda^2} K(x-y)^4 dxdy
\end{equs}
where the left vertex represents the variable $X$ and the right one denotes $Y$.
A generalisation of Lemma \ref{lemma:expectation} shows that
\begin{equs}\label{eq:pairing_problem}
	\mathbb{E}\left[ \prod_{i} \int_{\Lambda}H_{k_i}(X(x_i))dx_i\right] = \sum_{F \in \CF} N(F)\Pi_{\F} \left( F\right)
\end{equs}
 where $\int_{\Lambda}H_{k_i}(X(x_i))dx_i$ is regarded  as a vertex with $k_i$ half-edges to be paired with half-edges attached to other vertices, and 
 $N(F)$ is the number of distinct pairings of half-edges to form forests of Feynman diagrams isomorphic to $F$.
From the nature of the ``pairing half-edges" problem discussed above, it follows that
 the vertex set suffices to describe $N(F)$ and that all the diagram structures are inherited implicitly. Thus, $\Pi_{\F}(F)$ can also be calculated through the vertex set.
Therefore, it is natural to study the set of vertices which is called ``pre-Feynman diagrams" in \cite{pre_FD}. When the graphical structures are irrelevant, algebra on the vertex set is more suitable.
Later (in Proposition \ref{prop:Sta_Orb}) we will show that there exists a map $\CP$ which helps understand the pairing and gives $N$ under suitable symmetry factors for Feynman diagrams and for multi-indices respectively. This $\CP$ is one of the key elements in showing the equivalence of renormalisations $\hat{M}_\F$ and $\hat{M}_\M$.

In this section, we will show that the vertex set can be formally described by the combinatorial object ``multi-indices". 
Then, in section \ref{sec:BPHZ}, it will be proven that the BPHZ renormalisation via multi-indices can be implemented directly without pairing their half-edges to form Feynman diagrams. 
Moreover, we will use the pairing problem \eqref{eq:pairing_problem} as an example to discuss the choice of valuation map $\Pi_{\M}$ \eqref{eq:valuation_M}. One can generalise the algebraic results to other ``pairing half-edges" problem provided that reasonable valuation maps exist.

\subsection{Correspondence between Multi-indices and Feynman diagrams}
A multi-indice is defined through abstract variables $(z_k)_{k \in \mathbb{N}}$.  Over these variables, a multi-index $ \beta : \mathbb{N} \rightarrow \mathbb{N}  $ is a map with finite support and we define a multi-indice $z^\beta$ as the monomial
\begin{equs} \label{def:multi_indices}
	z^{\beta} : = \prod_{k \in \mathbb{N}} z_k^{\beta(k)},
\end{equs}
where ``finite support" means $\beta$ has finitely many non-zero entries.
Given a graph $ G = (\mathcal{V},\mathcal{E}) $, we associate the variable $z_{k(v)}$ to one vertex $v \in \CV$ such that the number of edges attached to $v$ is $k(v)$ and we call $k(v)$ the arity of $v$. Then $ \beta(k) \in \mathbb{N}$, one entry of $\beta$, counts the number of vertex with arity $k$ in the Feynman diagram, which is described by the following defined ``counting map" $\Phi:  \mathbf{F}  \mapsto \mathbf{M}$
\begin{equs} \label{eq:Phi_map}
	\Phi(\Gamma) := \prod_{ v \in \CV(\Gamma)} z_{k(v)}
\end{equs}
where $\mathbf{M}$ is the set of non-empty multi-indices $z^\beta$, i.e., the entries of $\beta$ cannot be all $0$. We also require multi-indices in the space $\mathbf{M}$ to be able to form meaningful Feynman diagrams by pairing free legs. This implies that $\mathbf{M}$ is the image of the counting map. One can ignore this restriction since, by definition in \eqref{eq:valuation_M} and Lemma \ref{lemma:expectation}, the valuation $\Pi_{\M}$ is $0$ for those multi-indices which do not satisfy this condition. The restriction is stated here purely for clarity.

The following notations will also be used in the sequel
\begin{itemize}
	\item $\tilde{\prod}_\M$ (or alternatively $\mu_\M$ or $\tilde{\bullet}_\M$) denotes the commutative forest product which is a juxtaposition of multi-indices, which indicates that there is no order in the forest, i.e., $z^\alpha \tilde{\bullet}_\M z^\beta \equiv z^\beta \tilde{\bullet}_\M z^\alpha$ and they are regarded as the same element. 
	\item $\CM$ is the space of forests of multi-indices in $\mathbf{M}$.
	\item $\emptyset_\M$ is the empty forest of multi-indices, which means the cardinal of this collection of non-empty multi-indices is $0$. It is the identity element of the forest product. Consequently, we say $\emptyset_\M \in \CM$.
	\item $\langle \mathbf{M} \rangle$ is the space of linear span of  non-empty multi-indices.
	\item $\langle \CM \rangle$ is the space of linear span of  forests of multi-indices. 
\end{itemize}
 Therefore, $(\langle \CM \rangle, \tilde{\bullet}_\M)$ is the symmetric algebra of multi-indices with identity $\emptyset_\M$.
 
To further simplify the notation, we use $\tilde{z}^{\tilde{\beta}}$ denoting a forest of multi-indices, where $\tilde{\beta}$ is a collection of populated multi-indices without order.  For $\tilde{\beta} = \{\beta_1,\ldots,\beta_n\}$
\begin{equs}
	\tilde{z}^{\tilde{\beta}} := 	\tilde{\prod}_{\M \, j=1}^n z^{\beta_j}.
\end{equs}
The repetition of individual populated multi-indices in $\tilde{\beta}$ is allowed, which means $z^{\beta_1},\ldots,z^{\beta_n}$ are not necessarily distinct. For a forest of Feynman diagrams we define the counting map
\begin{equs}
	\Phi \left(\tilde\prod_{\F \, i=1}^n \Gamma_i\right) := \tilde \prod_{\M \, i = 1}^n \Phi(\Gamma_i).
\end{equs}
\vspace{6pt}

Similar to Feynman diagrams, the renormalisation through multi-indices can also be studied via Hopf algebras whose grade is given by the norm of a multi-indice $z^\beta$, which is the sum of each element of $\beta$
	\begin{equs}
		|z^\beta| := \sum_{k \in \mathbb{N}} \beta(k).
	\end{equs}
We define a symmetry factor different from the ones in the literature, which describes the adjoint relation between the product and its dual coproduct.
\begin{equs} \label{eq:SF_multiindices}
		S_\M(z^{\beta}) := \prod_{k \in \mathbb{N}} \beta(k)! (k!)^{\beta(k)}
\end{equs}
This symmetry factor consists of two parts: $\beta(k)!$ and $(k!)^{\beta(k)}$. From a Feynman diagram  point of view, $\beta(k)!$ counts the permutation of vertices with the same arity, and $(k!)^{\beta(k)}$ amounts to the cardinal of permuting half-edges attached to the same node. 
The above-mentioned values can be extended to a forest of multi-indices. The norm is defined as
\begin{equs}
	\left|\tilde{\prod}_{\M \, i=1}^m z^{\beta_i}\right| := \sum_{i=1}^m |z^{\beta_i}|.
\end{equs}
The symmetry factor for a forest $\tilde{z}^{\tilde{\beta}} =\tilde{\prod}_{\M \, i=1}^m \left(z^{\beta_i}\right)^{\tilde{\bullet_{\M}}r_i}$	with distinct $\beta_i$ is
\begin{equs}  \label{symmetry_factor_2}
	S_\M \left(\tilde{z}^{\tilde{\beta}}\right) := \prod_{i=1}^m r_i!\left(S_\M(z^{\beta_i})\right)^{r_i}.
\end{equs} 
\vspace{6pt}

For two forests of multi-indices $\tilde{z}^{\tilde{\alpha}}$ and  $\tilde{z}^{\tilde{\beta}}$, the pairing (do not confuse it with the pairing of half-edges) is defined as the inner product
\begin{equs} \label{eq:inner_multiindices}
	\langle \tilde{z}^{\tilde{\alpha}}, \tilde{z}^{\tilde{\beta}}\rangle := S_\M(\tilde{z}^{\tilde{\alpha}})\delta^{\tilde{\alpha}}_{\tilde{\beta}},
\end{equs}
where $\delta^{\tilde{\alpha}}_{\tilde{\beta}}$ is the usual Dirac's delta. Then we can define a map $\CP$ which is proved in proposition \ref{prop:adjoint_P_Phi} to be the adjoint of the counting map $\Phi$ under this inner product .
\begin{definition} \label{def:map_P}
	The map $\CP: \mathbf{M} \rightarrow \langle \mathbf{F}\rangle$ is defined as
	\begin{equs}
		\CP(z^\beta) := \sum_{\Gamma: \Phi(\Gamma) = z^\beta} \frac{S_\M(z^\beta)}{S_\F(\Gamma)} \Gamma.
	\end{equs}
	where $S_\F(\Gamma)$ is the symmetry factor of the Feynman diagram $\Gamma$ and its explicit formula will be discussed in Section \ref{subsec:symmetry}.
\end{definition}
Since the support of $\beta$ is finite, the map $\CP$ as a sum is well-defined.
\begin{proposition} \label{prop:injective_P}
	The map $\CP$ is injective.
\end{proposition}
\begin{proof}
	This is simply because while summing up $\Gamma: \Phi(\Gamma) = z^\beta$ it can be seen from the definition of the counting map $\Phi$ that the sets of  Feynman diagrams $\{\Gamma: \Phi(\Gamma) = z^\beta\}$ are different for different $z^\beta$.
\end{proof}
\begin{proposition} \label{prop:adjoint_P_Phi}
		Under the inner product of multi-indices defined in \eqref{eq:inner_multiindices}  and the following defined inner product of Feynman diagrams
	\begin{equs} 	\label{inner_product_diagram}
		\langle\Gamma_1, \Gamma_2\rangle := \delta^{\Gamma_1}_{\Gamma_2}S_\F(\Gamma_1),
	\end{equs} 
	$\CP$ and $\Phi$ are adjoint , i.e., $
		\langle \Phi(\Gamma), z^\beta \rangle = \langle \Gamma,  \CP(z^\beta)\rangle $.
\end{proposition}
\begin{proof}
	We firstly check the case that $\Phi(\Gamma) \ne z^\beta$. By the definition of the inner product, $\langle \Phi(\Gamma), z^\beta \rangle = 0$. Since $\Phi(\Gamma) \ne z^\beta$, there is no term including $\Gamma$ in the sum $\CP(z^\beta)$. 
	Thus, 
	\begin{equs}
	\langle \Gamma,  \CP(z^\beta)\rangle   = 0 =	\langle \Phi(\Gamma), z^\beta \rangle.
	\end{equs}
	Then, for $\Phi(\Gamma) =  z^\beta$, we have $\langle \Phi(\Gamma), z^\beta \rangle  = S_\M(z^\beta)$ and 
	\begin{equs}
			\langle \Gamma,  \CP(z^\beta)\rangle = S_\F(\Gamma)\frac{S_\M(z^\beta)}{S_\F(\Gamma)} = S_\M(z^\beta).
	\end{equs}
\end{proof}
\vspace{6pt}

 In the BPHZ renormalisation, it is necessary to detect the divergent subgraphs through their degrees. Therefore, we also need a degree map for the renormalisation via multi-indices.
Since the degree of a Feynman diagram depends only on the considered model, the number of vertices and edges but not on the structure of the diagram, Feynman diagrams with the same associated multi-indices have the same degree. Comparing with the degree of Feynman diagrams defined in \eqref{eq:Fdeg},
it is natural to define the degree of a multi-indice as
\begin{equs}
	\deg z^\beta : =   \frac{\ell}{2}\sum_{k \in \mathbb{N}}k \beta(k) + d(|z^\beta| - 1)
\end{equs}
where $\ell$ and $d$ are the same as defined in \eqref{eq:Fdeg}. Then, for any $\Gamma \in, \mathbf{F}$ 
\begin{equs} \label{eq:deg}
	\deg \Gamma = \deg \Phi(\Gamma).
\end{equs}
In the sequel of the paper, the space of multi-indices with non-positive degrees is denoted by $\mathbf{M}_{-}$. Moreover, $\CM_{-}$ represents the space of non-empty forests in which each individual multi-indice has a non-positive degree.

\subsection{Insertion product and extraction-contraction coproduct of multi-indices}
The reduced extraction-contraction coproduct $\Delta_\M$ of multi-indices is the dual of the product $\star_\M$ understood as `` simultaneous insertion" which is the Guin–Oudom generalisation of the product $ \blacktriangleright : \mathbf{M} \times \mathbf{M} \mapsto \langle \mathbf{M}\rangle$ defined below. For $z^\beta, z^\alpha \in \mathbf{M}$ 
\begin{equs} \label{insertion_product}
	z^\beta \blacktriangleright z^\alpha : =\sum_{k \in \mathbb{N}}\left(D^{k}z^{\beta}\right)
	\left(\partial_{z_{k}}z^\alpha\right)
\end{equs}
where $D$ is the derivation defined as 
\begin{equs}
	D := \sum_{k \in \mathbb{N}} z_{k+1} \partial_{z_k}
\end{equs}
and $ \partial_{z_k} $ is the ordinary partial derivative in the coordinate $ z_k $. 
This product represents replacing one single variable $z_k$ by $D^{k}z^{\beta}$. Notice that the empty multi-indice $z^\mathbf{0}$ is not in the set $\mathbf{M}$, which means inserting the empty multi-indice is forbidden. The derivative $\partial_{z_k}$ ensures that the insertion of multi-indices encapsulates the feature of its Feynman diagram counterpart: one vertex in the Feynman diagram can only be used as the insertion position at most once. The same as for Feynman diagrams, one can define a rule for inserting multi-indices $\blacktriangleright_{\CR}$ by projecting the result to terms whose trunk multi-indices has support  $\subseteq \CK(\CR)$.
\vspace{6pt}

Since the insertion product of multi-indices traces the arity change while inserting Feynman diagrams, one has the following morphism property.
\begin{proposition} \label{prop:morphism}
	For any Feynman diagrams $\Gamma_1, \Gamma_2 \in \mathbf{F}$
	\begin{equs}
		\Phi(\Gamma_1 \triangleright \Gamma_2) = \Phi(\Gamma_1) \blacktriangleright \Phi(\Gamma_2).
	\end{equs}
\end{proposition}
\begin{proof}
	Let $L$ denote the number of half-edges to be (or ``can be" if there is a rule $\CR$) attached to $\Gamma_1$ and $k_v$ be the arity of the vertex $v$. Then we have 
	\begin{equs}
		\Gamma_1 \triangleright \Gamma_2 = \sum_{L \in \mathbb{N}}\sum_{ v \in \CV(\Gamma_2): k_{v} = L} \Gamma_1 \widehat{\triangleright}_v \Gamma_2
	\end{equs}
	The operator $\widehat{\triangleright}_v$ amounts to inserting $\Gamma_1$ to the position of $v$, which means we first detach all the edges connected to $v$. In this process, the frequencies of vertices in $\Gamma_2$ are kept except one vertex $v$ is eliminated. Meanwhile attaching $L$ free legs to $\Gamma_1$ means that the total increase in arities of vertices in $\Gamma_1$ is $L_1$.
	Therefore,  if $\partial_{z_{L}}\Phi(\Gamma_2) \ne 0$ and $k_v = L$,
	\begin{equs}
		\Phi(\Gamma_1 \widehat{\triangleright}_v \Gamma_2) = \frac{\Phi(\Gamma_2)}{z_{L}}\sum_{\ell} {L \choose \ell}\Phi(\Gamma_1)\prod_{i=1}^{|\CV(\Gamma_1)|}\frac{z_{k_{u_i}+\ell_i}}{z_{k_{u_i}}}
	\end{equs}
	where $\ell$ is a $|\CV(\Gamma_1)|$-dimensional  multi-index with $|\ell| = L$ and $k_{u_i}$ is the arity of vertex $u_i \in \CV(\Gamma_1)$.
	Since $D^{L}$ obeys the Leibniz rule, we have 
	\begin{equs}
		\Phi(\Gamma_1 \widehat{\triangleright}_v \Gamma_2) = \frac{\Phi(\Gamma_2)}{z_{L}} D^{L}\Phi(\Gamma_1)
	\end{equs}
	for $\partial_{z_{L}}\Phi(\Gamma_2) \ne 0$. Thus, 
	\begin{equs}
		\Phi(\Gamma_1 \triangleright \Gamma_2) = \sum_{L \in \mathbb{N}}
		\sum_{ v \in \CV(\Gamma_1): k_{v} = L } 
		\frac{\Phi(\Gamma_2)}{z_{L}} D^{L}\Phi(\Gamma_1)
	\end{equs}
	which equals the right-hand side
	\begin{equs}
		\Phi(\Gamma_1) \blacktriangleright \Phi(\Gamma_2) = \sum_{L \in \mathbb{N}} \partial_{z_{L}}  \Phi(\Gamma_2) D^{L}\Phi(\Gamma_1)
	\end{equs}
	since $\partial_{z_{L}}$ counts the number of vertices with arity $L$ in $\Gamma_2$.
\end{proof}

\begin{definition} (Simultaneous insertion of multi-indices)
	\\	
	The simultaneous insertion $\star_\M : \CM \times \mathbf{M} \mapsto \langle \mathbf{M}\rangle$ which inserts a non-empty forest of $z^{\beta_i} \in \mathbf{M}$ into $z^\alpha \in \mathbf{M}$ is
	\begin{equs} \label{multiple_insertion}
		\tilde{\prod}_{\M \, i=1}^n z^{\beta_i} \star_\M z^\alpha:= \sum_{k_1,...,k_n\in \mathbb{N}}  
		\left(\prod_{i=1}^nD^{k_i}z^{\beta_i}\right)\left[\left(\prod_{i=1}^n\partial_{z_{k_i}}\right)z^\alpha\right] ,
	\end{equs}
	where $n \ge 1$.
\end{definition}
One can verify that 
\begin{equs} \label{counting_map_morphism}
	\Phi(	\tilde{\prod}_{\F \, i=1}^n \Gamma_i \star_{\F} \bar\Gamma ) = \Phi(\tilde{\prod}_{i=1}^n \Gamma_i) \star_\M \Phi(\bar\Gamma)
\end{equs}
through the same reasoning as in the proof of Proposition \ref{prop:morphism} or through the universal property \cite[page 29]{Abe} between the pre-Lie product and its universal enveloping algebra from the Guin-Oudom construction.
We can also extend this simultaneous insertion from $z^\alpha \in \mathbf{M}$ to $z^\alpha \in \CM$  through applying the Leibniz rule to $\left(\prod_{j=1}^n\partial_{z_{k_j}}\right)$ in the sense that
\begin{equs} \label{Leibniz_partial_k}
	\partial_{z_k}	\tilde\prod_{\M \, i=1}^m z^{\alpha_i} = \sum_{i=1}^m   \left(\tilde{\prod}_{\M \, j\ne i} z^{\alpha_j}\right)\, \tilde{\bullet}_\M \, \partial_{z_k} z^{\alpha_i}.
\end{equs}
\vspace{6pt}

Since the expression of the simultaneous product $\star_\M$ is explicit, we can derive the formula of its dual coproduct $\Delta_\M$ which is the reduced extraction-contraction coproduct directly through the adjoint relation
\begin{equs} \label{dual_Delta}
	\langle	\tilde{\prod}_{\M \, i=1}^n z^{\beta_i} \star_\M z^{\bar{\beta}}, z^{\beta} \rangle 
	= \langle	\tilde{\prod}_{\M \, i=1}^n z^{\beta_i} \otimes z^{\bar{\beta}}, \Delta_\M z^{\beta} \rangle. 
\end{equs}
\begin{proposition} \label{LOT_coproduct2}
The explicit formula of the reduced extraction-contraction multi-indices coproduct $\Delta_\M $ is 
\begin{equs} \label{explicit_formula_extraction_1}
	\Delta_\M z^{\beta}  =&
	\sum_{\tilde{\prod}_{\M \, i=1}^n z^{\beta_i} \in \CM}
	 \sum_{z^\alpha \in \mathbf{M}}
	E(\tilde{\prod}_{\M \, i=1}^nz^{\beta_i}, z^\alpha ,z^{\beta})  
	\tilde{\prod}_{\M \, i=1}^n z^{\beta_i}
	\otimes   z^\alpha, \quad\quad	
\end{equs}
with
\begin{equs}
	&E(\tilde{\prod}_{\M \, i=1}^nz^{\beta_i}, z^\alpha ,z^{\beta}) 
	\\&=   
	\sum_{k_1,...,k_n\in \mathbb{N}}
	\sum_{\beta  = \hat{\beta}_1 + \cdots + \hat{\beta}_n+\hat\alpha} 
	\frac{S_\M(z^\beta)}{S_\M( \tilde{\prod}_{\M \, i=1}^n z^{\beta_i})S_\M(z^\alpha)}
	\frac{
		\langle \prod_{i=1}^n\partial_{z_{k_i}} z^{\alpha} , z^{\hat{\alpha}} \rangle}
	{S_\M(z^{\hat{\alpha}})}
	\prod_{i=1}^n
	\frac{\langle  
		D^{k_i}z^{\beta_i}, z^{\hat{\beta}_i} \rangle }{S_\M(z^{\hat{\beta}_i})},
\end{equs}
where there is an order among $\hat{\beta}_1, \cdots , \hat{\beta}_n$ as we have order among $k_1, \ldots, k_n$. For example, if $\beta = [1,2]$, we need to count both $\beta = [1,0]+[0,1]+[0,1]$ and $\beta = [0,1]+[1,0]+[0,1]$. 
Note that $ n \ge 1$ as we cannot insert the empty multi-indice.	
\end{proposition}
\begin{proof}
	We start with a general form which covers all $\CM \otimes \mathbf{M}$ and assume
	\begin{equs} 
		\Delta_\M z^{\beta}  =
		\sum_{\tilde{\prod}_{\M \, i=1}^n z^{\beta_i} \in \CM}
		 \sum_{ z^{\alpha} \in \mathbf{M}}
		E(\tilde{\prod}_{\M \, i=1}^nz^{\beta_i}, z^{\alpha} ,z^{\beta})  
		\tilde{\prod}_{\M \, i=1}^n z^{\beta_i}
		\otimes   z^{\alpha}, 
	\end{equs}
	one can regard  $\tilde{\prod}_{\M \, i=1}^n z^{\beta_i}$ as the forest simultaneously inserted into $z^{\alpha}$ through $\star_\M$ and has
	\begin{equs}
		\langle  \tilde{\prod}_{\M \, i=1}^n z^{\beta_i} \otimes z^{\alpha} ,	\Delta_\M z^{\beta} \rangle 
		= 
		S_\M( \tilde{\prod}_{\M \, i=1}^n z^{\beta_i}) S_\M(z^{\alpha})  	
		E(\tilde{\prod}_{\M \, i=1}^nz^{\beta_i}, z^{\alpha} ,z^{\beta}) .
	\end{equs}
	On the dual side, one has
	\begin{equs}
		&\tilde{\prod}_{\M \, i=1}^n z^{\beta_i} \star_\M z^{\alpha} 
		\\
		=& \sum_{z^\beta \in \mathbf{M}}
		\sum_{k_1,...,k_n\in \mathbb{N}}  
		\frac
		{\left\langle 
		\left(\prod_{i=1}^nD^{k_i}z^{\beta_i}\right)\left[\left(\prod_{i=1}^n\partial_{z_{k_i}}\right)z^\alpha\right] ,	 z^{\beta} \right\rangle }
		{S_\M(z^\beta)}
		{z^\beta}
		\\ = &
		\sum_{z^\beta \in \mathbf{M}}
		\sum_{\beta  = \hat{\beta}_1 + \cdots + \hat{\beta}_n + \hat{\alpha}}
		\sum_{k_1,...,k_n\in \mathbb{N}}  
		\\&
		\frac
		{\left\langle 
		\left(\prod_{i=1}^nD^{k_i}z^{\beta_i}\right)\left[\left(\prod_{i=1}^n\partial_{z_{k_i}}\right)z^\alpha\right] ,	z^{\hat\alpha} \prod_{i=1}^n z^{\hat\beta_i} \right\rangle }{S_\M(z^{\hat\alpha} \prod_{i=1}^n z^{\hat\beta_i})}
		{z^{\hat\alpha} \prod_{i=1}^n z^{\hat\beta_i}}
		\\ = &
		\sum_{z^\beta \in \mathbf{M}}
		\sum_{\beta  = \hat{\beta}_1 + \cdots + \hat{\beta}_n + \hat{\alpha}}
		\sum_{k_1,...,k_n\in \mathbb{N}}	
		\frac{
			\langle \prod_{i=1}^n\partial_{z_{k_i}} z^{\alpha} , z^{\hat{\alpha}} \rangle}
		{S_\M(z^{\hat{\alpha}})}
		\prod_{i=1}^n
		\frac{\langle  
			D^{k_i}z^{\beta_i}, z^{\hat{\beta}_i} \rangle }{S_\M(z^{\hat{\beta}_i})}
		  z^{\hat\alpha}\prod_{i=1}^nz^{\hat\beta_i}
	\end{equs}
	which yields
	\begin{equs}
		&\frac{\langle  \tilde{\prod}_{\M \, i=1}^n z^{\beta_i} \star_\M z^{\alpha}  ,	z^{\beta} \rangle}{S_\M(z^\beta)}
		\\=& 
		\sum_{\beta  = \hat{\beta}_1 + \cdots + \hat{\beta}_n + \hat{\alpha}}
		\sum_{k_1,...,k_n\in \mathbb{N}}	
		\prod_{i=1}^n
		\frac{\langle  
			D^{k_i}z^{\beta_i}, z^{\hat{\beta}_i} \rangle }{S_\M(z^{\hat{\beta}_i})}
		\frac{
			\langle \prod_{i=1}^n\partial_{z_{k_i}} z^{\alpha} , z^{\hat{\alpha}} \rangle}
		{S_\M(z^{\hat{\alpha}})}  z^\beta.
	\end{equs}
	The equality from the duality
	\begin{equs}
		\langle \tilde{\prod}_{\M \, i=1}^n z^{\beta_i}\otimes \prod_{i=1}^n z_{k_i}  ,	\Delta_\M z^{\beta} \rangle & = 
		\langle  \tilde{\prod}_{\M \, i=1}^n z^{\beta_i} \star_1 \prod_{i=1}^n z_{k_i}  ,	z^{\beta} \rangle
	\end{equs}
	allows us to conclude.
\end{proof}
\begin{remark}
	The proof is a generalisation of the proof of \cite[Proposition 3.8]{BH24} as we do not require the symmetry factor to be multiplicative and do not have the restriction on the number of multi-indices in the forest here.
\end{remark}
The extraction-contraction coproduct is its reduced version adding the primitive terms
\begin{equs}
	\Delta^{\! -}_\M z^{\beta} = \emptyset_\M \otimes z^\beta + z^\beta \otimes \emptyset_\M + \Delta_\M z^{\beta}.
\end{equs}
The adjoint dual of $\Delta^{\! -}_\M$ denoted by $\bar \star_\M$ is the extension of $\star_\M$ by
defining: $\bar \star_\M = \star_\M$ if there is no empty forest involved, otherwise
\begin{equs}
	\emptyset_\M \bar \star_{\M} z^\alpha = z^\alpha \quad \text{ and } \quad 
	\tilde{\prod}_{\M \, i=1}^n z^{\beta_i} \bar \star_{\M} \emptyset_\M = \tilde{\prod}_{\M \, i=1}^n z^{\beta_i}.
\end{equs}

\section{Renormalisation through multi-indices}\label{sec:BPHZ}
In this section, we will introduce the renormalisation via multi-indices coproduct and show how it is equivalent to the Feynman diagram-based BPHZ renormalisation. Finally, we will give the explicit formula of the renormalised measure obtained from the BPHZ renormalised cumulant expansion of the original measure.

\subsection{Symmetry factors of Feynman diagrams } \label{subsec:symmetry}
Recall that we defined the  inner product in \eqref{inner_product_diagram} as
\begin{equs}
	\langle\Gamma_1, \Gamma_2\rangle := \delta^{\Gamma_1}_{\Gamma_2}S_\F(\Gamma_1),
\end{equs}
which means that the symmetry factor is important in the adjoint relation for describing the duality of coproduct and its corresponding product. Meanwhile it commutes multi-indices with Feynman diagrams through the adjoint relation between maps $\Phi$ and $\CP$. In this subsection, we will discuss what the suitable symmetry factors of Feynman diagrams are and how they are related to the symmetry factors of multi-indices.

We choose the symmetry factors that make the 
map $\CP$  lift multi-indice $z^\beta$ to all possible Feynman diagrams formed by pairing half-edges encoded by $z^\beta$, i.e., 
\begin{equs} \label{eq:P_map_function}
	\CP(z^\beta) = \sum_{\Phi(\Gamma) = z^\beta} \frac{S(z^\beta)}{S(\Gamma)} \Gamma = \sum_{\Phi(\Gamma) = z^\beta} N(\Gamma) \Gamma
\end{equs}
holds.
It turns out that the symmetry factor of Feynman diagrams in the literature \cite{CK2, VS2007}, which is defined as the size of their automorphism group is suitable.
\begin{definition} (Automorphism group of Feynman diagrams)\label{def:auto_group}
\\
An automorphism of a Feynman diagram $\Gamma \in \mathbf{F}$ is obtained by permuting its edges and vertices such that the result Feynman diagram is isomorphic to $\Gamma$ and it keeps the endpoints of all its edges. This is equivalent to saying that we have two types of group actions $g_{\CV}$ permuting the vertices and $ g_{\CE}$ permuting edges such that 
\begin{equs}
\Gamma \cong	\Gamma' := g_{\CV} \circ g_{\CE}( \Gamma) 
\end{equs}
and for any $e \in \CE(\Gamma)$
\begin{equs}
	\{g_{\CV} (v_{e_{+}}) , g_{\CV}(v_{e_{-}}) \}= \{v_{g_{\CE}(e)_{+}},  v_{g_{\CE}(e)_{-}}\}
\end{equs}
where $v_{e_{+}}, v_{e_{-}}$ are two endpoints of the edge $e$ and there is no order between two endpoints in the set which means
\begin{equs}
	\{g_{\CV} (v_{e_{+}}) , g_{\CV}(v_{e_{-}}) \} = 
		\{g_{\CV} (v_{e_{-}}) , g_{\CV}(v_{e_{+}}) \}
\end{equs}
The automorphism group $\mathrm{Aut}(\Gamma)$ is the set of all the automorphisms $g_{\CV} \circ g_{\CE}$ of $\Gamma$.
\end{definition}

Then the symmetry factor of a Feynman diagram $\Gamma \in \mathbf{F}$ is defined as the cardinal
\begin{equs}
	S_\F(\Gamma) := |\mathrm{Aut}(\Gamma)|.
\end{equs}
It can be generalised to the symmetry factor of a forest of Feynman diagrams by
\begin{equs}
	S_\F(\tilde\prod_{\F \, i=1}^m \Gamma_i^{\tilde{\bullet_{\F}}r_i}) = \prod_{i=1}^m r_i! S_\F(\Gamma_i)^{r_i}
\end{equs}
where $\Gamma_i$ are distinct and $r_i$ is the number of Feynman diagrams isomorphic to $\Gamma_i$ in the forest. The definition of the map $\CP$ can be generalised accordingly to forests as following: 
\begin{equs}
	& \CP: \CM \rightarrow \langle \CF \rangle
	\\&
	\CP(\tilde \prod_{\M \, i = 1}^n z^{\beta_i}) := \sum_{\tilde \prod_{\F \, i = 1}^n \Gamma_i: \Phi(\Gamma_i) = z^{\beta_i}} \frac{S_\M(\tilde \prod_{i = 1}^n z^{\beta_i})}{S_\F(\tilde \prod_{ i = 1}^n \Gamma_i)} \tilde \prod_{\F \, i = 1}^n \Gamma_i.
\end{equs}
Then by the definitions of symmetry factors  of both the Feynman diagram and multi-indices one can verify that Proposition \ref{prop:injective_P} and \ref{prop:adjoint_P_Phi} hold for this generalisation.
\vspace{6pt}

Since Feynman diagrams are formed when one tries to pair half-edges of pre-Feynman diagrams which are multi-indices. The symmetry factor of multi-indices can be described as the cardinal of the following defined ``isomorphism group".
\begin{definition} (Isomorphism group of Feynman diagrams)\label{def:iso_group}
	\\
	In an isomorphism, an edge should be always viewed as a pairing of two half-edges, one attached to each endpoint. 
	An isomorphism of Feynman diagrams $\Gamma \in \mathbf{F}$ is obtained by two types of group actions $h_{\CL}$ permuting half-edges  and $h_{\CV}$ permuting vertices such that  
	the Feynman diagram obtained after the permutation is isomorphic to $\Gamma$
	\begin{equs}
		\Gamma  \cong	\Gamma'' := h_{\CV} \circ h_{\CL}( \Gamma) .
	\end{equs}
	and the half-edges should always be attached to the same vertex as before the permutation.
	The isomorphism group $\mathrm{Iso}(\Gamma)$ is the set of all the isomorphisms $h_{\CV} \circ h_{\CL}$ of $\Gamma$.
\end{definition}

From this definition, one can see that all the isomorphisms can be obtained by permuting the half-edges attached to the same vertex and then
permuting vertices with the same arity. During the permutation of vertices, the half-edges should be permuted accordingly such that it is attached to the same vertex as it was before the vertices permutation. Therefore, the symmetry factor of multi-indices is the cardinal of the isomorphism group of Feynman diagrams. 

The quotient of the isomorphism group by the automorphism group has the property stated in the following proposition, which ensures that $\CP$ satisfies \eqref{eq:P_map_function}.
\begin{proposition} \label{prop:Sta_Orb}
For a Feynman diagram $\Gamma \in \mathbf{F}$, its symmetry factor and the symmetry factor of multi-indice $\Phi(\Gamma)$ satisfy
	\begin{equs}
		S_\M(\Phi(\Gamma)) = N(\Gamma)S_\F(\Gamma),
	\end{equs}
	where $N(\Gamma)$ is the number of distinct pairings of half-edges in the pre-Feynman diagram $\Phi(\Gamma)$ that can form Feynman diagrams isomorphic to $\Gamma$. As a result, $\CP$ satisfies \eqref{eq:P_map_function}.
\end{proposition}
\begin{proof}
	We denote a group action in $\mathrm{Iso}(\Gamma)$ as $h(\Gamma): = h_{\CV} \circ h_{\CL}( \Gamma)$.
	Then $S_\F(\Gamma)$ is the size of the stabiliser that keeps the pairing of half-edges in the pre-Feynman diagram of $\Gamma$
	\begin{equs}
		 & \mathrm{Stab}(\Gamma): = \{h(\Gamma) \in \mathrm{Iso}(\Gamma): h(\Gamma) \in  \mathrm{Aut}(\Gamma)\},
		 \\&
		 S_\F(\Gamma) = |\mathrm{Stab}(\Gamma)|.
	\end{equs}
	Meanwhile, $N(\Gamma)$ is the cardinal of the quotient $\mathrm{Iso}(\Gamma) / \mathrm{Aut}(\Gamma)$ and therefore it is the length of the orbit corresponding to the stabiliser $\mathrm{Stab}(\Gamma)$.
	\begin{equs}
		 &\mathrm{Orb}(\Gamma):= \{ h(\Gamma): h(\Gamma) \in \mathrm{Iso}(\Gamma) / \mathrm{Aut}(\Gamma)\},
		 \\&
		 N(\Gamma) =|\mathrm{Orb}(\Gamma)|.
	\end{equs}
	By the Orbit-Stabilizer theorem, 
	\begin{equs}
		S_\M(\Phi(\Gamma)) = N(\Gamma)S_\F(\Gamma).
	\end{equs}
\end{proof}
\begin{remark}
This proposition is equivalent to \cite[Eq. 68]{Faris}, where the order of the so-called ``stabilizer group" is exactly the symmetry factor of multi-indices. 
\end{remark}
Let us illustrate the automorphism group, isomorphism group, and Proposition \ref{prop:Sta_Orb} by the following example.
\begin{example}
One simple example is $\Gamma = \FGIII$. We index its vertices and edges in the following way to describe the permutations:
	\begin{equs}
		\Gamma = \FGIIIa.
	\end{equs}
The automorphism can be obtained by permuting the two vertices $v_1, v_2$ and permuting the three edges (red, blue, green)
as they do not change the pairing of half-edges and keep the endpoints of each edge. Therefore 
\begin{equs}
	S_\F(\Gamma) = 2!3! = 12.
\end{equs}
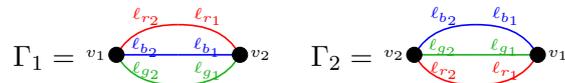
\begin{equs}
\Gamma_1 = 	 \FGIIIc \quad\quad \Gamma_2 =  \FGIIIb
\end{equs}
For example, $\Gamma_1$ and $\Gamma_2$ are obtained from automorphisms of $\Gamma$. $\Gamma_1$ is the result from permuting vertices $v_1$ and $v_2$. $\Gamma_2$ can be achieved by permuting vertices $v_1$ and $v_2$ and also permuting three edges (red edge $\mapsto$ green edge, green edge $\mapsto$ blue edge, blue edge $\mapsto$ red edge). Meanwhile, $\Gamma_1$ and $\Gamma_2$ can also be obtained from isomorphisms of $\Gamma$. For example, $\Gamma_1$ can be recovered by permuting vertices $v_1$ and $v_2$ and requiring that the half-edges are attached to the vertex where it was connected to in $\Gamma$, say $\ell_{r_1}, \ell_{b_1}, \ell_{g_1}$ are still attached to $v_1$ after permutation.
$\Gamma_2$ can be viewed as swapping two vertices and then permuting half-edges attached to the same vertex  but still keep the pairing $\{\ell_{r_1}, \ell_{r_2}\}, \{\ell_{b_1}, \ell_{b_2}\}, \{\ell_{g_1}, \ell_{g_2}\}$. For example ``$\ell_{b_1} \mapsto \ell_{r_1}$ and $\ell_{b_2} \mapsto \ell_{r_2}$" is exactly ``blue edge $\mapsto$ red edge" as the pairing is preserved.

Therefore, the symmetry factor of $\Phi(\Gamma) = z_3^2$ is $3!3!2! = 72$ and the cardinal of the quotient can be calculated by finding all the possible pairings of half-edges that forms Feynman diagrams isomorphic to $\Gamma$. In this case, it is $N(\Gamma) = 3! = 6$ as we have to pair the three half-edges of $v_1$ with those of $v_2$. 
For example,
	\begin{equs}
		\Gamma_3 = \FGIIId
	\end{equs}
can be obtained by an isomorphism  swapping $\ell_{g_1}$ and $\ell_{r_1}$. However, it is not automorphic to $\Gamma$ as the pairing is broken and $\Gamma_3$ is obtained from a different coset as $\Gamma, \Gamma_1,\Gamma_2$. 
\vspace{6pt}

Finally one can check
\begin{equs}
	N(\Gamma)S_\F(\Gamma)  = 72 = S_\M(\Phi(\Gamma)).
\end{equs}
\end{example}

We now reconsider the pairing problem \eqref{eq:pairing_problem} with the Hermite polynomials $H_{k}$.
Accordingly we
define the valuation map of multi-indices as 
\begin{equs} \label{eq:valuation_M}
		\Pi_\M(\tilde{\prod}_{i=1}^n z^{\beta_i}) := & \prod_{i=1}^n \Pi_\M z^{\beta_i},
		\\
		\Pi_\M(\prod_{k \in \mathbb{N}} z_k^{\beta(k)}): = &\mathbb{E}\left[\prod_{k \in \mathbb{N}} \left(\int_\Lambda H_{k}(X(x_k)) dx_k\right)^{\beta(k) }\right]
		\\
		-& 
			\sum_{n \ge 2} 
		\sum_{\tiny \substack{\CM \ni \tilde{z}^{\tilde\beta} = \tilde\prod_{\M \, i=1}^n z^{\beta_i} \\ \sum_{i=1}^n \beta_i = \beta} }
		\prod_{k \in \mathbb{N}} {\beta(k) \choose \beta_1(k), \ldots, \beta_n(k)}
		\prod_{i=1}^n  \Pi_\M( z^{\beta_i}  )
	\end{equs}
where the sum $\sum_{\tiny \substack{\CM \ni \tilde{z}^{\tilde\beta} = \tilde\prod_{\M \, i=1}^n z^{\beta_i} \\ \sum_{i=1}^n \beta_i = \beta} }$ runs over all possible forest whose frequencies of vertices match $\beta$. Note that we sum over forests instead of $z^{\beta_i}$. Therefore, there is no order among $z^{\beta_i}$.
\begin{corollary} \label{coro:commute} For any $z^\beta \in \mathbf{M}$ the following equality holds
\begin{equs}
	\Pi_\M(z^\beta) =  \Pi_\F \circ \CP (z^\beta).
\end{equs}
\end{corollary}

\begin{proof}
		Recall that we compute the expectation in the following way
	\begin{equs}
		\mathbb{E}\left[\prod_i \int_{\Lambda}H_{k_i}(X(x_i)) dx_i\right] = \sum_{F \in \CF} N(F)\Pi_{\F} \left( F\right).
	\end{equs}
	Therefore, the proof is equivalent to show the following diagram commutes
		\begin{equs}
		\begin{tikzcd}
			&\mathbb{E}\left[\prod_i \int_{\Lambda}H_{k_i}(X(x_i)) dx_i\right]
			& \langle \CM \rangle  \arrow[l, "\Pi_\M"]  \arrow[r, "\mathcal{P}"]  
			& \langle  \CF \rangle  \arrow[ll, bend left=60, "\Pi_{\F}"] 
		\end{tikzcd}
		.
	\end{equs}
Since $\Pi_\M$ and $\Pi_\F$ preserve their corresponding forest products, it boils down to show
\begin{equs}\label{eq:prove}
\CP \left( z^\beta  + \sum_{n \ge 2} 
	\sum_{\tiny \substack{\CM \ni \tilde{z}^{\tilde\beta} = \tilde\prod_{\M \, i=1}^n z^{\beta_i} \\ \sum_{i=1}^n \beta_i = \beta} }
	\prod_{k \in \mathbb{N}} {\beta(k) \choose \beta_1(k), \ldots, \beta_n(k)}
	\tilde{z}^{\tilde{\beta}}  \right) = \sum_{F \in \CF} N(F) F . \quad
\end{equs}

	For any forest of multi-indices $\tilde{z}^{\tilde{\beta}} = \tilde \prod_{j=1}^{m} ({z^{\beta_j}})^{\tilde\bullet_\M r_j}$ with distinct $z^{\beta_j}$,
\begin{equs}
	\CP (\tilde{z}^{\tilde{\beta}}) = \sum_{\tilde \prod_{\F \, i = 1}^n \Gamma_i: \Phi(\Gamma_i) = z^{\beta_i}} 
	\frac{S_\M(\tilde \prod_{\M \, j=1}^{m} (z^{\beta_j})^{\tilde\bullet_\M r_j})}{S_\F(\tilde \prod_{\F \, i = 1}^n \Gamma_i)} 
	\tilde \prod_{\F \, i = 1}^n \Gamma_i.
\end{equs}
Rewriting the forest $\tilde \prod_{\F \, i = 1}^n \Gamma_i$ as
\begin{equs}
	\tilde \prod_{\F \, i = 1}^n \Gamma_i = \tilde \prod_{\F \, l = 1}^q \Gamma_l^{\tilde\bullet_\F R_l}
\end{equs}
for distinct $\Gamma_l$ leads to 
\begin{equs}
	\frac{S_\M(\tilde \prod_{\M \, j=1}^{m} (z^{\beta_j})^{\tilde\bullet_\M r_j})}{S_\F(\tilde \prod_{\F \, i = 1}^n \Gamma_i)} 
	\tilde \prod_{\F \, i = 1}^n \Gamma_i
	=\frac{\prod_{j=1}^m r_j ! S_\M(z^{\beta_j})^{r_j}}{\prod_{l=1}^q R_l ! \prod_{i=1}^n S_\F(\Gamma_i)} \tilde \prod_{\F \, l = 1}^q \Gamma_l^{\tilde\bullet_\F R_l}.
\end{equs}
By Proposition \ref{prop:Sta_Orb}, 
\begin{equs} \label{eq:individaul_pairing}
	\frac{\prod_{j=1}^m S_{\M}(z^\beta)^{r_i}}{\prod_{i=1}^n S_{\F}(\Gamma_i)} = \prod_{i=1}^n N(\Gamma_i).
\end{equs}
counts the number of pairings of each individual pre-diagram. Then, we have to calculate the repetition in the symmetry of the forest. 
We  define a index set $\CJ_j$ such $\Phi(\Gamma_l) = z^{\beta_j}$ for any $l \in \CJ_j$.
Then we have
\begin{equs}
	\sum_{l \in \CJ_j} R_l = r_j
\end{equs}
and thus
\begin{equs}
	\frac{\prod_{j=1}^m r_j ! }{\prod_{l=1}^q R_l ! } = \prod_{j=1}^m {r_j \choose \{R_l\}_{l \in \CJ_j}}
\end{equs}
which counts different ways to assign $r_j z^{\beta_j}$ to Feynman diagrams whose images of the counting map are $z^{\beta_j}$.

	Let us then analyse the right-hand side of \eqref{eq:prove}. While forming a forest of Feynman diagrams from $\mathbb{E}\left[\prod_i \int_{\Lambda}H_{k_i}(X(x_i)) dx_i\right]$, There are two steps:  Firstly partition vertices into groups. Secondly, form individual connected Feynman diagrams using vertices in each group respectively. The classification of vertices is essentially computed by
	\begin{equs}
 \sum_{n \ge 1} 
	\sum_{\tiny \substack{\CM \ni \tilde{z}^{\tilde\beta} = \tilde\prod_{\M \, i=1}^n z^{\beta_i} \\ \sum_{i=1}^n \beta_i = \beta} }
	\sum_{\Phi(\tilde \prod_{\F \, l = 1}^q \Gamma_l^{\tilde\bullet_\F R_l}) = \tilde{z}^{\tilde\beta}}
	\prod_{k \in \mathbb{N}} {\beta(k) \choose \beta_1(k), \ldots, \beta_n(k)} \prod_{j=1}^m {r_j \choose \{R_l\}_{l \in \CJ_j}}.
	\end{equs}
	The second step is achieved by \eqref{eq:individaul_pairing}. Hence, we can conclude the proof.
\end{proof}
\begin{remark}
	This corollary can be generalised to other pairing-half-edges problems. 
\end{remark}

We finish this subsection by showing that the reduced extraction-contraction coproduct $\Delta_{\F}$ is adjoint to the product $\star_\F$ under the inner product defined in \eqref{inner_product_diagram}.
\begin{theorem} \label{them: adjoint_Feyman}
For any  $\tilde{\prod}_{\F \, i=1}^n \Gamma_i \in \CF_{-}$ and any $\Gamma, \bar\Gamma \in \mathbf{F} $,
	\begin{equs}
		\langle \Delta_\F \Gamma , \tilde{\prod}_{\F \, i=1}^n \Gamma_i  \otimes \bar\Gamma \rangle = 	\langle \Gamma, \tilde{\prod}_{\F \, i=1}^n \Gamma_i  \star_\F \bar\Gamma \rangle .
	\end{equs}

\end{theorem}

\begin{proof}
The main idea is firstly combining and 
adapting the proofs of Lemma 12 and Lemma 14 in \cite{VS2007} to show the adjoint relation between a single insertion and a single extraction-contraction, and
then using the property of Guin--Oudom construction to generalise it to the simultaneous insertion and the reduced extraction-contraction coproduct.
Since the insertion defined in \cite{VS2007} is different from ours, necessary modifications have to be done.
\vspace{6pt}

We start with
\begin{equs}
	\langle \bar\Delta_\F \Gamma , \gamma \otimes \bar\Gamma \rangle = 	\langle \Gamma, \Gamma_1 \triangleright \bar\Gamma \rangle 
\end{equs}
as when $n=1$ in the forest, $\star_{\F}$ is reduced to $\triangleright$, and $\bar\Delta_\F$ is obtained by projecting result of $\Delta_\F$ to $\mathbf{F}_{-} \otimes \mathbf{F}$. Suppose now we have one 
$\Gamma= \gamma \triangleright_{(v,\phi)} \bar\Gamma$ where $v$ denotes the vertex which is the insertion spot in $\bar \Gamma$ and $\phi$ is the map sending vertices in $\gamma$ to  subsets of free legs in $C_v(\bar\Gamma)$. These subsets together form a partition of $C_v(\bar\Gamma)$. It amounts to say for any vertex $u \in \CV(\gamma)$
\begin{equs}
	\phi(u) = \{\ell_1,...,\ell_{m_u}\}
\end{equs}
where $\ell_i \in C_v(\bar\Gamma)$ for $i =1, \ldots, m_u$ and $\{\phi(u)\}_{u \in \CV(\gamma)}$ is a partition of $C_v(\bar\Gamma)$. It can be seen that 
the sum of all possible $\phi$ gives the operator $\CG(C_v(\bar\Gamma), \gamma)$ in the insertion defined in \eqref{eq:def_insertion}.
\vspace{6pt}
	
Let $M(v,\phi)$ be the number of $\gamma'$ that are images of $\gamma$ under some element in $\mathrm{Aut}(\gamma \triangleright_{(v,\phi)} \bar\Gamma)$.
Then we have the size of the automorphism that sends $\gamma$ to itself
	\begin{equs} \label{eq:stab_gamma}
		|\mathrm{Aut}(\gamma \triangleright_{(v,\phi)} \bar\Gamma)_{\gamma}| = \frac{|\mathrm{Aut}(\gamma \triangleright_{(v,\phi)} \bar\Gamma)| }{M(v,\phi)} = \frac{S_\F(\gamma \triangleright_{(v,\phi)} \bar\Gamma) }{M(v,\phi)}
	\end{equs}
	because $M(v,\phi)$ is the length of the orbit and $\mathrm{Aut}(\gamma \triangleright_{(v,\phi)} \bar\Gamma)_{\gamma}$ is the stabiliser. Let us rewrite the reduced extraction-contraction coproduct for  $ \Gamma \cong \gamma \triangleright_{(v,\phi)} \bar\Gamma$ as
		\begin{equs}
		\bar\Delta_\F ( \Gamma) = \sum_{\gamma \in \mathbf{F}_{-}}\sum_{\bar \Gamma \in \mathbf{F}}
		n(\gamma, \bar\Gamma,\Gamma) \gamma \otimes \bar\Gamma
	\end{equs}
where $n(\gamma, \bar\Gamma,\Gamma)$ is the number of extraction of $\gamma'$ from $\Gamma$ such that $\gamma'$ is isomorphic to $\gamma$, and the remaining part $\bar\Gamma'$ is isomorphic to $\bar\Gamma$. Let $(v,\phi)$ denote an equivalent class for any $\hat v, \hat \phi$  that after inserting $\hat \gamma \cong \gamma$ at $\hat v$ in $\hat \Gamma \cong \bar \Gamma$ through $\hat \phi$, $\hat \gamma$ is the image of $\gamma$ under some element in $\mathrm{Aut}(\Gamma)$. Then $M(v,\phi)$ is the size of the class $(v,\phi)$.
Since the sum of the length of all different orbits gives the size of the group, we have
\begin{equs}
		n(\gamma, \bar\Gamma, \Gamma) = \sum_{(v',\phi')} M(v',\phi')
\end{equs}
where the sum runs over all equivalent classes such that $\gamma \triangleright_{(v',\phi')} \bar\Gamma \cong \Gamma$. 
\vspace{6pt}

	Since the permutation in the automorphism $\mathrm{Aut}(\gamma \triangleright_{(v,\phi)} \bar\Gamma)_{\gamma}$ is either inside $\gamma$ or not,
	the quotient group
	\begin{equs}
		\mathrm{Aut}(\gamma \triangleright_{(v,\phi)} \bar\Gamma)_{\gamma} / \mathrm{Aut}(\gamma)_{(v,\phi)} \simeq \mathrm{Aut}(\bar\Gamma)_{(v,\phi)}
	\end{equs}
	where $\mathrm{Aut}(\gamma)_{(v,\phi)}$ is the intersection $\mathrm{Aut}(\gamma) \bigcap \mathrm{Aut}(\gamma \triangleright_{(v,\phi)} \bar\Gamma)_{\gamma}$ and $\mathrm{Aut}(\bar\Gamma)_{(v,\phi)}$ is the subgroup of $\mathrm{Aut}(\bar\Gamma)$ such that $v$ is mapped to itself and for edges attached to $v$ only those that have the same pre-image in $\phi$ can be permuted.
Therefore, by the orbit-stabiliser theorem
\begin{equs}
	|\mathrm{Aut}(\gamma \triangleright_{(v,\phi)} \bar\Gamma)_{\gamma}| = 
	\frac{S_\F(\gamma)S_\F(\bar\Gamma)}{|\mathrm{Aut}(\gamma)[(v,\phi) ]| |\mathrm{Aut}(\bar\Gamma)[(v,\phi)]|}
\end{equs}
where $\mathrm{Aut}(\mathrm{T})[\cdot]$ is the orbit of $\mathrm{Aut}(\mathrm{T})_{\cdot}$ for any graph $\mathrm{T}$. Moreover,
from \eqref{eq:stab_gamma} we have
\begin{equs}
	\frac{S_\F(\gamma)S_\F(\bar\Gamma)M(v,\phi)}{S_\F(\gamma \triangleright_{(v,\phi)} \bar\Gamma)}
	= {|\mathrm{Aut}(\gamma)[(v,\phi) ]| |\mathrm{Aut}(\bar\Gamma)[(v,\phi)]|}.
\end{equs}
Consequently we have 
\begin{equs} 
	n(\gamma, \bar \Gamma, \Gamma \cong \gamma \triangleright_{(v,\phi)}\bar\Gamma) &= \sum_{(v',\phi')} M(v',\phi')
	\\& = 
	\sum_{(v',\phi')}
	\frac{|\mathrm{Aut}(\gamma)[(v',\phi') ]| |\mathrm{Aut}(\bar\Gamma)[(v',\phi')]|S_\F(\gamma \triangleright_{(v,\phi)} \bar\Gamma)}{S_\F(\gamma)S_\F(\bar\Gamma)}
	\\& = 
	\frac{S_\F(\gamma \triangleright_{(v,\phi)} \bar\Gamma)}{S_\F(\gamma)S_\F(\bar\Gamma)} \sum_{(v',\phi')} |\mathrm{Aut}(\gamma)[(v',\phi') ]| |\mathrm{Aut}(\bar\Gamma)[(v',\phi')]|
\end{equs}
Finally, since $|\mathrm{Aut}(\gamma)[(v,\phi) ]| |\mathrm{Aut}(\Gamma)[(v,\phi)]|$ gives the number of  ways inserting $\gamma$ to $\bar\Gamma$ such that the result $\Gamma \cong \gamma \triangleright_{(v,\phi)} \bar\Gamma$, one has
\begin{equs} \label{eq:mn}
	n(\gamma, \bar\Gamma, \Gamma \cong \gamma \triangleright_{(v,\phi)}\bar\Gamma) 
	= 
	\frac{S_\F(\gamma \triangleright_{(v,\phi)} \bar\Gamma)}{S_\F(\gamma)S_\F(\bar\Gamma)} 
	m(\gamma, \bar\Gamma,  \Gamma \cong \gamma \triangleright_{(v,\phi)}\bar\Gamma),
\end{equs}
where $m(\gamma, \bar\Gamma,  \Gamma \cong \gamma \triangleright_{(v,\phi)}\bar\Gamma)$ is the number of $\Gamma \cong \gamma \triangleright_{(v,\phi)}\bar\Gamma$ one can get by inserting $\gamma$ to $\bar\Gamma$.

Then we have to show the case of $n > 1$ the following equation is satisfied:
	\begin{equs}
	\langle \Delta_\F \Gamma , \tilde{\prod}_{\F \, i=1}^n \Gamma_i  \otimes \bar\Gamma \rangle = 	\langle \Gamma, \tilde{\prod}_{\F \, i=1}^n \Gamma_i  \star_\F \bar\Gamma \rangle .
\end{equs}
When $\Gamma_i $ are distinct, it is an immediate consequence of \eqref{eq:mn} since
\begin{equs}
	n(\tilde{\prod}_{\F \, i=1}^n \Gamma_i, \bar \Gamma, \Gamma) = \prod_{i=1}^n \sum_{(v_i',\phi_i')} M(v_i',\phi_i'),
\end{equs}
\begin{equs}
	|\mathrm{Aut}(\tilde{\prod}_{\F \, i=1}^n \Gamma_i \star_{\F \prod_{i=1}^{n}(v_i,\phi_i)} \bar\Gamma)_{\tilde{\prod}_{\F \, i=1}^n \Gamma_i}| 
	&= \frac{|\mathrm{Aut}(\tilde{\prod}_{\F \, i=1}^n \Gamma_i \star_{\F \,\prod_{i=1}^{n}(v_i,\phi_i)} \bar\Gamma)| }{\prod_{i=1}^n M(v_i,\phi_i)} 
	\\&
	= \frac{S_\F(\Gamma) }{\prod_{i=1}^n M(v_i,\phi_i)},
\end{equs}
and 
	\begin{equs}
	\mathrm{Aut}(\tilde{\prod}_{\F \, i=1}^n \Gamma_i \star_{\F \,\prod_{i=1}^{n}(v_i,\phi_i)} \bar\Gamma)_{\tilde{\prod}_{\F \, i=1}^n \Gamma_i} / \prod_{i=1}^n \mathrm{Aut}(\Gamma_i)_{(v_i,\phi_i)} \simeq \mathrm{Aut}(\bar\Gamma)_{\prod_{i=1}^n(v,\phi)}
\end{equs}
where $\star_{\F \,\prod_{i=1}^{n}(v_i,\phi_i)}$ is the simultaneous insertion according to pairs $(v_i,\phi_i)$, and $M(v_i, \phi_i)$ is the number of $\Gamma_i'$ which is the image of $\Gamma_i$ under some automorphism $|\mathrm{Aut}(\tilde{\prod}_{\F \, i=1}^n \Gamma_i \star_{\F \, \prod_{i=1}^{n}(v_i,\phi_i)} \bar\Gamma)|$. Similar for the coefficient 
$m(\tilde{\prod}_{\F \, i=1}^n \Gamma_i, \bar \Gamma, \Gamma).$

Since we defined the symmetry factor of forests of Feynman diagrams as
\begin{equs}
	S_\F(\tilde\prod_{\F \, j=1}^q \Gamma_i^{\tilde{\bullet_{\F}}r_j}) = \prod_{j=1}^q r_j! S_\F(\Gamma_i)^{r_j}
\end{equs}
where $r_j!$ gives the permutation coefficient in both $n(\tilde\prod_{\F \, j=1}^q \Gamma_i^{\tilde{\bullet_{\F}}r_j}, \bar \Gamma, \Gamma)$ and $m(\tilde\prod_{\F \, j=1}^q \Gamma_i^{\tilde{\bullet_{\F}}r_j}, \bar \Gamma, \Gamma)$,
the case with repetition is a consequence of the case when $\Gamma_i$ in the forest are distinct.
\end{proof}

\begin{corollary}
	The adjoint dual of $\Delta^{\! -}_\F$ denoted by $\bar \star_\M$ is the extension of $\star_\M$ by
	defining $\bar \star_\M = \star_\M$ when there is no empty forest, and otherwise defining
	\begin{equs}
		\emptyset_\F \bar \star_{\F} \Gamma = \Gamma \quad \text{ and } \quad 
		\tilde{\prod}_{\F \, i=1}^n \Gamma_i \bar \star_{\F} \emptyset_\F = \tilde{\prod}_{\F \, i=1}^n \Gamma_i.
	\end{equs}
\end{corollary}

\begin{remark}
	It can be seen that the adjoint relation will still be valid if we change $\bar\Gamma \in \CF_{-}$ to $\bar\Gamma \in \CF$ in $\Delta_\F$ and change $\tilde{\prod}_{\F \, i=1}^n \Gamma_i \in \CF_{-}$ to $\tilde{\prod}_{\F \, i=1}^n \Gamma_i \in \CF$ in $\star_\F$ at the same time.
\end{remark}


\subsection{BPHZ Renormalisation via multi-indices}

If one wants to use $\star_\F$, $\star_\M$, $\Delta_\F$, and $\Delta_\M$ to implement the BPHZ renormalisation, terms with divergence have to be found using the $\deg$ map. Therefore, from this sub-section on, we clarify that $\star_\F$, $\star_\M$, $\Delta_\F$, $\Delta_\M$, and their full versions satisfy:
\begin{itemize}
	\item $F \star_\F \Gamma = 0$ for any $F \in \CF \setminus \CF_{-}$ ,
	\item  $F  \bar{\star}_{\F}  \Gamma = 0$ if $F \in \CF \setminus \{\CF_{-}, \emptyset_\F\}$ and $\Gamma \ne \emptyset_\F$,
	\item $\tilde{z}^{\tilde{\beta}} \star_\M z^\alpha = 0$ for any $\tilde{z}^{\tilde{\beta}} \in \CM \setminus \CM_{-}$ ,
	\item  $\tilde{z}^{\tilde{\beta}}   \bar{\star}_{\M}  z^\alpha  = 0$ if $\tilde{z}^{\tilde{\beta}} \in \CM \setminus \{\CM_{-}, \emptyset_\M\}$ and $z^\alpha \ne \emptyset_\M$,
	\item The terms $F \otimes \Gamma $ are set to be $0$ in $\Delta_{\F}$ and $\Delta_{\F}^{\!-}$ for any $ F \in \CF \setminus \CF_{-}$ except primitive terms in $\Delta_{\F}^{\!-}$,
	\item The terms $\tilde{z}^{\tilde{\beta}} \otimes z^\alpha $ are set to be $0$ in $\Delta_{\M}$ and $\Delta_{\M}^{\!-}$ for any $\tilde{z}^{\tilde{\beta}} \in \CM \setminus \CM_{-}$ except primitive terms in $\Delta_{\M}^{\!-}$.
\end{itemize}
It can be verified that all the properties and theorems in the previous chapters are still valid as long as this degree restriction described above is put accordingly to each product and its adjoint coproduct.
\vspace{6pt}

Same as for Feynman diagrams, twisted antipode $  \CA_{\M} : \langle\CM_{-}\rangle \rightarrow \langle\CM\rangle $ of multi-indices is defined as 
\begin{equs} \label{eq:def_antipode_MI}
	\begin{aligned}
\	& \CA_\M(\emptyset_\M) := \emptyset_\M,  \quad \CA_{\M}( \tilde z^{\tilde{\alpha}} \tilde{\bullet}_\M \tilde z^{\tilde{\beta}} ) := \CA_{\M}( \tilde z^{\tilde{\alpha}} ) \tilde{\bullet}_\M \CA_{\M}( \tilde z^{\tilde{\beta}} ) \quad \text{ for any } \tilde z^{\tilde{\alpha}}, \tilde z^{\tilde{\beta}} \in \mathcal{M}_-,
	\\&
	\CA_\M( z^{ \beta}) := - z^{\beta} - \mu_{\M} \circ(\CA_\M \otimes \id)\circ\Delta_{\M}  z^{ \beta},
	\quad \text{ for any }  z^{ \beta} \in \mathbf{M}_-.
	\end{aligned}
\end{equs}
For any forests $\tilde{z}^{\tilde \gamma} \in \CM \setminus \CM_{-}$ we say $\CA_{\M}(\tilde{z}^{\tilde \gamma}) = 0$.

For the detection of divergence, Corollary \ref{coro:commute} together with equality \eqref{eq:deg} proves the multi-indices counterpart to Proposition \ref{prop:divergence_degree}. It follows the theorem below.
\begin{theorem} \label{them:main}
The multi-indices renormalisation for $\tilde{z}^{\tilde{\alpha}} \in \CM$
\begin{equs}
	\hat{M}_{\M} (\tilde{z}^{\tilde{\alpha}}) := \left( \Pi_\M \left(\CA_\M(\cdot )\right) \otimes \id  \right) \Delta^{\!-}_{\M} (\tilde{z}^{\tilde{\alpha}})
\end{equs}
is equivalent to the BPHZ renormalisation in the sense that
\begin{equs}
	\Pi_{\M} \circ \hat{M}_\M = \Pi_{\F} \circ \hat{M}_\F\circ \CP
\end{equs}
\end{theorem}
\begin{proof}
	The fact that the twisted antipodes are defined through the reduced extraction-contraction coproduct together with the Corollary~\ref{coro:commute} indicates that it boils down to prove the following theorem.
\end{proof}
	\begin{theorem} \label{them:main2}
	For any multi-indice $z^\beta \in \mathbf{M}$
		\begin{equs}
			\left( \CP \otimes \CP \right) \circ \Delta_\M (z^\beta) = \Delta_\F \circ \CP (z^\beta).
		\end{equs}
	\end{theorem}
	
	\begin{proof}
		We start from the right-hand side and by the duality between $\star_\F$ and $\Delta_{\F}$ in Theorem~\ref{them: adjoint_Feyman} ,
		 for any $\tilde{\prod}_{\F \, i=1}^n \Gamma_i  \in \CF_{-}$ and $\bar\Gamma \in \mathbf{F}$, one has
			\begin{equs}
				\langle \Delta_{\F} \circ\CP(z^\beta) ,  \tilde{\prod}_{\F \, i=1}^n \Gamma_i  \otimes \bar\Gamma \rangle = 	\langle \CP(z^\beta),  \tilde{\prod}_{\F \, i=1}^n \Gamma_i  \star_\F \bar\Gamma \rangle.
		\end{equs}
		By the duality between $\CP$ and $\Phi$ in Proposition~\ref{prop:adjoint_P_Phi} we have
		\begin{equs}
			\langle \CP(z^\beta),\tilde{\prod}_{\F \, i=1}^n \Gamma_i  \star_\F \bar\Gamma \rangle  =	\langle z^\beta, \Phi(\tilde{\prod}_{\F \, i=1}^n \Gamma_i  \star_\F \bar\Gamma) \rangle.
		\end{equs}
		Then we calculate the left-hand side. By the duality between $\CP$ and $\Phi$ we have
		\begin{equs}
			\langle 	\left( \CP \otimes \CP \right) \Delta_\M z^{\beta}  , \tilde{\prod}_{\F \, i=1}^n \Gamma_i  \otimes \bar\Gamma \rangle 
			= \langle \Delta_\M z^{\beta}  , \Phi(\tilde{\prod}_{\F \, i=1}^n \Gamma_i )\otimes \Phi(\bar \Gamma) \rangle 
		\end{equs}
		Then by the duality between $\star_\M$ and $\Delta_\M$ we have
		\begin{equs}
			\langle \Delta_\M z^{\beta}  , \Phi(\tilde{\prod}_{\F \, i=1}^n \Gamma_i )\otimes \Phi(\bar \Gamma) \rangle 
			= 	\langle 	 z^{\beta}  ,\Phi(\tilde{\prod}_{\F \, i=1}^n \Gamma_i) \star_\M \Phi(\bar \Gamma) \rangle 
		\end{equs}
		From Proposition \ref{prop:morphism}, and the universal property \cite[page 29]{Abe} between the pre-Lie product and its universal enveloping algebra from the 
		Guin-Oudom generalisation 
		\begin{equs}
			\langle z^{\beta}  ,\Phi(\tilde{\prod}_{\F \, i=1}^n \Gamma_i) \star_\M \Phi(\bar \Gamma) \rangle 
			= \langle 	 z^{\beta}  , \Phi(\tilde{\prod}_{\F \, i=1}^n \Gamma_i \star_\F \bar \Gamma) \rangle
		\end{equs}
		which allows us to conclude.
	\end{proof}
	
 We illustrate the results from Corollary \ref{coro:commute} and Theorem \ref{them:main} by the following commutative diagram
{\small 
	\begin{equation}	\label{main_diagram}
		\begin{tikzcd} 
			&\mathbb{E}\left[\prod_i \int_{\Lambda}H_{k_i}(X(x_i)) dx_i\right] &&\R_{\BPHZ} \left(\mathbb{E}\text{$\left[\prod_i \int_{\Lambda}H_{k_i}(X(x_i)) dx_i\right] $}\right) 
			\\
			& \langle\CM \rangle \arrow[u, "\Pi_\M"] \arrow[d, "\mathcal{P}"]   \arrow[r, "\Delta_\M^-"] & \langle\CM_{-}\rangle \otimes \langle \CM \rangle
			\arrow[r, "  \mu_{\M} \circ (\CA_{\M} \otimes \id )"] \arrow[d, "\CP \otimes \CP"] \quad 
			& \quad \langle\CM \rangle \arrow[u, "\Pi_{\M}"] \arrow[d, "\CP "] 
			\\
			& \langle\CF \rangle \arrow[r, "\Delta_\F^-"] \arrow[uu, bend left=60, "\Pi_\F"]
			&\langle \CF_{-} \rangle \otimes \langle \CF \rangle  \quad
			\arrow[r, " \mu_\F \circ (\CA_{\F} \otimes \id )"]& \quad \langle \CF \rangle
			\arrow[uu,bend right=60,"\Pi_\F"]
		\end{tikzcd}
	\end{equation}
where $\R_{\BPHZ}$ denotes the BPHZ renormalisation.

Before moving to the main theorem stating the method to find the renormalised measure, let us rigorously prove the Feynman diagram representation of the cumulant expansion of the partition function. 
	\begin{lemma} \label{prop:cumulant_F}
		The cumulant expansion of $Z(\alpha) / Z(0)$ admits the following Feynman diagram representation.
		\begin{equs}
			\log	\mathbb{E}\left[\exp\left(-\int_\Lambda\sum_{k \in \mathbb{N}} \alpha_k H_{k}(X(x)) dx\right)\right]  = \sum_{\Gamma \in \mathbf{F}}\frac{\Upsilon^\alpha_\F (\Gamma)N(\Gamma)}{\hat{S}_\F(\Gamma)} \Pi_\F(\Gamma).
		\end{equs}
		where the coefficients are 
		\begin{equs} 
			\Upsilon^\alpha_\F (\Gamma) = \prod_{k \in \mathbb{N}} (-\alpha_k)^{\beta(\Gamma, k)}, \quad \quad
			\hat{S}_\F(\Gamma) = \prod_{k \in \mathbb{N}} \beta(\Gamma, k)!
		\end{equs}
		where $\beta(\Gamma, k)$ is the frequency of $z_k$ in $\Phi(\Gamma)$.
	\end{lemma}
	\begin{proof}
		By the Linked Cluster Theorem  \cite{connect1,connect2,connect3} (see also \cite[Section 3.3]{connect4} and \cite[Propisition 3.1]{BK23} for the explanation in the Hopf algebra setting), the moment-generating function is obtained by projecting moments in the power series to the valuation of connected Feynman diagrams, i.e.,
		\begin{equs} \label{eq:def_cumulant}
			&\log	\mathbb{E}\left[\exp\left(-\int_\Lambda\sum_{k \in \mathbb{N}} \alpha_k H_{k}(X(x)) dx\right)\right] 
			\\=& 
			\sum_{n=2}^\infty \frac{1}{n!}\proj \left(\mathbb{E}\left[
			\left(\sum_{k \in \mathbb{N}} -\alpha_k \int_{\Lambda}H_{k}(X(x)) dx\right)^n
			\right]\right)
		\end{equs}
		where $\proj$ is the projection.
		Here the connected Feynman diagrams are formed by regarding each $H_k$ as a vertex with $k$ half-edges and by then pairing all half-edges from different vertices in all possible ways. 
		
		Then, by the multinomial theorem, the second line of \eqref{eq:def_cumulant} equals
		\begin{equs}
			&\sum_{n=2}^\infty \frac{1}{n!}
			\sum_{\substack{\{m_k\}_{k \in \mathbb{N}} \\ \sum_{k}m_k = n}}
			{n \choose \{m_k\}_{k \in \mathbb{N}}} \prod_{k \in \mathbb{N}}(-\alpha_k)^{m_k}
			\proj \left(
			\mathbb{E}\left[
			\prod_{k \in \mathbb{N}} 
			\left(\int_{\Lambda}H_{k}(X(x)) dx\right) ^n
			\right]\right)		
			\\=&
			\sum_{n=2}^\infty 
			\sum_{\substack{\{m_k\}_{k \in \mathbb{N}} \\ \sum_{k}m_k = n}}
			\prod_{k \in \mathbb{N}} \frac{(-\alpha_k)^{m_k}}{m_k!}
			\proj \left(
			\mathbb{E}\left[
			\prod_{k \in \mathbb{N}} 
			\left(\int_{\Lambda}H_{k}(X(x)) dx\right) ^n
			\right]\right)
			\\=&
			\sum_{\Gamma \in \mathbf{F}}\frac{\Upsilon^\alpha_\F (\Gamma)}{\hat{S}_\F(\Gamma)}
			\left(  N(\Gamma)\Pi_\F(\Gamma)  \right)
		\end{equs}
		where for $k \in \mathbb{M}$,  $m_k \in \mathbb{N}$.
\end{proof}

Finally, we are able to find the renormalised measure from renormalised cumulant by the following Theorem which is the generalisation of \cite[Prop. 3.10]{BK23}.
}
\begin{theorem} \label{prop:exponential}
	For a model with rule $\CR$ and the renormalisation $\R_{\BPHZ}$ through the linear map 
	\begin{equs}
		\hat{M}_\M := \left(\Pi_{\M} (\CA_{\M}\cdot) \otimes \id \right) \Delta_{\M}^{-}  
	\end{equs}
	of multi-indices, one has in terms of formal series
		\begin{equs}
		\R_{\BPHZ} \exp \left(-  \int_\Lambda \sum_{k \in \CK(\CR)} \alpha_k H_k(X(x)) dx \right) = \exp \left(-  \int_\Lambda \sum_{k \in \CK(\CR) \cup \{0 \}} (\alpha_k +\gamma_k) H_k(X(x)) dx \right)
	\end{equs}
where $H_k$ represents the $k$-th order Hermite polynomial, and the coefficients are
	\begin{equs}
		\gamma_k =   - \sum_{z^{\check{\delta}} \in \check{\mathbf{M}}_{\CR,k}} 
		\sum_{z^{\delta} \in \mathbf{M}_-}  \Pi_{\M}\CA_{\M}(z^{\delta})
		 \frac{ \langle  D^k z^{\delta} ,z^{\check{\delta}} \rangle}{k! \hat{S}_{\M}(z^{\check{\delta}}) S_{\M}(z^{\delta})}
		  \Upsilon^{\alpha}_\M[z^{\check{\delta}}],
	\end{equs}
\begin{equs}
	\Upsilon^{\alpha}_\M[z^\beta] = \prod_{k \in \mathbb{N}} (-\alpha_k)^{\beta(k)}, \quad \hat{S}_{\M}(z^{\beta}) = \prod_{k \in \mathbb{N}} \beta(k)! .
\end{equs}
Here $\check{\mathbf{M}}_{\CR,k}$ is the space of multi-indices that can form Feynman diagrams with $k$ unpaired half-edges, and for any $\check{\delta}(k) \in \check{\mathbf{M}}_{\CR,k}$ the arity $k \in \CK(\CR)$ if $\check{\delta}(k) \ne 0$.
	\end{theorem}
\begin{proof} 
Lemma \ref{prop:cumulant_F}, Corollary \ref{coro:commute} together with the equalities $\hat{S}_{\M}(\Phi(\Gamma)) = \hat{S}_{\F}(\Gamma)$ and $\Upsilon_\M^{\alpha}(\Phi(\Gamma)) = \Upsilon_\F^{\alpha}(\Gamma)$ yield that 
  	\begin{equs}
  	\log \mathbb{E} \left[\exp \left(-  \int_\Lambda \sum_{k \in \CK(\CR)} \alpha_k H_k(X(x)) dx \right) \right]  = \sum_{z^{\beta} \in \mathbf{M}_\CR}  \frac{\Upsilon_\M^{\alpha}[z^{\beta}]}{\hat{S}_{\M}(z^{\beta})} \Pi_{\M}z^{\beta}
  \end{equs}
  By the same reasoning
  	\begin{equs}
  	\log \mathbb{E} \left[\exp \left(-  \int_\Lambda \sum_{k \in \CK(\CR) } (\alpha_k +\gamma_k) H_k(X(x) \right) dx) \right]  
  	= \sum_{z^{\beta} \in \mathbf{M}_\CR}  \frac{\Upsilon_\M^{(\alpha+\gamma)}[z^{\beta}]}{\hat{S}_{\M}(z^{\beta})} \Pi_{\M}z^{\beta}.
  \end{equs}
 By Theorem \ref{them:main},
 \begin{equs}
 	\R_{\BPHZ} \exp \left(-  \int_\Lambda \sum_{k \in \CK(\CR)} \alpha_k H_k(X(x)) dx \right)
 	 = \sum_{z^{\beta} \in \mathbf{M}_\CR}  \frac{\Upsilon_\M^{\alpha}[z^{\beta}]}{\hat{S}_{\M}(z^{\beta})} \Pi_{\M} \hat{M}_\M z^{\beta}.
 \end{equs}
%
Then let us focus on the renormalisation of multi-indices, and suppose there exist an operator $\hat{M}_\M^*$ such that
\begin{equs}
 \sum_{z^{\beta} \in \mathbf{M}_\CR}  \frac{\Upsilon_\M^{\alpha}[z^{\beta}]}{\hat{S}_{\M}(z^{\beta})} \hat{M}_\M z^{\beta}
	 = \sum_{z^{\theta} \in \mathbf{M}_\CR }  \frac{\Upsilon^{\alpha}[\hat{M}_\M^* z^{\theta}]}{\hat{S}_{\M}(z^{\theta})}  z^{\theta}.
\end{equs}
Since the adjoint relation gives that the number of isomorphic extraction-contraction
\begin{equs}
	\frac{\langle\tilde{z}^{\tilde{\mu}} \otimes z^\theta , \Delta^{\! -}_{\M_\CR} z^\beta \rangle }{S(\tilde{z}^{\tilde{\mu}})S(z^\theta )}
	=
	\frac{\langle\tilde{z}^{\tilde{\mu}} \bar\star_{\M_\CR} z^\theta , z^\beta \rangle }{S_{\M}(z^\beta)}   \frac{S_{\M}(z^\beta) }{ S_{\M}(z^\theta)S_{\M}(\tilde{z}^{\tilde{\mu}})},
\end{equs}
$ \hat{M}_\M^* $ has the expression
\begin{equs}
	\hat{M}_\M^* z^{\theta} = z^{\theta}  
	+ \sum_{z^{\beta}\in \mathbf{M}_\CR} 
\sum_{\tilde{z}^{\tilde{\mu}} \in \mathcal{M}_-}   \frac{\langle\tilde{z}^{\tilde{\mu}} \bar\star_{\M_\CR} z^\theta , z^\beta \rangle }{S_{\M}(z^\beta)}  (\Pi_{\M}\CA_{\M}(\tilde{z}^{\tilde{\mu}})) \frac{S_{\M}(z^\beta) \hat{S}_{\M}(z^\theta)}{\hat{S}_{\M}(z^{\beta}) S_{\M}(z^\theta)S_{\M}(\tilde{z}^{\tilde{\mu}})} z^{\beta}.
\end{equs}
One can easily observe, due to the multiplicativity of $ \frac{\hat{S}_{\M}(\cdot)}{S_{\M}(\cdot)} $, that $ \hat{M}_\M^* $ is multiplicative in the sense that
 \begin{equs}
 	\hat{M}_\M^* z^{\theta} = \prod_{k \in \mathbb{N}} (\hat{M}_\M^* z_k)^{\theta(k)}
 \end{equs}
\begin{equs}
	\hat{M}_\M^* z_k = z_k
	+ \sum_{z^{\beta}\in \mathbf{M}_\CR} 
	\sum_{{z}^{{\mu}} \in \mathbf{M}_-}   \frac{\langle{z}^{{\mu}} \blacktriangleright_{\M_\CR} z_k , z^\beta \rangle }{S_{\M}(z^\beta)}  (\Pi_{\M}\CA_{\M}({z}^{{\mu}})) \frac{S_{\M}(z^\beta) }{\hat{S}_{\M}(z^{\beta}) k! S_{\M}({z}^{{\mu}})} z^{\beta},
\end{equs}
where we use $\blacktriangleright_{\M_\CR}$  instead of its primitive-terms extension of $\bar\blacktriangleright_{\M_\CR}$ as both $z_k$ and $z^\mu$ are non-empty ($z_k$ is not a proper multi-indice but we abuse the notation of $\blacktriangleright_{\M_\CR}$ as the formula of this product is the same as $\blacktriangleright_{\M_\CR}$ defined before. Only the space that the product works on is different.).
Therefore, one has
\begin{equs}
	\Upsilon_\M^{\alpha}[\hat{M}^* z^{\theta}] = \Upsilon_\M^{\hat{M}_\M^{*}\alpha}[ z^{\theta}]
	\end{equs}
where $ \Upsilon_\M^{\hat{M}_\M^{*}\alpha} $ is defined by
\begin{equs}
	\Upsilon^{\hat{M}^{*}\alpha}[z_k] := \Upsilon_\M^{\alpha}[\hat{M}_\M^* z_k] = -\alpha_k - \gamma_k,
\end{equs}
and then it is extended multiplicatively, which allows us to conclude.
Note that $\gamma_0$ is from the primitive term $\Pi_{\M}(z^\beta) \otimes \emptyset_\M$.
	\end{proof}
	\begin{remark}
		The adjoint argument used in the previous proof is at the core of all the proofs for getting renormalised equations for singular SDEs with rough paths (see \cite{BCFP,L23}) and singular SPDEs with Regularity Structures (see \cite{BCCH,BB21,BL23}). Let us mention that in the present case, the definition of $ \hat{M}_\M^*  $ is more complicated due to the use of two combinatorial factors $ S_\M(\cdot) $ and $ \hat{S}_\M(\cdot) $.  
		\end{remark}

		One has in fact a group structure on $\hat{M}_\M$. Indeed, one has
		\begin{equs}
			\hat{M}_\M := \left( \ell^{\BPHZ}_{\M} \otimes \id  \right) \Delta^{\!-}_{\M}, \quad \ell^{\BPHZ}_{\M} := \Pi_{\M} \CA_{\M}.
		\end{equs}
	The map $ \ell^{\BPHZ}_{\M} : \langle\CM_{-}\rangle \rightarrow \mathbb{R} $ is a  character being multiplicative for the forest product. 
Now the map $ \hat{\Delta}^{\!-}_{\M} $ is defined from $ \langle\CM_{-}\rangle$ into $ \langle\CM_{-}\rangle \otimes \langle\CM_{-}\rangle $ via the same defintion as for $ \Delta^{\!-}_{\M}  $: 
\begin{equs}
	\hat{\Delta}^{\!-}_{\M} =  \left( \id \otimes \pi_{\M}^{-}  \right)\Delta^{\!-}_{\M} 
\end{equs}
where $ \pi_{\M}^{-} $ is the canonical projection from $ \langle\CM\rangle $ into $ \langle\CM_{-}\rangle $. One can see easily that $ \Delta_{\M}^{\!-} $ is a coaction namely
\begin{equs}
	\left( \id  \otimes \Delta_{\M}^{\!-}  \right)\Delta_{\M}^{\!-}  = \left(  \hat{\Delta}_{\M}^{\!-} \otimes \id \right)\Delta_{\M}^{\!-}.
\end{equs}
Then, one can define an antipode $ \hat {\CA}_{\M} : \langle\CM_{-}\rangle \rightarrow \langle\CM_{-}\rangle $ by
\begin{equs} \label{eq:def_antipode_MI_-}
\begin{aligned}	& \hat{\CA}_\M(\emptyset_\M) := \emptyset_\M,  \quad \hat{\CA}_{\M}( \tilde z^{\tilde{\alpha}} \tilde{\bullet}_\M \tilde z^{\tilde{\beta}} ) := \hat{\CA}_\M( \tilde z^{\tilde{\alpha}} ) \tilde{\bullet}_\M \hat{\CA}_\M( \tilde z^{\tilde{\beta}} ) \quad \text{ for any } \tilde z^{\tilde{\alpha}}, \tilde z^{\tilde{\beta}} \in \mathcal{M}_-,
	\\&
	\hat{\CA}_\M( z^{ \beta}) = - z^{\beta} - \mu_{\M} \circ(\hat{\CA}_\M \otimes \id)\circ\hat{\Delta}_{\M}  z^{ \beta},
	\quad \text{ for any }  z^{ \beta} \in \mathbf{M}_-.
	\end{aligned}
\end{equs}
where	$\hat{\Delta}_{\M}$ is the reduced map of 	$\hat{\Delta}_{\M}^{\!-}$. Equipped with $ \hat{\Delta}_{\M}^{\!-} $, $ 	\hat{\CA}_\M( z^{ \beta}) $ and the forest product, $ \langle\CM_{-}\rangle $ is Hopf algebra and $ \langle\CM \rangle $ equipped with $\hat{\Delta}_{\M}^{\!-}$ is a left comodule over $ \langle\CM \rangle $. One has also a group structure on the characters of $ \langle\CM_{-}\rangle $ given by
\begin{equs}
	\mathcal{G}_{\M}^{-} = \lbrace f_{\M} :\langle\CM_{-}\rangle \rightarrow \mathbb{R}, \, f_{\M}(\tilde z^{\tilde{\alpha}} \tilde{\bullet}  \tilde z^{\tilde{\beta}}) = f_{\M}(\tilde z^{\tilde{\alpha}}) \tilde{\bullet}  f_{\M}(\tilde z^{\tilde{\beta}})  \rbrace.
	\end{equs}
The product for this group is given by $ \star_{\M}^- $ and its inverse by the antipode for every $f_{\M}, g_{\M} \in \CG_{\M}^-$ by
\begin{equs}
f_{\M} \star_{\M}^- g_{\M} :=\left( f_{\M} \otimes g_{\M} \right) \hat{\Delta}^{\!-}_{\M}, \quad	f^{-1}_{\M} = f_{\M}(\hat{\CA}_\M \cdot).
\end{equs} 
Then,  one has
\begin{equs}
	M_{f_{\M}} \circ M_{g_{\M}} = M_{f_{\M} \star_{\M}^- g_{\M}}, \quad   M_{f_{\M}} := \left( f_{\M} \otimes \id \right) \Delta_{\M}^{\!-}.
\end{equs}
One has a similar construction directly on the Feynman diagrams with the extraction-contraction coproduct  $ \hat{\Delta}^{\!-}_{\F} $  defined from $ \langle\CF_{-}\rangle$ into $ \langle\CF_{-}\rangle \otimes \langle\CF_{-}\rangle $ via the same defintion as for $ \Delta^{\!-}_{\F}  $: 
\begin{equs}
	\hat{\Delta}^{\!-}_{\F} =  \left( \id \otimes  \pi_{\F}^{-}  \right)\Delta^{\!-}_{\F} 
\end{equs}
where $ \pi_{\F}^{-} $ is the canonical projection from $ \langle\CF\rangle $ into $ \langle\CF_{-}\rangle $. we denote by $ \star_{\F}^- $ the convolution product associated with $ \hat{\Delta}^{\!-}_{\F} $ and $ \CG_{\F}^- $ is a group of characters. Using the counting map and the identity \eqref{counting_map_morphism}, one can move from one group to the other: 
\begin{equs}
\Phi \left( f_{\F} \star_{\F}^-	 g_{\F} \right) = 	f_{\M} \star_{\M}^-	 g_{\M}, \quad f_{\M} = \Phi \left( f_{\F} \right), \quad  	 g_{\M} = \Phi \left( g_{\F} \right),
\end{equs}
where $ 	f_{\M}, g_{\M} \in \CG_{\M}^- $ and $ 	f_{\F}, g_{\F} \in \CG_{\F}^- $.

\section{Example: Renormalisation of the $\Phi^4$ measure} \label{sec:Example}
In this section, we will use multi-indices to renormalise the $\Phi^4$ measure and show how it is equivalent to the BPHZ renormalisation of some Feynman diagrams appearing in the cumulant expansion. The renormalisation of $\Phi^4$ measure using a new type of Hopf algebra with some ``monomials" representing the vertices instead of Feynman diagrams was initially studied in \cite{BK23}. The main idea of this section is to show that these ``monomials" are essentially some multi-indices and the Hopf algebra the authors used can be formalised as the extraction-contraction of multi-indices which allows us to generalise their results to broader models.

\subsection{The model}
We consider the $\Phi^4$ model on the $d$-dimensional torus $\Lambda = \mathbb{T}^d = (\mathbb{R}/\mathbb{Z})^d$
\begin{equs}
	\partial_t\phi(t,x) = \Delta \phi(t,x)-m^2\phi(t,x) - \varepsilon\phi(t,x)^3+\xi(t,x)
\end{equs}
with $ \phi: \mathbb{R}_{+} \otimes \Lambda \mapsto \mathbb{R}$.
The corresponding invariant measure is 
\begin{equs}
	\mu_{\varepsilon}(d\phi) = \frac{1}{Z(\varepsilon)} \exp\left\{ -\int_\Lambda \left(\frac{1}{2}\Vert \nabla \phi(x)\Vert^2 +\frac{1}{2}m^2\phi(x)^2+\frac{\varepsilon}{4}\phi(x)^4\right)dx \right\}d\phi
\end{equs}
where $Z(\varepsilon)$ is the partition function such that $\int_{\mathbb{R}} \mu_{\varepsilon}(d\phi) = 1$.
For simplicity we consider $m=1$. It can be extended to any $m$ by a Gaussian change of measure. Then, let $\alpha  = \frac{\varepsilon}{4}$ and the measure can be rewritten as
\begin{equs}
	\mu_{\alpha}(d\phi) = \frac{1}{Z(\alpha)} \exp\left\{ -\int_\Lambda \left(\frac{1}{2}\Vert \nabla \phi(x)\Vert^2 +\frac{1}{2}\phi(x)^2+\alpha \phi(x)^4\right)dx \right\}d\phi
\end{equs}
However, in the second-dimensional torus case, this measure diverges, which makes Wick renormalisation  necessary. This amounts to changing all the products of $\phi(x)^n$ to the Wick power $:\phi(x)^n:$ \cite{WB08}. For Gaussian $\phi(x)$, the Wick power can be described by the Hermite polynomial
\begin{equs}
	:\phi(x)^n: = H_n\left(\phi(x), \mathrm{Var} \left(\phi(x)\right)\right).
\end{equs}

\subsection{Renormalise $\Phi^4$ measure through multi-indices}
For $d >2$, besides the Wick renormalisation we also need the BPHZ renormalisation. We will implement the renormalisation by cumulant expansion of the partition functions used in \cite{BK23} and take the example $d=3$.
Firstly denote the exponent after Wick renormalisation as 
\begin{equs}
	G_\alpha (\phi) :&= \int_\Lambda \left(\frac{1}{2}\Vert \nabla \phi(x)\Vert^2 +\frac{1}{2}:\phi(x)^2:+\alpha :\phi(x)^4:\right)dx 
	\\& = G_0(\phi)+ \alpha \int_\Lambda  :\phi(x)^4:dx .		
\end{equs}
Then, we have the partition functions
\begin{equs}
	Z(\alpha) = Z(0) \mathbb{E}\left[e^{-\alpha \int_\Lambda :\phi(x)^4:dx }\right].
\end{equs}
where under the expectation the Wick power of $\phi$ is the Hermit polynomial of Gaussian free field.
By the
	Linked Cluster Theorem
, the cumulant expansion
\begin{equs}\label{eq:cumulant}
	\mathrm{log}	\mathbb{E}\left[e^{-\alpha \int_\Lambda :\phi(x)^4:dx }\right] = \sum_{n=2}^\infty (-\alpha)^n \frac{\kappa_n}{n!} 
\end{equs}
where $\kappa_n$ is the $\mathbb{E}\left[\left(\int_\Lambda :\phi(x)^4:dx \right)^n\right]$ projected to all the connected Feynman diagrams in the pairing-half-edges problem \eqref{eq:pairing_problem} in which the kernel is the $d$-dimensional Green's function. By Corollary~\ref{coro:commute}, it is equivalent to lifting the expectation to the forest of multi-indices with only one element which is the single multi-indice $z_4^n$, which means
\begin{equs}
	\mathrm{log}	\mathbb{E}\left[e^{-\alpha \int_\Lambda :\phi(x)^4:dx }\right] = \sum_{n=2}^\infty  \frac{(-\alpha)^n}{n!} \Pi_\M(z_4^n)
\end{equs}
Therefore, we should renormalise the $\Phi^4$ measure via the reduced extraction-contraction coproduct $\Delta_\M$ operating on $z_4^n$.
\vspace{6pt}

For the $d=3$ case, the degree map of multi-indices is
\begin{equs}
	\deg z^\beta =   -\frac{1}{2}\sum_{k \in \mathbb{N}}k \beta(k) + 3(|z^\beta| - 1)
\end{equs}
as when $d=1$ the Green's function behaves like $-1$-Hölder. One can check that the only forests in $\CM_{-}$ can be extracted from $z_4^n$ are $({z_3^2})^{\tilde{\bullet}^m}$ for $n \ge 4$ and $m \le \lfloor\frac{n}{2}\rfloor$.
Therefore, the rule of insertion is 
\begin{equs}
	\CR  = \{k: k =2, 4\}.
\end{equs}
Then for $n \ge 4$ extraction-contraction coproduct is
\begin{equs}
	\Delta_\M z_4^n =
	\sum_{m=1}^{\lfloor\frac{n}{2}\rfloor}
	E\left(({z_3^2})^{\tilde{\bullet}^m}, z_2^mz_4^{n-2m} ,z_4^n\right)  
	 ({z_3^2})^{\tilde{\bullet}^m}
	\otimes   z_2^mz_4^{n-2m}.
\end{equs}
Then we have to calculate
\begin{equs}
	&E\left(({z_3^2})^{\tilde{\bullet}^m}, z_2^mz_4^{n-2m} ,z_4^n\right)  
	\\&=   
	\sum_{k_1,...,k_m\in \mathbb{N}}
	\sum_{\beta  = \hat{\beta}_1 + \cdots + \hat{\beta}_n+\hat\alpha} 
	\frac{4^mn!}{m!^2(n-2m)!}
	\frac{
		\langle \prod_{i=1}^m\partial_{z_{k_i}} z^{\alpha} , z^{\hat{\alpha}} \rangle}
	{S(z^{\hat{\alpha}})}
	\prod_{i=1}^m
	\frac{\langle  
		D^{k_i}z^{\beta_i}, z^{\hat{\beta}_i} \rangle }{S(z^{\hat{\beta}_i})},
\end{equs}
with $z^\alpha  = z_2^mz_4^{n-2m}$.
According to the rule, the only way allowed to add edges to $z_3^2$ is to form one $z_4^2$. i.e., the possible $z^{\hat\beta_i}$ is $z_4^2$ and thus the only possible $z^{\hat\alpha}$ is $z_4^{n-2m}$. Therefore,
\begin{equs}
&	\prod_{i=1}^m
	\frac{\langle  
		D^{k_i}z^{\beta_i}, z^{\hat{\beta}_i} \rangle }{S(z^{\hat{\beta}_i})}
	= 2^m, \quad 
	\frac{
		\langle \prod_{i=1}^m\partial_{z_{k_i}} z^{\alpha} , z^{\hat{\alpha}} \rangle}
	{S(z^{\hat{\alpha}})} = m!
	.
\end{equs}
Finally, we have 
\begin{equs} \label{eq:coproduct_example}
	\Delta_\M z_4^n =
	\sum_{m=1}^{\lfloor\frac{n}{2}\rfloor}
	\frac{2^{3m}n!}{m!(n-2m)!}
	({z_3^2})^{\tilde{\bullet}^m}
	\otimes   z_2^mz_4^{n-2m}.
\end{equs}
Then, the twisted antipode
\begin{equs} 
\CA_\M(z_4^2) = -z_4^2,
	\quad
	\CA_\M(z_4^3) = -z_4^3,
	\quad
	\CA_\M(z_4^n) =  0,
	\quad \text{ for } n\ge4.
\end{equs}
Meanwhile, the corresponding renormalisation is
\begin{equs} 
\hat{M}_\M(z_4^2) &= z_4^2 - \Pi_{\M} z_4^2,
	\quad
	\hat{M}_\M(z_4^3) = z_4^3 - \Pi_{\M} z_4^3,
	\\
	\hat{M}_\M(z_4^n) &=  \sum_{m=0}^{\lfloor\frac{n}{2}\rfloor} \frac{(-8)^{m}n!}{m!(n-2m)!}
	(\Pi_{\M}{z_3^2})^{m}
	  z_2^mz_4^{n-2m},
	\quad \text{ for } n\ge4.
\end{equs}
Consequently, the cumulant expansion \eqref{eq:cumulant} after renormalisation through $\hat{M}_\M$ is
\begin{equs}
 \sum_{n=2}^\infty  \frac{(-\alpha)^n}{n!} \hat{M}_\M z_4^n = 	& -  \frac{\alpha^2}{2}\Pi_\M(z_4^2) 
	+  \frac{\alpha^3}{3!}\Pi_\M(z_4^3)
	+ \frac{\alpha^2}{2} z_4^2 
	-  \frac{\alpha^3}{3!} z_4^3 
	\\+&
	\sum_{n=4}^\infty  
\sum_{m=0}^{\lfloor\frac{n}{2}\rfloor} \frac{(-8)^{m}n!}{m!(n-2m)!}
(\Pi_{\M}{z_3^2})^{m}
z_2^mz_4^{n-2m}.
\end{equs}
Let $p = n-m$, $q= n-2m$ and $\beta = 8\alpha^2 \Pi_\M({z_3^2})$. Then the expansion is equivalent to
\begin{equs}
& \sum_{p=0}^\infty  \sum_{q=0}^\infty \frac{1}{p!}{p \choose q} 
(-\beta )^{p-q}(-\alpha)^p z_2^{p-q} z_4^q -  \frac{\alpha^2}{2}\Pi_\M(z_4^2) 
+  \frac{\alpha^3}{3!}\Pi_\M(z_4^3)
\\=&
 \sum_{p=0}^\infty   \frac{1}{p!} (-\beta z_2- \alpha z_4)^p -  \frac{\alpha^2}{2}\Pi_\M(z_4^2) 
 +  \frac{\alpha^3}{3!}\Pi_\M(z_4^3).
\end{equs}
This is in agreement with the renormalisation by Theorem 5.2.4 in \cite{B22}. It is given by 
\begin{equs} \label{renormalisation_measure2}
	\hat Z_N(\alpha) = Z(0) \mathbb{E}\left[e^{-\alpha \int_\Lambda :\phi(x)^4:dx - \beta \int_\Lambda :\phi(x)^2:dx  - \frac{\alpha^2}{2}\Pi_\M(z_4^2) 
		+  \frac{\alpha^3}{3!}\Pi_\M(z_4^3) }\right]
\end{equs}
which is the same as one in \cite{BK23}.

Moreover, we can check that the same result can be obtained from Theorem \ref{prop:exponential}.
Now, we proceed with the following renormalisation
\begin{equs}
\R_{\BPHZ} \exp\left(-  \int_\Lambda \sum_{k \in \CK(\CR)} \alpha_k H_k(X(x)) dx \right) = \exp \left(-  \int_\Lambda \sum_{k \in \CK(\CR) \cup \{0 \}} (\alpha_k +\gamma_k) H_k(X(x)) dx \right).
\end{equs}
	\begin{equs}
	\gamma_k =   - \sum_{z^{\check{\delta}} \in \check{\mathbf{M}}_{\CR,k}} 
	\sum_{z^{\delta} \in \mathbf{M}_-}  \Pi_{\M}\CA_{\M}(z^{\delta})
	\frac{ \langle  D^k z^{\delta} ,z^{\check{\delta}} \rangle}{k! \hat{S}_{\M}(z^{\check{\delta}}) S_{\M}(z^{\delta})}
	\Upsilon^{\alpha}_\M[z^{\check{\delta}}].
\end{equs}
Firstly, $\gamma_4 = 0$ as there is no such $D^4 z^\delta$ obeys the rule $\CR$. Then we have to compute 
	\begin{equs}
	\gamma_2 &=   -  (-\Pi_{\M}z_3^2) \frac{\langle  D^2 z_3^2 ,z_4^2 \rangle}{2! \hat{S}(z_4^2) S(z_3^2)} \Upsilon_\M^{\alpha}[z_4^2]
	\\&= -
	(-\Pi_{\M}z_3^2) \frac{ 2 \cdot 2!4!^2}{2! 2! (3!^2 2!)} (-\alpha)^2
	\\&= 8 \alpha^2 \Pi_{\M}z_3^2 =  \beta.
\end{equs} 
With a similar computation, one has
\begin{equs}
	\gamma_0 =  \frac{\alpha^2}{2}\Pi_\M(z_4^2) 
	- \frac{\alpha^3}{3!}\Pi_\M(z_4^3). 
	\end{equs}
The previous computation shows that we obtain \eqref{renormalisation_measure2}.

\subsection{Some examples of Theorem~\ref{them:main2}}
Let us consider the example of $z_4^4$. According to \eqref{eq:coproduct_example}, the reduced coproduct  is
\begin{equs}
	\Delta_\M z_4^4 &=
	\frac{8 \times 4!}{2!}
	z_3^2
	\otimes   z_2z_4^2 
	+\frac{64 \times 4!}{2!}
	({z_3^2})^{\tilde{\bullet}^2}
	\otimes   z_2^2 
	\\&= 4!\left(4z_3^2
	\otimes   z_2z_4^2  + 32({z_3^2})^{\tilde{\bullet}^2}
	\otimes   z_2^2  \right)	
\end{equs}
Then lift multi-indices in the result to Feynman diagrams and we have 
\begin{equs}
	\CP\left(({z_3^2})^{\tilde{\bullet}^m}\right) 
	= \frac{S(\left(z_3^2\right)^{\tilde{\bullet}^m}) }{S(\text{$\FGIII$}^{\tilde{\bullet}_\F^m})} \text{$\FGIII$}^{\tilde{\bullet}_\F^m}
	= \frac{m! (3!)^{2m}(2!)^m}{m!(2!)^m(3!)^m} = (3!)^m \text{$\FGIII$}^{\tilde{\bullet}_\F^m},
\end{equs}
\begin{equs}
	\CP\left(z_2z_4^2 \right) = \frac{S(z_2z_4^2 )}{S(\text{\FGIIIplus})} \text{\FGIIIplus}= \frac{2!4!^22!}{2!3!} \text{\FGIIIplus}= 8\times 4! \text{\FGIIIplus},
\end{equs}
\begin{equs}
	\CP(z_2^2) = \frac{S(z_2^2)}{S(\text{\YII})} \text{\YII} = \frac{2!^22!}{2!2!} \text{\YII} = 2\text{\YII}.
\end{equs}
Therefore,
\begin{equs}
	(\CP \otimes \CP) \circ \Delta_\M (z_4^4)  = 4\times 4! ^3  \left( 2 \text{\FGIII} \otimes \text{\FGIIIplus} +  \text{\FGIII} \text{\FGIII} \otimes \text{\YII}\right).
\end{equs}
Let us then check the other direction $\Delta_\F \circ \CP(z_4^4)$ of the commutative diagram.
\begin{equs}
	\CP(z_4^4) &= \frac{S(z_4^4)}{S(\ZIVI)}  \ZIVI+ \frac{S(z_4^4)}{S(\ZIVII)}  \ZIVII+ \frac{S(z_4^4)}{S(\ZIV)} \ZIV 
	\\&=
	 \frac{4!^5}{8\times2!^4} \ZIVI
	 + \frac{4!^5}{2!^32!^2} \ZIVII
	 + \frac{4!^5}{4\times3!^2} \ZIV 
	 \\&= 2^83^5\ZIVI 
	 + 2^{10}3^5 \ZIVII
	 + 4\times 4! ^3 \ZIV.
\end{equs}
Notice that $\ZIVI$ and $\ZIVII$ have no divergent subgraph and thus their coproduct is $0$. We then need to calculate
\begin{equs}
	\Delta_\F \ZIV = 2 \FGIII \otimes \FGIIIplus + \FGIII \FGIII \otimes \YII.
\end{equs} 
which concludes 
\begin{equs}
	(\CP \otimes \CP) \circ \Delta_\M (z_4^4)  = \Delta_\F \circ \CP(z_4^4).
\end{equs}


\begin{thebibliography}{Cha10}
\expandafter\ifx\csname url\endcsname\relax
  \def\url#1{\texttt{#1}}\fi
\expandafter\ifx\csname urlprefix\endcsname\relax\def\urlprefix{URL }\fi
\expandafter\ifx\csname href\endcsname\relax
  \def\href#1#2{#2}\fi
\expandafter\ifx\csname burlalt\endcsname\relax
  \def\burlalt#1#2{\href{#2}{\texttt{#1}}}\fi

  
 

 \bibitem{pre_FD}
 	A.~Abdesselam. {\em Feynman diagrams in algebraic combinatorics.} Semin. Lothar. Comb.
 	\textbf{B49c}, (2002), 45 p., electronic only.
 	\burlalt{http://eudml.org/doc/123420}{http://eudml.org/doc/123420}.
 	
 	\bibitem{Abe}
 	E.~Abe. \newblock{ \em Hopf algebras}, \newblock Cambridge Tracts in Mathematics, vol. 74, Cambridge University Press,
 	Cambridge-New York, 1980. Translated from the Japanese by Hisae Kinoshita and Hiroko Tanaka.
 	
 
 \bibitem{BB21}
 I.~{Bailleul}, Y.~{Bruned}, \newblock { \em Locality for singular stochastic PDEs}. 
 \newblock \burlalt{arXiv:2109.00399}{http://arxiv.org/abs/2109.00399}. 
 
  \bibitem{config1}
  B.~Bollobás
  \newblock {\em A Probabilistic Proof of an Asymptotic Formula for the Number of Labelled Regular Graphs}.
  \newblock Eur. J. Comb.  \textbf{1}, no.~4, (1980), 311--316.
  \newblock
  \burlalt{doi:10.1016/S0195-6698(80)80030-8}{https://doi.org/10.1016/S0195-6698(80)80030-8}.
  
  \bibitem{BCCH}
 { \rm Y. Bruned, A. Chandra, I. Chevyrev,
  M. Hairer}.
\newblock {\em Renormalising SPDEs in regularity structures}.
\newblock J. Eur. Math. Soc. (JEMS), \textbf{23}, no.~3, (2021), 869--947.
\newblock
  \burlalt{doi:10.4171/JEMS/1025}{http://dx.doi.org/10.4171/JEMS/1025}.
  
  \bibitem{BCFP} 
  Y.~{Bruned}, I.~{Chevyrev}, P.~K. {Friz}, R.~{Preiss}.
  \newblock{ \em  A rough path perspective on renormalization.}
  \newblock  J. Funct. Anal. \textbf{277}, no.~11, (2019), 108283.
  \newblock 
  \burlalt{doi:10.1016/j.jfa.2019.108283}{http://dx.doi.org/10.1016/j.jfa.2019.108283}.

  \bibitem{B22}
N.~Berglund.
\newblock { \em An Introduction to Singular Stochastic PDEs}. 
\newblock EMS Press (2022).
\burlalt{doi:10.4171/ELM/34}{https://doi.org/10.4171/elm/34}.

  
\bibitem{BEH}
{\rm Y. Bruned, K. Ebrahimi-Fard, Y. Hou}.
\newblock {\em Multi-indice $B$-series.}
\newblock J. Lond. Math. Soc. \textbf{111}, no.~1, (2025), e70049.
\newblock
\burlalt{doi:10.1112/jlms.70049}{https://doi.org/10.1112/jlms.70049}.

  \bibitem{BH24}
{\rm Y. Bruned,  Y. Hou}.
\newblock {\em Multi-indices coproducts from ODEs to singular SPDEs.} To appear in Transactions of the American Society.
\burlalt{arXiv:2405.11314.}{https://arxiv.org/pdf/2405.11314}.


\bibitem{BHZ}
{\rm Y. Bruned, M. Hairer, L. Zambotti}.
\newblock {\em Algebraic renormalisation of regularity structures.}
\newblock Invent. Math. \textbf{215}, no.~3, (2019), 1039--1156.
\newblock
  \burlalt{doi:10.1007/s00222-018-0841-x}{https://dx.doi.org/10.1007/s00222-018-0841-x}.
  
  
  
 \bibitem{BK23}
 N.~Berglund, T.~{Klose}.
 \newblock { \em Perturbation theory for the $\Phi^4_3$ measure, revisited with Hopf algebras}. 
 \burlalt{arXiv:2207.085559}{https://arxiv.org/abs/2207.085559}. 
 
 
 \bibitem{BYK23}
 Y.~Bruned, F.~Katsetsiadis.
 \newblock {\em Post-Lie algebras in Regularity Structures}.  Forum of Mathematics, Sigma \textbf{11}, e98, (2023), 1--20.
 \newblock 
 \burlalt{doi:10.1017/fms.2023.93}{http://dx.doi.org/10.1017/fms.2023.93}. 
 
 \bibitem{BL23}
 Y.~{Bruned}, P.~{Linares}.
 \newblock {\em  A top-down approach to algebraic renormalization in regularity structures based on multi-indices.} Arch. Ration. Mech. Anal. \textbf{248}, 111 (2024). 
 \burlalt{doi:10.1007/s00205-024-02041-4}{http://dx.doi.org/10.1007/s00205-024-02041-4}. 
 
 	\bibitem{BM22}
 Y.~Bruned, D.~Manchon.
 \newblock {\em Algebraic deformation for (S)PDEs}. J. Math. Soc. Japan. \textbf{75}, no.~2, (2023), 485--526.
 \newblock 
 \burlalt{doi:10.2969/jmsj/88028802}{http://dx.doi.org/10.2969/jmsj/88028802}. 
 
  
  \bibitem{BP}
  N.~N. Bogoliubow, O.~S. Parasiuk.
  \newblock \emph{\"{U}ber die {M}ultiplikation der {K}ausalfunktionen in der
  	{Q}uantentheorie der {F}elder}.
  \newblock Acta Math. \textbf{97}, (1957), 227--266.
  \newblock
  \burlalt{doi:10.1007/BF02392399}{http://dx.doi.org/10.1007/BF02392399}.
  
   \bibitem{connect4}
  C.Brouder.
  \newblock {\em Quantum field theory meets Hopf algebra}.
  \newblock Math. Nachr. \textbf{282}, no.~12, (2009), 1664--1690.
  \newblock
  \burlalt{doi:10.1002/mana.200610828}{https://doi.org/10.1002/mana.200610828}.
  
   \bibitem{connect1}
 K.~A. Brueckner.
  \newblock {\em Many-Body Problem for Strongly Interacting Particles. II. Linked Cluster Expansion}.
  \newblock Phys. Rev. \textbf{100}, (1955), 36--45.
  \newblock
  \burlalt{doi:10.1103/PhysRev.100.36}{http://dx.doi.org/10.1103/PhysRev.100.36}.
  
  \bibitem{BS}
  Y.~{Bruned}, K.~{Schratz}. \newblock { \em Resonance based schemes for dispersive equations via decorated 
  	trees}. Forum of Mathematics, Pi, 10, E2. 
  \newblock \burlalt{doi:10.1017/fmp.2021.13}{https://doi.org/10.1017/fmp.2021.13}.
  
	\bibitem{Butcher72}
{ \rm J. C. Butcher}.
\newblock {\em An algebraic theory of integration methods.}
\newblock Math. Comp. \textbf{26}, (1972), 79--106.
\newblock \burlalt{doi:10.2307/2004720}{http://dx.doi.org/10.2307/2004720}.
  
 
\bibitem{CEM}
D.~Calaque, K.~Ebrahimi-Fard, D.~Manchon.
\newblock { \em Two interacting {H}opf algebras of trees: a {H}opf-algebraic approach
  to composition and substitution of {B}-series.}
\newblock Adv. in Appl. Math. \textbf{47}, no.~2, (2011), 282--308.
\newblock 
  \burlalt{doi:10.1016/j.aam.2009.08.003}{http://dx.doi.org/10.1016/j.aam.2009.08.003}.  
  
  
 
 

 \bibitem{CH16}
A.~Chandra, M.~Hairer.
\newblock {\textsl{An analytic {BPHZ} theorem for regularity structures.}}
\newblock \burlalt{arXiv:1612.08138}{http://arxiv.org/abs/1612.08138}.





\bibitem{CHV05}
P.~{Chartier}, E.~{Hairer},  G.~{Vilmart}.
\newblock {\em A substitution law for B-series vector
	fields.}
\newblock INRIA Report, no. 5498, (2005).
\newblock
\burlalt{inria-00070509}{https://inria.hal.science/inria-00070509}.

\bibitem{CHV07}
P.~{Chartier}, E.~{Hairer}, G.~{Vilmart}.
\newblock \emph{Algebraic structures of {B}-series}.
\newblock Found. Comput. Math. \textbf{10}, no.~4, (2010), 407--427.
\newblock
\burlalt{doi:10.1007/s10208-010-9065-1}{http://dx.doi.org/10.1007/s10208-010-9065-1}.


\bibitem{CK1}
A.~Connes, D.~Kreimer.
\newblock { \em Hopf algebras, renormalization and noncommutative geometry.}
\newblock Comm. Math. Phys. \textbf{199}, no.~1, (1998), 203--242.
\newblock
  \burlalt{doi:10.1007/s002200050499}{http://dx.doi.org/10.1007/s002200050499}.

\bibitem{CK2}
A.~Connes, D.~Kreimer.
\newblock { \em Renormalization in quantum field theory and the {R}iemann-{H}ilbert
  problem {I}: the {H}opf algebra structure of graphs and the main theorem.}
\newblock  Comm. Math. Phys. \textbf{210}, (2000), 249--73.
\newblock
  \burlalt{doi:10.1007/s002200050779}{http://dx.doi.org/10.1007/s002200050779}.
%
 
  
  

  
 \bibitem{Faris}
 W.D.~Faris.
 \newblock {\em Combinatorial species and Feynman diagrams.}
 Semin. Lothar. Comb.
 \textbf{61A}, (2011), Article B61An, electronic only.
 \burlalt{http://www.kurims.kyoto-u.ac.jp/EMIS/journals/SLC/wpapers/s61Afaris.pdf}{http://www.kurims.kyoto-u.ac.jp/EMIS/journals/SLC/wpapers/s61Afaris.pdf}.
 
\bibitem{config2}
B.K. Fosdick, D.B. Larremore, J. Nishimura, J. Ugander.
\newblock { \em Configuring Random Graph Models with Fixed Degree Sequences.}
\newblock SIAM Review \textbf{60}, no.~2, (2018), 315--355.
\newblock
\burlalt{doi:10.1137/16M1087175}{http://dx.doi.org/10.1137/16M1087175}.

\bibitem{connect3}
A.L. Fetter, J.D. Walecka.
\newblock { \em Quantum Theory of Many-Particle Systems.}
\newblock Dover Publications (2003)
\newblock
\burlalt{ISBN: 978-0-486-42827-7}{https://store.doverpublications.com/products/9780486428277}.
 
  
  
  \bibitem{GMZ24}
  X. Gao, D. Manchon, Z. Zhu.
  \newblock {\em Free Novikov algebras and the Hopf algebra of decorated multi-indices}. 
  \newblock \burlalt{arXiv:2404.09899
  }{https://arxiv.org/abs/2404.09899}.
 
 
 \bibitem{connect2}
 J. Goldstone
 \newblock { \em Derivation of the Brueckner many-body theory.}
 \newblock Proceedings of the Royal Society of London. \textbf{239}, no.~1217, (1957), 267--279.
 \newblock
 \burlalt{doi:10.1098/rspa.1957.0037}{http://dx.doi.org/10.1098/rspa.1957.0037}.   
    
    
    
\bibitem{Guin1}
D. Guin,  J. M. Oudom, \emph{Sur l'alg\`ebre enveloppante d'une
  alg\`ebre pr\'{e}-{L}ie}, C. R. Math. Acad. Sci. Paris \textbf{340} (2005),
  no.~5, 331--336.
  \burlalt{doi:10.1016/j.crma.2005.01.010}{https://doi.org/10.1016/j.crma.2005.01.010}. 
  

\bibitem{Guin2}
D. Guin,  J. M. Oudom, \emph{On the {L}ie enveloping algebra of a pre-{L}ie algebra}, J.
  K-Theory \textbf{2} (2008), no.~1, 147--167.
  \burlalt{doi:10.1017/is008001011jkt037}{https://doi.org/10.1017/is008001011jkt037}. 
  

\bibitem{reg}
{\rm M. Hairer}.
\newblock {\em A theory of regularity structures.}
\newblock Invent. Math. \textbf{198}, no.~2, (2014), 269--504.
\newblock
  \burlalt{doi:10.1007/s00222-014-0505-4}{https://dx.doi.org/10.1007/s00222-014-0505-4}.
  
   \bibitem{Hairer18}
  {\rm M. Hairer}.
  \newblock {\em An Analyst’s Take on the BPHZ Theorem.}
  \newblock Comput. Combin. Dyn. Stoch. Control, pages 429-476, Cham, 2018.
  \newblock
  \burlalt{doi:10.1007/978-3-030-01593-0_16}{https://doi.org/10.1007/978-3-030-01593-0_16}.
  
  \bibitem{Hepp}
  K.~Hepp.
  \newblock \emph{On the equivalence of additive and analytic renormalization}.
  \newblock Comm. Math. Phys. \textbf{14}, (1969), 67--69.
  \newblock
  \burlalt{doi:10.1007/BF01645456}{http://dx.doi.org/10.1007/BF01645456}.
  

\bibitem{K2006}
D.~Kreimer
\newblock \emph{Anatomy of a gauge theory}.
\newblock Ann. Phys. \textbf{321}, (2006), 2757--2781.
\newblock
\burlalt{doi:10.1016/j.aop.2006.01.004}{https://doi.org/10.1016/j.aop.2006.01.004}.




 



   \bibitem{L23}
  P.~Linares.
  \newblock {\em Insertion pre-Lie products and translation of rough paths based on multi-indices}. 
  \newblock \burlalt{arXiv:2307.06769
  }{https://arxiv.org/abs/2307.06769}.
  
 
 \bibitem{LOT}
P.~Linares, F.~Otto, M.~Tempelmayr.
\newblock {\em The structure group for quasi-linear equations via universal enveloping algebras}. Comm. Amer. Math.  \textbf{3}, (2023), 1--64.  
 \burlalt{doi:10.1090/cams/16}{https://dx.doi.org/10.1090/cams/16}.
 

 \bibitem{LO23}
 P.~Linares F.~Otto.
 \newblock {\em A tree-free approach to regularity structures: The regular case for quasi-linear equations}. 
 \newblock \burlalt{arXiv:2207.10627 
 }{https://arxiv.org/abs/2207.10627}.
 
 
 
 
 \bibitem{LOTT}
 {\rm P.~Linares, F.~Otto, M.~Tempelmayr, P.~Tsatsoulis}
 \newblock {\em A diagram-free approach to the stochastic estimates in regularity structures.} Invent. Math. \textbf{237}, (2024), 1469--1565.
 \burlalt{doi:10.1007/s00222-024-01275-z}{https://dx.doi.org/10.1007/s00222-024-01275-z}.





\bibitem{OSSW}
 F.~Otto, J.~Sauer, S.~Smith, H.~Weber.
\newblock {\em A priori bounds for quasi-linear SPDEs in the full sub-critical regime}. J. Eur. Math. Soc. (2024).
\burlalt{doi:10.4171/JEMS/1574}{http://dx.doi.org/10.4171/JEMS/1574}.



\bibitem{VS2007}
W. D.~van Suijlekom.
\newblock {\em Renormalization of Gauge Fields: A Hopf Algebra Approach.} Comm. Math. Phys. \textbf{276}, (2007), 773--798.
\burlalt{doi:10.1007/s00220-007-0353-9}{https://dx.doi.org/10.1007/s00220-007-0353-9}.

\bibitem{config3}
P.~van der Hoorn, M.~Olvera-Cravioto.
\newblock {\em Typical distances in the directed configuration model}.
\newblock Ann. Appl. Probab., \textbf{28}, no.~3, (2018), 1739--1792 .
\newblock
\burlalt{doi:110.1214/17-AAP1342}{http://dx.doi.org/10.1214/17-AAP1342}.





\bibitem{WB08}
A.~Wurm, M.~Berg
\newblock {\em Wick calculus}. American Journal of Physics vol.~76, no.~1, (2008), 65--72
\newblock \burlalt{doi:10.1119/1.2805232}{https://doi.org/10.1119/1.2805232}. 




 \bibitem{Zim}
W.~Zimmermann.
\newblock {\em Convergence of {B}ogoliubov's method of renormalization in momentum
	space.}
\newblock Comm. Math. Phys. \textbf{15}, (1969), 208--234.
\newblock
\burlalt{doi:10.1007/BF01645676}{http://dx.doi.org/10.1007/BF01645676}.



\end{thebibliography}
\end{document}